\documentclass[12pt,a4paper,twoside]{amsart} 
\usepackage{latexsym}
\usepackage{amsmath}
\usepackage{amssymb}
\usepackage{amsfonts}
\usepackage{ifthen}
\usepackage{mathrsfs}               
\usepackage{bbm}                    
\usepackage{bbold}                  
\usepackage{textcomp}               
\usepackage{amstext}
\usepackage{enumerate}
\usepackage{amstext}
\newcommand{\CC}{\mathbb{C}}

\newcommand{\NN}{\mathbb{N}}

\newcommand{\RR}{\mathbb{R}}

\newcommand{\ZZ}{\mathbb{Z}}
\newcommand{\supp}{\mathrm{supp}}

\newcommand{\dist}{\mathrm{dist}}
\newcommand{\Ran}{\mathrm{Ran}}

\newcommand{\el}{\mathrm{el}}

\newcommand{\id}{\mathbbm{1}}
\newcommand{\klg}{\leqslant} 
\newcommand{\grg}{\geqslant}          
\newcommand{\ve}{\varepsilon}
\newcommand{\vp}{\varphi}
\newcommand{\vk}{\varkappa}
\newcommand{\vr}{\varrho}
\newcommand{\vt}{\vartheta}
\newcommand{\vs}{\varsigma}
\newcommand{\vo}{\varpi}
\newcommand{\wt}[1]{\widetilde{#1}}
\newcommand{\SL}{\langle \,}                          
\newcommand{\SR}{\, \rangle}    
\newcommand{\SPn}[2]{\langle \,#1\,|\,#2\, \rangle} 
\newcommand{\SPb}[2]{\big\langle \,#1\,\big|\,#2\, \big\rangle} 
\newcommand{\SPB}[2]{\Big\langle \,#1\,\Big|\,#2\, \Big\rangle}
\newcommand{\ol}[1]{\overline{#1}} 
\newcommand{\ul}[1]{\underline{#1}}
\newcommand{\wh}[1]{\widehat{#1}}  
\newcommand{\bigO}{\mathcal{O}}    
\newcommand{\fourier}{\mathcal{F}} 
\newcommand{\V}[1]{\mathbf{#1}}
\newcommand{\valpha}{\boldsymbol{\alpha}}
\newcommand{\vxi}{\boldsymbol{\xi}}

\newcommand{\vsigma}{\boldsymbol{\sigma}}

\newcommand{\veps}{\boldsymbol{\varepsilon}}
\newcommand{\vmu}{\boldsymbol{\mu}}
\newcommand{\vnu}{\boldsymbol{\nu}}
\newcommand{\vkap}{\boldsymbol{\varkappa}}
\newcommand{\LO}{\mathscr{L}}      

\newcommand{\HP}{\mathscr{K}}

\newcommand{\Fock}{\mathscr{F}_{\mathrm{b}}}
\newcommand{\spec}{\mathrm{\sigma}}

\newcommand{\FD}{\widehat{D}}
\newcommand{\PAm}{P^-_\mathbf{A}}               
\newcommand{\PA}{P^+_{\mathbf{A}}}

\newcommand{\PApmm}{P^{\pm}_{\mathbf{A}_m}}

\newcommand{\PApm}{P^{\pm}_{\mathbf{A}}}

\newcommand{\PO}{P^+_{\mathbf{0}}}

\newcommand{\SA}{S_{\mathbf{A}}}  
\newcommand{\SO}{S_{\mathbf{0}}}
\newcommand{\SAm}{S_{\mathbf{A}_{m}}}
\newcommand{\SAmt}{S_{\widetilde{\mathbf{A}}_m}}

\newcommand{\SAt}{S_{\widetilde{\mathbf{A}}}}
\newcommand{\DA}{D_{\mathbf{A}}}                
\newcommand{\DAm}{D_{\mathbf{A}_m}} 
\newcommand{\DAmt}{D_{\widetilde{\mathbf{A}}_m}}

\newcommand{\DAt}{D_{\widetilde{\mathbf{A}}}} 
\newcommand{\DO}{D_{\mathbf{0}}}
\newcommand{\D}[1]{D_{#1}}
\newcommand{\RA}[1]{R_{\mathbf{A}}(#1)}
\newcommand{\RAm}[1]{R_{\mathbf{A}_m}(#1)}
\newcommand{\RAF}[1]{R_{\mathbf{A}}^F(#1)}
\newcommand{\RAt}[1]{R_{\widetilde{\mathbf{A}}}(#1)}
\newcommand{\RO}[1]{R_{\mathbf{0}}(#1)}
\newcommand{\Hf}{H_{\mathrm{f}}}                   
\newcommand{\HT}{\check{H}_{\mathrm{f}}}           
\newcommand{\Hft}{\widetilde{H}_{\mathrm{f}}}

\newcommand{\He}{H_{\mathrm{f},m}^\ve}

\newcommand{\Hed}{H_{\mathrm{f},m}^{\ve,d}}
\newcommand{\Hef}{H_{\mathrm{f},m}^{\ve,f}}
\newcommand{\NPneg}[1]{H^-_{#1}}
\newcommand{\NP}[1]{H^+_{#1}}

\newcommand{\FNPm}[1]{H_{#1,m}}
\newcommand{\FNPmo}[1]{H^{0}_{#1,m}}

\newcommand{\FNPme}[1]{H^{\ve}_{#1,m}}
\newcommand{\FNPmed}[1]{H^{\ve,d}_{#1,m}}
\newcommand{\FNPmj}[1]{H_{#1,m_j}}

\newcommand{\ad}{a^\dagger}                     

\newcommand{\pf}{\mathbf{p}_{\mathrm{f}}}

\newcommand{\ps}{{\vxi}_\star}

\newcommand{\FNP}[1]{H^{\mathrm{np}}_{#1}}

\newcommand{\PF}[1]{H^{\mathrm{PF}}_{#1}}

\newcommand{\HR}{\mathscr{H}}
\newcommand{\core}{\mathscr{D}}
\newcommand{\dom}{\mathcal{D}}
\newcommand{\form}{\mathcal{Q}}

\newcommand{\gc}{\gamma_{\mathrm{c}}} 
\newcommand{\gcnp}{\gamma_{\mathrm{c}}^{\mathrm{np}}}
\newcommand{\gcPF}{\gamma_{\mathrm{c}}^{\mathrm{PF}}}
\newcommand{\UV}{\Lambda}             
\newcommand{\Th}{\Sigma}              
\newcommand{\cA}{\mathcal{A}}
\newcommand{\cO}{\mathcal{O}}

\newcommand{\cR}{\mathcal{R}}
\newcommand{\cS}{\mathcal{S}}
\newcommand{\cH}{\mathcal{H}}

\newcommand{\cX}{\mathcal{X}}         
\newcommand{\cY}{\mathcal{Y}}       
\newcommand{\cZ}{\mathcal{Z}} 
\newcommand{\sC}{\mathscr{C}}
\newcommand{\sD}{\mathscr{D}}

\newcommand{\sH}{\mathscr{H}}

\newcommand{\sX}{\mathscr{X}}         
       
\newcommand{\sZ}{\mathscr{Z}} 
\newcommand{\fS}{\mathfrak{S}}

\renewcommand{\Im}{\mathrm{Im}\,}
\renewcommand{\Re}{\mathrm{Re}\,}
\newtheorem{theorem}{Theorem}[section]
\newtheorem{lemma}[theorem]{Lemma}
\newtheorem{proposition}[theorem]{Proposition}
\newtheorem{corollary}[theorem]{Corollary}
\newtheorem{hypothesis}[theorem]{Hypothesis}
\theoremstyle{remark}
\newtheorem{remark}[theorem]{Remark}
\newtheorem{example}[theorem]{Example}
\numberwithin{equation}{section}
\title[Ground states in relativistic QED]{
Existence of ground states of hydrogen-like atoms in
relativistic QED II:\\ The no-pair operator}
\author{Martin K\"onenberg}
\address{Martin K\"onenberg\\
Fakult\"at f\"ur Mathematik und Informatik\\
FernUniversit\"at Hagen\\
L\"utzowstra{\ss}e 125\\
D-58084 Hagen, Germany\\
{\em Present address:} Fa- kult\"at f\"ur Physik\\
Universit\"at Wien\\
Boltzmanngasse 5\\
1090 Vienna,
Austria.}
\email{martin.koenenberg@univie.ac.at}
\author{Oliver Matte}
\address{Oliver Matte
Institut f\"ur Mathematik\\
TU Clausthal\\
Erzstra{\ss}e 1\\
D-38678 Clausthal-Zellerfeld, Germany\\
{\em Present address:} Mathematisches Institut\\
Ludwig-Maximilians-Universit\"at\\
Theresienstra{\ss}e 39\\
D-80333 M\"unchen, Germany.}
\email{matte@math.lmu.de}
\author{Edgardo Stockmeyer}
\address{Edgardo Stockmeyer\\
Mathematisches Institut\\
Ludwig-Maximilians-Univer- sit\"at\\
Theresienstra{\ss}e 39\\
D-80333 M\"unchen, Germany.}
\email{stock@math.lmu.de}

\date{\today}
\begin{document}

\begin{abstract}
We consider a hydrogen-like atom in a quantized electromagnetic field
which is modeled by means of a no-pair operator acting in
the positive spectral subspace of the free Dirac operator
minimally coupled to the quantized vector potential.
We prove that the infimum of the spectrum
of the no-pair operator is an evenly degenerate eigenvalue.
In particular, we show that the bottom of its spectrum
is strictly less than its ionization threshold.
These results hold true, for arbitrary values of the fine-structure
constant and the ultra-violet cut-off and for all
Coulomb coupling constants less than the critical
one of the Brown-Ravenhall model, 
$2/(2/\pi+\pi/2)$.
For Coulomb coupling constants larger than the critical one,
we show that the quadratic form of the no-pair operator is unbounded below.
Along the way we discuss the domains and operator cores of the
semi-relativistic Pauli-Fierz
and no-pair operators, for Coulomb coupling constants
less than or equal to the critical ones.
\end{abstract}

\maketitle


\section{Introduction}
\label{sec-intro}

\noindent
In this article we continue our study of the existence of
ground states of hydrogen-like atoms and ions
in (semi-)relativistic models of quantum
electrodynamics (QED).
The model studied here is given by the following no-pair operator,
\begin{equation}\label{def-NP-intro}
\NP{\gamma}:=
\PA\,(\DA-\gamma/|\V{x}|+\Hf)\,\PA\,,
\end{equation}
where $\DA$ is the free Dirac operator minimally
coupled to the quantized vector potential,
$\V{A}$, and $\PA$ is the spectral projection
onto its positive spectral subspace,
$$
\PA:=\id_{[0,\infty)}(\DA)\,.
$$
Moreover, $\Hf$ is the
energy of the photon field  and $\gamma\grg0$ the Coulomb 
coupling constant. 
The quantized vector potential $\V{A}$ depends on two 
physical parameters, namely the fine structure constant, $e^2$,
and an ultra-violet cut-off parameter, $\UV$.
Thus, $\NP{\gamma}$ is acting in the projected Hilbert
space $\PA\,\HR$, where $\HR$ denotes the Hilbert tensor
product of $L^2(\RR^3_\V{x},\CC^4)$ and the bosonic Fock
space of the photon field.
The mathematical analysis of an analogue of $\NP{\gamma}$
describing a molecular system has been initiated 
in \cite{LiebLoss2002b,LiebLoss2002} where the stability
of matter of the second kind is established (under certain
restrictions on the nuclear charges and $e^2$ and $\UV$)
and an upper bound on the (positive) binding energy is given.
For more information on a general class of no-pair
operators with {\em classical} external electromagnetic fields
and on some applications of no-pair operators
in quantum chemistry and physics we refer to
\cite{MatteStockmeyer2008a} and the references therein.

Improving earlier results from \cite{MatteStockmeyer2009a} 
we show in the present article that the quadratic form
of $\NP{\gamma}$ is bounded below,
for arbitrary values of $e^2$ and $\UV$,
if and only if $\gamma\klg\gcnp$, where
\begin{equation}\label{def-gammac}
\gcnp:=2/(2/\pi+\pi/2)
\end{equation}
is the critical coupling
constant of the electronic Brown-Ravenhall operator \cite{EPS1996}.
In particular,
$\NP{\gamma}$ has a self-adjoint Friedrichs extension, provided that
$\gamma\klg\gcnp$.
The main result of the present paper states that,
for all $\gamma\in(0,\gcnp)$, 
the infimum of the spectrum of this Friedrichs extension
is a degenerate eigenvalue.
As an intermediate step we prove a binding condition
for $\NP{\gamma}$, $\gamma\in(0,\gcnp]$, ensuring that the ground state energy of
$\NP{\gamma}$ is strictly less than its ionization threshold.
Along the way we also study the 
domain and essential self-adjointness of $\NP{\gamma}$,
$\gamma\in[0,\gcnp]$, as well as of
the semi-relativistic Pauli-Fierz operator,
\begin{equation}\label{def-PF2}
\cH_\gamma:=\sqrt{(\vsigma\cdot(-i\nabla+\V{A}))^2+\id}\,
-\gamma/|\V{x}|+\Hf\,,\quad \gamma\in[0,\gcPF]\,,
\end{equation} 
where $\gcPF:=2/\pi<\gcnp$ is the critical constant
in Kato's inequality.
(The formal vector $\vsigma$ 
contains the Pauli spin matrices.)
The latter discussion improves on earlier results
of \cite{MatteStockmeyer2009a} and 
\cite{MiyaoSpohn2009}.

The existence of energy minimizing ground states
for atoms and molecules in {\em non}-relativistic QED, 
where the electrons are described
by Schr\"odinger operators,
is by now a well-established fact.
The first existence proofs have been given in
\cite{BFS1998b,BFS1999}, for small values
of $e^2$ and $\UV$. Later on
the existence of ground states for a molecular
Pauli-Fierz Hamiltonian has been shown in \cite{GLL2001},
for all values of $e^2$ and $\Lambda$,
assuming a certain 
binding condition, which has been verified in \cite{LiebLoss2002}.
In the last decade there appeared a large number of
further mathematical contributions to non-relativistic
QED. Here we only mention that ground state energies
and projections have also been studied by means of
infra-red finite algorithms and renormalization group
methods \cite{BCFS2003,BFP2006,BFS1998b,BFS1998a,BFS1999,BachKoenenberg2006,FGS2008}.
These sophisticated methods give very detailed results as they
rely on constructive algorithms rather than compactness
arguments. They work, however, only in a regime
where $e^2$ and/or $\UV$ are sufficiently small.

In our earlier companion paper \cite{KMS2009a}
we have already shown that $\cH_\gamma$ has a degenerate
ground state eigenvalue, for all $\gamma\in(0,\gcPF)$.
The existence of ground states
in a relativistic atomic model from QED where also the electrons and positrons
are treated as quantized fields 
is investigated in
\cite{BDG2004}. To this end infra-red regularizations are imposed
in the interaction terms of the Hamiltonian which is not necessary
in the model treated here.

We like to stress one essential feature which both operators
$\cH_\gamma$ and $\NP{\gamma}$ have in common:
namely their gauge invariance.
In fact, the possibility to pass to a suitable gauge
by means of a unitary Pauli-Fierz transformation
allows to prove two infra-red estimates 
serving as key ingredients in a certain compactness argument
introduced in \cite{GLL2001}.
We remark that when the projections $\PA$ in \eqref{def-NP-intro} 
are replaced by projections that do not contain
the vector potential, that is, by $\PO$, then the resulting 
operator is not gauge invariant anymore.
In this case one can still prove the existence of 
ground states 
provided that a mild infra-red regularization is imposed
on $\V{A}$ \cite{Koenenberg2004,Matte2000}. 
It seems, however, unlikely
that the infra-red regularization can be dropped
when $\PO$ is used instead of $\PA$
\cite{Koenenberg2004}. It is also known that a no-pair
model defined by means of $\PO$ becomes unstable
as soon as more than one electron is considered
\cite{GriesemerTix1999}.

Although in many parts the general strategy of our proofs
in \cite{KMS2009a} and the present paper
follows along the
lines of \cite{BFS1999} and \cite{GLL2001}
the analysis of the operators $\cH_\gamma$ and $\NP{\gamma}$
poses a variety of new and non-trivial mathematical problems. 
This is mainly caused by the non-locality 
of $\cH_\gamma$ and $\NP{\gamma}$
which both do not act as partial differential operators
on the electronic degrees of freedom anymore as it is the case
in non-relativistic QED.
In this respect $\NP{\gamma}$ is harder to
analyze than $\cH_\gamma$
since also the Coulomb potential and the radiation field
energy become non-local due to the presence of the
spectral projections $\PA$.

There is one question left open in \cite{KMS2009a} and the present paper,
namely whether $\cH_\gamma$ and $\NP{\gamma}$ still possess
ground state eigenvalues when $\gamma$ attains the respective
critical values. Instead of going through all proofs in
\cite{KMS2009a} and below and trying to adapt them 
to cover also the critical cases, it seems to be more convenient 
to approximate the ground state eigenvectors in the critical
cases by those found for sub-critical values of $\gamma$.
This argument requires an estimate on the spatial
localization of the ground state eigenvectors which is
uniform in $\gamma$. Earlier results \cite{MatteStockmeyer2009a}
provide, however, only $\gamma$-dependent estimates.
As a new derivation of a uniform bound 
would lengthen the present article too much
we shall work out these ideas elsewhere.

\smallskip

\noindent
{\em The organization of this article and brief remarks on some techniques.}
In Section~\ref{sec-model} we introduce the no-pair operator
and state our main results more precisely.
In Section~\ref{sec-esa} we derive some basic relative
bounds involving $\cH_\gamma$ and $\NP{\gamma}$ which
improve on earlier results of \cite{MatteStockmeyer2009a}.
Here we benefit from recent generalized Hardy-type inequalities
\cite{Frank2009,SoSoeSp2008} that allow to derive these relative
bounds also for critical values of $\gamma$.
Moreover, we discuss the domains and the essential
self-adjointness of $\cH_\gamma$ and $\NP{\gamma}$.
Section~\ref{sec-conv} is the core of this article as it
discusses the convergence of sequences of no-pair operators.
The results are new and tailor-made for the no-pair model.
They allow to implement some well-known arguments
developed to prove the existence of ground states in
non-relativistic QED \cite{BFS1999,GLL2001} in the present
setting.
In Section~\ref{sec-binding} we derive a binding condition
for the no-pair operator which is necessary in order to
apply the results of Section~\ref{sec-conv}.
In Section~\ref{sec-NPm} we prove the existence of 
ground state eigenvectors, $\phi_m$, assuming that the
the photons have a mass $m>0$. We employ a discretization
argument and proceed along the lines of \cite{BFS1999},
where $e^2$ and/or $\UV$ are assumed to be small.
The implementation of the discretization procedure
requires, however, our new results of Section~\ref{sec-conv}.
Moreover, as in \cite{KMS2009a}
we add a new observation which allows
to treat also large values of $e^2$ and $\UV$.
By means of a compactness argument very similar
to the one introduced in \cite{GLL2001} we then infer
the existence of ground states for
$\NP{\gamma}$ ($m=0$) in Section~\ref{sec-ex}.
This compactness argument makes use of some further non-trivial
ingredients: First, we need a bound on the spatial
localization of $\phi_m$, uniformly in $m>0$.
Such a bound has already been derived in \cite{MatteStockmeyer2009a}.
Second, we need two infra-red bounds \cite{BFS1999,GLL2001}
whose proofs are deferred to Section~\ref{sec-IR-bounds}.
Technically, our proof of the infra-red bounds
differs slightly from those in \cite{BFS1999,GLL2001}
as we first derive a representation formula
for $a(k)\,\phi_m$. The infra-red bounds are then
easily read off from this formula.
The main text is followed by two appendices
where some technical estimates on functions
of the Dirac operator are given (Appendix~\ref{app-conv})
and some properties of $\phi_m$ are discussed 
(Appendix~\ref{app-form}).

\smallskip

\noindent
{\em Frequently used notation.}
$\dom(T)$ denotes the domain of an operator $T$
in some Hilbert space and $\form(T)$ its form domain,
provided that $T=T^*>-\infty$.
If $T$ is self-adjoint, then $\RR\ni\lambda\mapsto\id_\lambda(T)$
denotes its spectral family and $\id_M(T)$ the spectral projection
corresponding to some measurable subset $M\subset\RR$.
$C(a,b,\ldots),C'(a,b,\ldots)$ etc. denote positive constants
which only depend on the quantities $a,b,\ldots$ displayed
in their arguments. Their values might change from one estimate
to another.
 

\section{Definition of the model and main results}
\label{sec-model}

\noindent
First, we recall some standard notation.
The Hilbert space underlying the atomic model studied
in this article is a subspace of
$\HR:=\HR_0$, where
\begin{equation}\label{def-HR}
\HR_m:=\,L^2(\RR^3_\V{x},\CC^4)\otimes\Fock[\HP_m]
=\int_{\RR^3}^\oplus
\CC^4\otimes\Fock[\HP_m]\,d^3\V{x}\,,\quad m\grg0\,.
\end{equation}
(In some of our proofs we choose $m>0$.)
Here the bosonic Fock space, 
$\Fock[\HP_m]=\bigoplus_{n=0}^\infty\Fock^{(n)}[\HP_m]$,
is modeled over the one photon Hilbert space 
\begin{equation*}
\HP_m:=L^2(\cA_m\times\ZZ_2,dk)\,,\quad
\int dk:=\sum_{\lambda\in\ZZ_2}\int_{\cA_m}\!\!d^3\V{k}\,,
\quad\cA_m:=\{|\V{k}|\grg m\}\,.
\end{equation*}
The letter
$k=(\V{k},\lambda)$ always
denotes a tuple consisting of a photon wave vector,
$\V{k}\in\RR^3$, and a polarization label, $\lambda\in\ZZ_2$.
The components of $\V{k}$ are denoted as
$\V{k}=(k^{(1)},k^{(2)},k^{(3)})$. We recall that
$\Fock^{(0)}[\HP_m]:=\CC$ and, for $n\in\NN$, 
$\Fock^{(n)}[\HP_m]:=\cS_n L^2((\cA_m\times\ZZ_2)^n)$,
where, for $\psi^{(n)}\in L^2((\cA_m\times\ZZ_2)^n)$,
$$
(\cS_n\,\psi^{(n)})(k_1,\ldots,k_n):=
\frac{1}{n!}\sum_{\pi\in\fS_n}\psi^{(n)}(k_{\pi(1)},\ldots,k_{\pi(n)})\,,
$$
$\fS_n$ denoting the group of permutations of $\{1,\ldots,n\}$.
For $f\in\HP_m$ and $n\in\NN_0$, we further define 
$\ad(f)^{(n)}:\Fock^{(n)}[\HP_m]\to\Fock^{(n+1)}[\HP_m]$
by $\ad(f)^{(n)}\,\psi^{(n)}:=\sqrt{n+1}\,\cS_{n+1}(f\otimes\psi^{(n)})$.
Then $\ad(f):=\bigoplus_{n=0}^\infty\ad(f)^{(n)}$
and $a(f):=\ad(f)^*$ are the standard bosonic creation and
annihilation operators satisfying
the canonical commutation relations
\begin{equation}\label{CCR}
[a^\sharp(f)\,,\,a^\sharp(g)]=0\,,\qquad
[a(f)\,,\,\ad(g)]=\SPn{f}{g}\,\id\,,\qquad
f,g\in\HP_m\,,
\end{equation}
where $a^\sharp$ is $\ad$ or $a$.
For a three-vector of functions 
$\V{f}=(f^{(1)},f^{(2)},f^{(3)})\in\HP^3_m$, we write
$a^\sharp(\V{f}):=(a^\sharp(f^{(1)}),a^\sharp(f^{(2)}),a^\sharp(f^{(3)}))$.
Then the quantized vector potential
associated to some 
measurable family $\V{G}_\V{x}\in\HP^3_m$, $\V{x}\in\RR^3$,
is the triple of operators
$\V{A}=(A^{(1)},A^{(2)},A^{(3)})$ given as
\begin{equation}\label{def-Aphys}
\V{A}\equiv\V{A}[\V{G}]:=\int^\oplus_{\RR^3}
\id_{\CC^4}\otimes\V{A}(\V{x})\,d^3\V{x}\,,
\qquad
\V{A}(\V{x}):=\ad(\V{G}_\V{x})+a(\V{G}_\V{x})\,.
\end{equation}
Next, we recall that the second quantization, $d\Gamma(\vo)$,
of some Borel function $\vo:\cA_m\times\ZZ_2\to\RR$
is the direct sum
$d\Gamma(\vo):=\bigoplus_{n=0}^\infty d\Gamma^{(n)}(\vo)$,
where $d\Gamma^{(0)}(\vo):=0$, and, for $n\in\NN$,
$d\Gamma^{(n)}(\vo)$ is
the maximal multiplication operator in $\Fock^{(n)}[\HP_m]$
associated with the symmetric function 
$(k_1,\ldots,k_n)\mapsto\vo(k_1)+\dots+\vo(k_n)$.

Our main results deal with the physical choices
of $\vo$ and $\V{G}_\V{x}$ given in Example~\ref{ex-Gphys}.
Since many results of the technical parts of this paper
are applied to modified versions of these physical choices
it is, however, convenient to introduce the following more general
hypothesis:

\begin{hypothesis}\label{hyp-G}
$\vo:\cA_m\to[0,\infty)$ is a measurable function 
such that $0<\vo(k):=\vo(\V{k})$, 
for almost every $k=(\V{k},\lambda)\in\cA_m\times\ZZ_2$.
For almost every 
$k\in\cA_m\times\ZZ_2$, 
$\V{G}(k)$
is a bounded, twice continuously differentiable function,
$\RR^3\ni\V{x}\mapsto\V{G}_{\V{x}}(k)\in\CC^3$,
such that the map $(\V{x},k)\mapsto\V{G}_\V{x}(k)$
is measurable. 
There is some $d\in(0,\infty)$ such that, for $\ell\in\{-1,0,1,\ldots,7\}$,
\begin{align}\label{def-d3}
\int\vo(k)^{\ell}\,\|\V{G}(k)\|^2_\infty\,dk
\,&\klg\,d\,,\quad 
\int\vo(k)^{-1}\,\|\nabla_{\V{x}}\wedge\V{G}(k)\|^2_\infty\,dk
\,\klg\,d\,,
\end{align}
where
$\|\V{G}(k)\|_\infty:=\sup_{\V{x}}|\V{G}_{\V{x}}(k)|$, etc.
\end{hypothesis}

\begin{example}\label{ex-Gphys}
In the physical model we are interested in
we have $m=0$ and
the radiation field energy, $\Hf$, is given by
\begin{equation}\label{def-Hf}
\Hf:=d\Gamma(\omega)\,,\qquad
\omega(k):=|\V{k}|\,,\;k=(\V{k},\lambda)\in\RR^3\times\ZZ_2\,.
\end{equation}
A physically interesting choice for $\V{G}_\V{x}$ is given as follows:
Writing
\begin{equation}\label{def-kbot}
\V{k}_\bot\,:=\,(k^{(2)}\,,\,-k^{(1)}\,,\,0)\,,\qquad
\V{k}=(k^{(1)},k^{(2)},k^{(3)})\in\RR^3,
\end{equation}
we introduce the following polarization vectors,
\begin{equation}\label{pol-vec}
\veps(\V{k},0)\,=\,
\frac{\V{k}_\bot}{|\V{k}_\bot|}
\,,\qquad
\veps(\V{k},1)\,=\,
\frac{\V{k}}{|\V{k}|}\,\wedge\,\veps(\V{k},0)\,,
\end{equation}
for almost every $\V{k}\in\RR^3$, and set
\begin{equation}\label{def-Gphys}
\V{G}_\V{x}^{e,\UV}(k)
:=-e\,\frac{\id_{\{|\V{k}|\klg\UV\}}}{2\pi\sqrt{|\V{k}|}}
\,e^{-i\V{k}\cdot\V{x}}\,\veps(k)\,,
\end{equation}
for all $\V{x}\in\RR^3$
and almost every $k=(\V{k},\lambda)\in\RR^3\times\ZZ_2$.
Here $\UV>0$ is an ultra-violet cut-off parameter whose
value can be chosen arbitrarily large. The value of 
$e\in\RR$ does not affect the validity of our results either.
(In nature we have $e^2\approx1/137$. For
in the units chosen above --
energies are measured in units of the rest energy of the electron and
$\V{x}$ is measured in units of one Compton wave length
divided by $2\pi$ -- the square of the elementary charge $e>0$
is equal to Sommerfeld's fine-structure constant.)
\hfill$\Diamond$
\end{example}

\smallskip

\noindent
Finally, we recall the definition of the Dirac operator, $\DA$,  
minimally coupled to $\V{A}$.
Let $\alpha_1,\alpha_2,\alpha_3$, and $\beta=\alpha_0$
denote the hermitian 4\texttimes4 Dirac matrices
obeying the Clifford algebra relations
 \begin{equation}\label{Clifford}
\alpha_i\,\alpha_j+\alpha_j\,\alpha_i\,=\,2\,\delta_{ij}\,\id\,,
\qquad i,j\in\{0,1,2,3\}\,.
\end{equation}
They act on the second tensor factor in
$\HR_m=L^2(\RR^3_\V{x})\otimes\CC^4\otimes\Fock[\HP_m]$ and,
in the standard representation, they are given in terms of
the Pauli matrices,
\begin{align*}
\sigma_1\,&=\,
\begin{pmatrix}
0&1\\1&0
\end{pmatrix}\,,\quad
\sigma_2\,=\,
\begin{pmatrix}
0&-i\\i&0
\end{pmatrix}\,,\quad
\sigma_3\,=\,
\begin{pmatrix}
1&0\\0&-1
\end{pmatrix}\,,
\end{align*}
as 
$$
\alpha_j=
\begin{pmatrix}0&\sigma_j\\\sigma_j&0
\end{pmatrix}\,,\quad
j\in\{1,2,3\}\,, \qquad
\beta=
\begin{pmatrix}\id_2&0\\0&-\id_2
\end{pmatrix}\,.
$$
The free Dirac operator minimally coupled to $\V{A}$ is now given as
\begin{equation}\label{def-DA}
\DA:=\valpha\cdot(-i\nabla_\V{x}+\V{A})+\beta
:=
\sum_{j=1}^3\alpha_j\,(-i\partial_{x_j}+A^{(j)})\,+\,\beta
\,.
\end{equation}
Under the assumptions on $\V{G}_\V{x}$
given in Hypothesis~\ref{hyp-G} it is clear that
$\DA$ is well-defined a priori on the dense domain
\begin{equation}\label{def-DmCm}
\core_m:=C_0^\infty(\RR^3,\CC^4)\otimes\sC_m\,.
\quad\textrm{(Algebraic tensor product.)}
\end{equation}
Here, $\sC_m\subset\Fock[\HP_m]$ denotes the subspace of
all $(\psi^{(n)})_{n=0}^\infty\in\Fock[\HP_m]$ such that
only finitely many components $\psi^{(n)}$ are non-zero
and such that each $\psi^{(n)}$, $n\in\NN$, is essentially
bounded with compact support.
Moreover,
a straightforward application of Nelson's commutator theorem
with test operator $-\Delta_\V{x}+d\Gamma(\vo)+1$ 
(see, e.g., \cite{LiebLoss2002})
reveals that
$\DA$ is essentially self-adjoint on $\sD_m$, for all
$\V{G}_\V{x}$ fulfilling Hypothesis~\ref{hyp-G}.
We again use the symbol
$\DA$ to denote its closure starting from $\sD_m$.
Then the spectrum of $\DA$ is contained
in the union of two half-lines,
$
\spec(\DA)\subset(-\infty,-1]\cup[1,\infty)
$,
and we denote the orthogonal projections
onto the corresponding positive and negative spectral subspaces by
\begin{equation}\label{def-PA}
\PApm:=\id_{\RR^\pm}(\DA)=
\frac{1}{2}\,\id\pm\frac{1}{2}\,\SA\,,\qquad
\SA:=\DA\,|\DA|^{-1}.
\end{equation}
For later reference we recall that the sign function, $\SA$, 
of $\DA$ can be represented in terms of the resolvent
\begin{equation}\label{def-RA}
\RA{iy}:=(\DA-iy)^{-1},\qquad y\in\RR\,,
\end{equation}
as a strongly convergent
principal value \cite[Lemma~VI.5.6]{Kato},
\begin{equation}\label{sgn}
\SA\,\vp
=\lim_{\tau\to\infty}
\int_{-\tau}^\tau\RA{iy}\,\varphi\,\frac{dy}{\pi}
=\lim_{\tau\to\infty}
\int_{-\tau}^\tau\RA{-iy}\,\varphi\,\frac{dy}{\pi}
\,, \quad 
\varphi\in\HR_m\,.
\end{equation}
The no-pair operator studied in this paper
is a self-adjoint
operator acting in the positive spectral subspace $\PA\HR_m$.
It is defined a priori on the dense domain
$\PA\,\sD_m\subset\PA\HR_m$ by
\begin{equation}\label{def-NPhypG}
\NP{\gamma,\vo,\V{G}}:=
\PA\,(\DA-\gamma/|\V{x}|+d\Gamma(\vo))\,\PA\,.
\end{equation}
Thanks to \cite[Proof of Lemma~3.4(ii)]{MatteStockmeyer2009a},
which implies that $\PA$ maps the subspace
$\dom(\DO)\cap\dom(d\Gamma(\vo))$ into itself,
and Hardy's inequality,
we actually know that $\NP{\gamma,\vo,\V{G}}$ is well-defined on $\sD_m$.
We recall the definition \eqref{def-gammac} and state our first result
which improves on \cite[Theorem~2.1]{MatteStockmeyer2009a},
where the semi-boundedness of $\NP{\gamma,\vo,\V{G}}$ has been
shown, for sub-critical values of $\gamma$. Its new proof
is independent of \cite[Theorem~2.1]{MatteStockmeyer2009a}.

\begin{proposition}\label{prop-stab}
Assume that $\vo$ and $\V{G}$ fulfill Hypothesis~\ref{hyp-G}.
Then the quadratic form of $\NP{\gamma,\vo,\V{G}}$ is bounded
below on $\PA(\dom(\DO)\cap\dom(d\Gamma(\vo)))$, 
if and only if $\gamma\klg\gcnp$.
\end{proposition}

\begin{proof}
This proposition is proved at the end of Section~\ref{sec-esa}.
\end{proof}

\smallskip

\noindent
In particular, $\NP{\gamma,\vo,\V{G}}$ has a self-adjoint
Friedrichs extension provided that $\gamma\in[0,\gcnp]$.
In the rest of this section we denote this extension again by the
same symbol $\NP{\gamma,\vo,\V{G}}$. The next
theorem gives a bound on the binding energy, i.e. the gap between the 
ground state energy and the ionization threshold
of $\NP{\gamma,\vo,\V{G}}$ defined, respectively, as
\begin{equation}\label{def-ETh}
E_{\gamma}(\vo,\V{G}):=\inf\spec(\NP{\gamma,\vo,\V{G}})\,,
\;\,\gamma\in(0,\gcnp]\,,
\quad \Th(\vo,\V{G}):=\inf\spec(\NP{0,\vo,\V{G}})\,.
\end{equation}

\begin{theorem}[{\bf Binding}]\label{thm-binding-NP}
Assume that $\vo$ and $\V{G}$ fulfill Hypothesis~\ref{hyp-G}
and that $\V{G}_\V{x}(k)=e^{-i\vmu(k)\cdot\V{x}}\,\V{g}(k)$,
for all $\V{x}\in\RR^3$ and almost every $k\in\cA_m\times\ZZ_2$,
where $\vmu,\V{g}:\cA_m\times\ZZ_2\to\RR^3$ are measurable
such that $|\vmu|\klg\vo$ almost everywhere.
Let $\gamma\in(0,\gcnp]$.
Then there is some $c\in(0,\infty)$, depending only on
$\gamma$ and the parameter $d$, such that
$
\Th(\vo,\V{G})-E_\gamma(\vo,\V{G})\grg c
$.
\end{theorem}

\begin{proof}
This theorem is proved in Subsection~\ref{ssec-binding-proof}.
\end{proof}

\smallskip

\noindent
Next, we state the main result of this article
dealing with the physical choices $m=0$,
$\vo=\omega$, and 
$\V{G}_\V{x}=\V{G}^{e,\UV}_\V{x}$ as given in Example~\ref{ex-Gphys}.
In this case we abbreviate
$\NP{\gamma}:=\NP{\gamma,\omega,\V{G}^{e,\UV}}$
and $E_\gamma:=E_\gamma(\omega,\V{G}^{e,\UV})$.

\begin{theorem}[{\bf Existence and non-uniqueness of ground states}]
\label{thm-ex-NP}
For\\ $e\in\RR$, $\UV>0$, and $\gamma\in(0,\gcnp)$,
$E_{\gamma}$ is an evenly
degenerated eigenvalue of $\NP{\gamma}$.
\end{theorem}

\begin{proof}
The fact that $E_{\gamma}$ is an eigenvalue
is proved in Section~\ref{sec-ex}.
In the following we apply Kramers' theorem
to show that $E_\gamma$ is evenly degenerated.
Similarly
as in \cite{MiyaoSpohn2009}, where the same observation is
made for eigenvalues of the semi-relativistic Pauli-Fierz operator, 
we introduce the anti-linear operator 
$$
\vt\,:=\,J\,\alpha_2\,C\,R
\,=\,
-\alpha_2\,J\,C\,R
\,,\qquad 
J\,:=\,\begin{pmatrix}0&\id_2\\-\id_2&0\end{pmatrix}\,,
$$
where $C:\sH\to\sH$ denotes complex conjugation,
$C\,\psi:=\ol{\psi}$, $\psi\in\sH$, and 
$R:\sH\to\sH$ is the parity transformation
$(R\,\psi)(\V{x}):=\psi(-\V{x})$, for almost every $\V{x}\in\RR^3$
and every $\psi\in\sH=L^2(\RR^3_\V{x},\CC^4\otimes\Fock[\HP_0])$.
Obviously, 
$[\vt\,,\,-i\partial_{x_j}]=[\vt\,,1/|\V{x}|\,]=[\vt\,,\,\Hf]=0$, on 
$\dom(\DO)\cap\dom(\Hf)$.
Since $\alpha_2$ squares to one and $C\,\alpha_2=-\alpha_2\,C$,
as all entries of $\alpha_2$ are purely imaginary,
we further get $\vt^2=-\id$ and $[\vt\,,\,\alpha_2]=0$. Moreover,
the Dirac matrices $\alpha_0$, $\alpha_1$, and $\alpha_3$ have real 
entries  and $[J\,\alpha_2\,,\,\alpha_j]=J\,\{\alpha_2,\alpha_j\}=0$
by \eqref{Clifford}, whence $[\vt\,,\,\alpha_j]=0$, for $j\in\{0,1,3\}$. 
Finally, $[\vt\,,\,e^{\pm i\V{k}\cdot\V{x}}]=0$ implies
$[\vt\,,\,A^{(j)}]=0$ on $\dom(\Hf^{1/2})$, for $j\in\{1,2,3\}$.
It follows that $[\vt\,,\,\DA]=0$ on $\core_0=\vt\,\core_0$ and, since
$\DA$ is essentially self-adjoint on $\core_0$, we obtain
$\vt\,\dom(\DA)=\dom(\DA)$ and $[\vt\,,\,\DA]=0$ on $\dom(\DA)$,
which implies $\vt\,\RA{iy}-\RA{-iy}\,\vt=0$ on $\HR$, for every $y\in\RR$.
Using the representation \eqref{sgn} we conclude that
$[\vt\,,\,\PA]=0$ on $\HR$. 
In particular, $\vt$ can be considered as operator acting
on $\PA\HR$.
Furthermore, we obtain
$\NP{\gamma}\,\vt-\vt\,\NP{\gamma}=0$ on $\core_0$.
Hence, the quadratic
forms of $\NP{\gamma}$ and $-\vt\,\NP{\gamma}\,\vt$ coincide
on $\core_0$,
which is a form core for $\NP{\gamma}$.
This readily implies $\vt\,\dom(\NP{\gamma})=\dom(\NP{\gamma})$
and $[\NP{\gamma}\,,\,\vt]=0$.
On account of $\vt^2=-\id$ and the formula
\begin{equation}\label{claire2}
\SPn{\vt\,\vp}{\vt\,\psi}\,=\,\SPn{\psi}{\vp}\,,\quad
\SPn{\vt\,\vp}{\psi}\,=\,-\SPn{\vt\,\psi}{\vp}\,,
\qquad\vp,\psi\in\sH, 
\end{equation}
Kramers' degeneracy theorem now shows that every eigenvalue 
of $\NP{\gamma}$ is evenly degenerated.
(In fact, $\NP{\gamma}\,\phi=E_\gamma\,\phi$ implies
$\NP{\gamma}\,\vt\,\phi=E_\gamma\,\vt\,\phi$, and $\phi\bot\vt\,\phi$
since $\SPn{\vt\,\phi}{\phi}=-\SPn{\vt\,\phi}{\phi}$ by \eqref{claire2}.)
\end{proof}

\begin{remark}
Every ground state eigenvector of $\NP{\gamma}$
is exponentially localized in the $L^2$-sense 
with respect to the electron coordinates \cite{MatteStockmeyer2009a};
see \eqref{eq-exp-loc} below. 
\end{remark}


\section{Relative bounds and essential self-adjointness}
\label{sec-esa}

\noindent
The aim of this section is to discuss the domains
and essential self-adjointness of the no-pair operators
defined by \eqref{def-PA} and \eqref{def-NPhypG}
and to provide some basic relative bounds.
It is actually more convenient 
from a technical point of view to extend $\NP{\gamma,\vo,\V{G}}$
to an operator acting in the full Hilbert space $\HR_m$
by adding
$$
\NPneg{\gamma,\vo,\V{G}}
:=\PAm\,(-\DA-\gamma/|\V{x}|+d\Gamma(\vo))\,\PAm\,,
$$
defined a priori on $\PAm\,\core_m$.
A brief computation shows that
\begin{equation}\label{def-FNP}
\FNP{\gamma,\vo,\V{G}}:=
\NP{\gamma,\vo,\V{G}}\oplus\NPneg{\gamma,\vo,\V{G}}
=\frac{1}{2}\,\PF{\gamma,\vo,\V{G}}
+\frac{1}{2}\,\SA\,\PF{\gamma,\vo,\V{G}}\,\SA\quad
\textrm{on}\;\core_m\,,
\end{equation}
where
\begin{equation}\label{def-PF}
\PF{\gamma,\vo,\V{G}}:= |D_{\V{A}}|-\gamma/{|\V{x}|}+d\Gamma(\vo)\quad
\textrm{on}\;\core_m\,,
\end{equation}
is the semi-relativistic Pauli-Fierz operator.
It turns out that the distinguished self-adjoint realizations 
of $\NP{\gamma,\vo,\V{G}}$
and $\NPneg{\gamma,\vo,\V{G}}$ found later on
are unitarily equivalent. In fact,
the unitary and symmetric matrix 
$\tau:= \alpha_1\,\alpha_2\,\alpha_3\,\beta$ 
leaves $\core_m$ invariant and 
anti-commutes with $\DA$, whence $\tau\,\PA=\PAm\,\tau$.
Consequently, we have
\begin{equation}\label{H+H-}
\NPneg{\gamma,\vo,\V{G}}=\tau\,\NP{\gamma,\vo,\V{G}}\,\tau\quad
\textrm{on}\;\core_m\,.
\end{equation}
In the rest of this section we always assume
that $\vo$ and $\V{G}$ fulfill Hypothesis~\ref{hyp-G}.
$C(d,a,b,\ldots), C'(d,a,b,\ldots)$, etc. denote
positive constants which depend only on the parameter
$d$ appearing in Hypothesis~\ref{hyp-G}
and the additional parameters $a,b,\ldots\;$ displayed in their arguments.
Their values might change from one estimate to another.

To start with we collect a number of useful estimates.
As a consequence of \eqref{Clifford} and the $C^*$-equality we have
\begin{equation}\label{C*}
\|\valpha\cdot \V{v}\|_{\LO(\CC^4)}=|\V{v}|\,,\quad \V{v}\in\RR^3,\qquad
\|\valpha\cdot \V{z}\|_{\LO(\CC^4)}\klg\sqrt{2}\,|\V{z}|\,,\quad 
\V{z}\in\CC^3,
\end{equation}
where $\valpha\cdot\V{z}:=\alpha_1\,z^{(1)}+\alpha_2\,z^{(2)}+\alpha_3\,z^{(3)}$,
for $\V{z}=(z^{(1)},z^{(2)},z^{(3)})\in\CC^3$.
A standard exercise using \eqref{def-d3}, \eqref{C*}, 
the Cauchy-Schwarz inequality, and the canonical commutation
relations \eqref{CCR}, yields
\begin{align}\label{rb-A}
\|\valpha\cdot \V{A}\,\psi\|^2&\klg6d\,
\|(d\Gamma(\vo)+1)^{1/2}\,\psi\|^2,
\quad\psi\in\dom(d\Gamma(\vo)^{1/2})\,.
\end{align}
In particular, $\valpha\cdot\V{A}$ is a symmetric operator
on $\dom(d\Gamma(\vo)^{1/2})$.
We also employ the following consequence of
\cite[Lemma~3.3]{MatteStockmeyer2009a}:
For every $\nu\in[0,1]$,
$\SA$ maps $\dom(d\Gamma(\vo)^\nu)$ into
itself and 
\begin{align}\label{faysal1}
\big\|(d\Gamma(\vo)+1)^{\nu}\,\SA\,
(d\Gamma(\vo)+1)^{-\nu}\big\|&\klg C({d})\,,\quad\nu\in[0,1]\,.
\end{align}
In Lemma~\ref{le-faysal} we prove that
$\triangle S:=\SA-\SO$ maps $d\Gamma(\vo)^{1/2}$ into
$\bigcup_{\kappa<1}\dom(|\DO|^\kappa)$, and
\begin{equation}\label{faysal-neu}
\big\|(d\Gamma(\vo)+1)^{\mu}\,|\DO|^{\kappa}\,\triangle S\,
(d\Gamma(\vo)+1)^{\nu}\big\|
\klg C({d},\kappa)\,,
\end{equation}
for all $\mu,\nu\in[-1,1]$, $\mu+\nu\klg-1/2$, and $\kappa\in[0,1)$.
(A similar but less general bound has been
obtained in \cite{MatteStockmeyer2009a}.)
Next, we recall the following strengthened
version of the generalized Hardy inequality 
obtained in \cite{SoSoeSp2008}, for $\kappa=1$, and in 
\cite{Frank2009} in full generality (and arbitrary dimension):
Let $0<\ve<\kappa<3$ and let
$h_\kappa:=2^\kappa\Gamma([3+\kappa]/4)^2/\Gamma([3-\kappa]/4)^2$ denote
the sharp constant in the generalized Hardy inequality
in three dimensions, so that $h_{1}=2/\pi$. Then there is some
$C(\kappa,\ve)\in(0,\infty)$ such that
\begin{equation}\label{FSSS}
(C(\kappa,\ve)/\ell^{\kappa-\ve})\,(-\Delta)^{\ve/2}\klg
(-\Delta)^{\kappa/2}-h_{\kappa}\,|\V{x}|^{-\kappa}+\ell^{-\kappa},\qquad\ell>0\,.
\end{equation}
The well-known corollary, $h_{\kappa}\,|\V{x}|^{-\kappa}\klg|\DO|^\kappa$,
together with \eqref{faysal-neu} yields
\begin{align}\label{faysal2V}
\big\|\,|\V{x}|^{-\kappa}(d\Gamma(\vo)+1)^{\mu}\,\triangle S\,
(d\Gamma(\vo)+1)^{\nu}\big\|
&\klg C'({d},\kappa)\,,
\end{align}
for all $\mu,\nu\in[-1,1]$, $\mu+\nu\klg-1/2$, and $\kappa\in[0,1)$.
Finally, we recall the following
special case of \cite[Corollary~3.4]{Matte2009}:
\begin{equation}\label{faysal3}
\big\|\,|\DA|^{1/2}\,[\SA\,,\,d\Gamma(\vo)]\,(d\Gamma(\vo)+1)^{-1/2}\,\big\|
\klg\,C({d})\,.
\end{equation}
In view of \eqref{def-FNP} and \eqref{def-PF} we have
\begin{align}\label{max1}
\FNP{\gamma,\vo,\V{G}}-\FNP{\gamma,\vo,\V{0}}=
X_1-\frac{\gamma}{2}\,X_2+\frac{1}{2}\,X_3\,,\qquad
\PF{\gamma,\vo,\V{G}}-\PF{\gamma,\vo,\V{0}}=X_1\,,
\end{align}
where
\begin{align}\label{max2}
X_1&:=|\DA|-|\DO|=\SA\,\valpha\cdot\V{A}+\triangle S\,\DO\,,
\\\label{max3}
X_2&:=\SA\,|\V{x}|^{-1}\SA-\SO\,|\V{x}|^{-1}\SO
=\SA\,|\V{x}|^{-1}\triangle S+\triangle S\,|\V{x}|^{-1}\SO\,,
\\\label{max4}
X_3&:=\SA\,d\Gamma(\vo)\,\SA-\SO\,d\Gamma(\vo)\,\SO
=\SA\,[d\Gamma(\vo)\,,\,\SA]\,.
\end{align}
Here we used $d\Gamma(\vo)=\SO\,d\Gamma(\vo)\,\SO$
in the last line. We know that the operator identities 
\eqref{max1}--\eqref{max4} are valid at least on 
$\dom(\DO)\cap\dom(d\Gamma(\vo))$.

\begin{lemma}\label{le-lea}
For all $\vp\in\dom(\DO)\cap\dom(d\Gamma(\vo))$, $\ve\in(0,1]$,
$\delta>0$, and $j=1,2,3$,
\begin{align}\label{max5}
|\SPn{\vp}{X_j\,\vp}|&\klg\delta\,
\SPn{\vp}{(|\DO|^\ve+d\Gamma(\vo))\,\vp}+C(d,\delta,\ve)\,\|\vp\|^2,
\\\label{max6}
\|X_j\,\vp\|
&\klg\delta\,\|\,|\DO|^{\ve}\,\vp\|+\delta\,\|d\Gamma(\vo)\,\vp\|
+C'({d},\delta,\ve)\,\|\vp\|\,.
\end{align}
\end{lemma}

\begin{proof}
We pick some $\vp\in\dom(\DO)\cap\dom(d\Gamma(\vo))$
and put $\Theta:=d\Gamma(\vo)+1$.
On account of \eqref{rb-A} and \eqref{faysal-neu} we have
$|\SPn{\vp}{X_1\,\vp}|\klg 
C(d,\ve)\,\|\,|\DO|^{\ve/4}\,\vp\|\,\|\Theta^{1/2}\,\vp\|$
and
$\|X_1\,\vp\|\klg C'(d,\ve)\,\|\,|\DO|^{\ve/4}\otimes\Theta^{1/2}\,\vp\|$.
By the inequality between the weighted geometric and arithmetic means 
we further have, for $\nu\in[0,1]$, $\ve\in(0,1]$, and $\delta>0$,
\begin{align}\nonumber
\|\,|\DO|^{\ve/4}\otimes\Theta^{\nu/2}\,\vp\|^2
&\klg \|\,|\DO|^{\ve/2}\,\vp\|\,\|\Theta^\nu\,\vp\|
\klg\|\vp\|^{1/2}\,\|\,|\DO|^{\ve}\,\vp\|^{1/2}\,\|\Theta^\nu\,\vp\|
\\\label{gerald}
&\klg\SPb{\vp}{\big(\delta^2\,|\DO|^{2\ve}
+\delta^2\,\Theta^{2\nu}+(2\delta)^{-6}\big)\,\vp}\,,
\end{align}
which yields \eqref{max5}\&\eqref{max6}, for $j=1$
(and with new $\delta$ and $\ve$).

Since 
$|\DO|\,\triangle S=\DA-\DO+(|\DO|-|\DA|)\,\SA=\valpha\cdot\V{A}-X_1\,\SA$
we obtain, using \eqref{rb-A}, \eqref{faysal1}, 
and \eqref{max6} with $j=1$,
\begin{align*}
\big\|\,|\DO|\,\triangle S\,\vp\big\|
&\klg
\delta\,\big\|\,|\DO|^\ve\SA\,\vp\big\|
+\delta\,(1+C(d))\,\|d\Gamma(\vo)\,\vp\|+C'\,\|\vp\|
\\
&\klg\delta\,\big\|\,|\DO|^\ve\triangle S\,\vp\big\|+
\delta\,\big\|\,|\DO|^\ve\,\vp\big\|
+\bigO(\delta)\,\|d\Gamma(\vo)\,\vp\|+C'\,\|\vp\|\,,
\end{align*}
where $C'\equiv C'(d,\delta,\ve)$.
Choosing $\ve\klg1/2$ we further observe that, for all $\rho>0$,
\begin{equation}\label{max7}
\big\|\,|\DO|^\ve\,\triangle S\,\vp\big\|
=\SPb{\triangle S\,\vp}{|\DO|^{2\ve}\,\triangle S\,\vp}^{1/2}\klg
\rho\,\big\|\,|\DO|^{2\ve}\,\triangle S\,\vp\big\|+ C(\rho)\,\|\vp\|\,.
\end{equation}
Assuming $\delta\klg1$, $\rho\klg1/2$, and
combining the previous two estimates we obtain
$$
2^{-1}\big\|\,|\V{x}|^{-1}\triangle S\,\vp\big\|
\klg\big\|\,|\DO|\,\triangle S\,\vp\big\|\klg
2\delta\,\big\|\,|\DO|^\ve\vp\big\|
+\bigO(\delta)\,\|d\Gamma(\vo)\,\vp\|+C''\|\vp\|.
$$
Moreover, \eqref{faysal2V} yields
$\|\triangle S\,|\V{x}|^{-1}\SO\,\vp\|
\klg C({d},\ve)\,\|\,|\DO|^{\ve/4}\otimes\Theta^{1/2}\,\vp\|$,
for every $\ve>0$, and we readily obtain \eqref{max6} with $j=2$.
To prove \eqref{max5} with $j=2$, we estimate
\begin{align}\label{max8}
|\SPn{\vp}{S_{\wt{\V{A}}}\,|\V{x}|^{-1}\triangle S\,\vp}|
&\klg C(d,\ve)\,\big\|\,|\DO|^{\ve/4}\SAt\,\vp\big\|\,\|\Theta^{1/2}\,\vp\|\,,
\end{align}
where $\wt{\V{A}}$ ist $\V{0}$ or $\V{A}$.
If $\wt{\V{A}}=\V{0}$, then we use \eqref{gerald} to estimate
the RHS in \eqref{max8} from above by the RHS in \eqref{max5}.
In the case $\wt{\V{A}}=\V{A}$ we further estimate 
$\|\,|\DO|^{\ve/4}\SA\,\vp\|
\klg\|\,|\DO|^{\ve/4}\triangle S\,\vp\|+\|\,|\DO|^{\ve/4}\,\vp\|$
and use \eqref{gerald} once again, as well as the following
consequence of \eqref{faysal-neu} and \eqref{max7},
\begin{align*}
C&(d,\ve)\,\big\|\,|\DO|^{\ve/4}\triangle S\,\vp\big\|\,\|\Theta^{1/2}\,\vp\|
\\
&\klg \rho\,\big\|\,|\DO|^{\ve/2}\triangle S\,\vp\big\|\,\|\Theta^{1/2}\,\vp\|
+C(d,\ve,\rho)\,\|\vp\|\,\|\Theta^{1/2}\,\vp\|
\\
&\klg \rho\,C(d,\ve)\,\|\Theta^{1/2}\,\vp\|^2
+C(d,\ve,\rho)\,\|\vp\|\,\|\Theta^{1/2}\,\vp\|
\\
&\klg \delta\,\SPn{\vp}{\Theta\,\vp}+C'(d,\delta,\ve)\,\|\vp\|^2,
\end{align*}
where $\rho>0$ is chosen such that $\rho\,C(d,\ve)=\delta/2$.

For $j=3$, \eqref{max5}\&\eqref{max6} are simple consequences
of \eqref{faysal3} and \eqref{max4}.
\end{proof}

\smallskip

\noindent
In the next theorem we write
$$
\gcnp:=2/(2/\pi+\pi/2)\,,\qquad \gcPF:=2/\pi\,,
$$
and we shall first use the full strength
of \eqref{FSSS}. We shall also employ its analogue
for the Brown-Ravenhall operator acting in $L^2(\RR^3,\CC^4)$,
\begin{equation}\label{def-BR}
B_\gamma^\el:=|\DO|-(\gamma/2)|\V{x}|^{-1}
-(\gamma/2)\,\SO\,|\V{x}|^{-1}\SO\,,\quad
\gamma\in[0,\gcnp]\,.
\end{equation}
$B_\gamma^\el$ is defined by means of a Friedrichs extension
starting from $C_0^\infty(\RR^3,\CC^4)$ and it is
known that $B_\gamma^\el\grg1-\gamma>0$, for $\gamma\in[0,\gcnp]$
\cite{EPS1996,Tix1998}.
The analogue of \eqref{FSSS} for $B_\gamma^\el$ is proven in \cite{Frank2009}
in the {massless} case and can be written as
\begin{align*}
(-\Delta)^{\frac{1}{2}}-(\gcnp/2)|\V{x}|^{-1}
-(\gcnp/2)\,\SO^{(0)}\,|\V{x}|^{-1}\SO^{(0)}&\grg
(C(\ve)/\ell^{1-\ve})\,(-\Delta)^{\frac{\ve}{2}}-
\ell^{-1},
\end{align*}
for $\ve\in(0,1)$ and $\ell>0$,
where 
$\SO^{(0)}$ acts in Fourier space by multiplication with
$\valpha\cdot\vxi/|\vxi|$.
Since the symbol of $\SO$ is 
$(\valpha\cdot\vxi+\beta)/\SL\vxi\SR$, we have
$[(\SO^{(0)}-\SO)\,\psi]^\wedge(\vxi)
=\SL\vxi\SR^{-1}\,F(\vxi)\,\wh{\psi}(\vxi)$,
where $\|F(\vxi)\|\klg2$. Hence,
$\|\,|\V{x}|^{-1}(\SO^{(0)}-\SO)\|\klg4$ 
by Hardy's inequality.
Of course, $(-\Delta)^{1/2}\klg|\DO|$, and we conclude that,
for $\ve\in(0,1)$,
\begin{align}\label{FSW-m}
(-\Delta)^{\ve/2}&\klg(\ell^{1-\ve}/C(\ve))\, B^\el_{\gcnp}
+(4\gcnp\,\ell^{1-\ve}+\ell^{-\ve})/C(\ve)\,,\quad\ell>0\,.
\end{align}
In particular, 
$\dom(B^\el_{\gamma})\subset\form(B^\el_{\gamma})
\subset\bigcap_{\ve<1}\form(|\DO|^\ve)$.

\begin{theorem}\label{thm-Friedrichs-ext}
Let $\sharp\in\{\mathrm{np},\mathrm{PF}\}$ and 
$\gamma\in [0,\gc^\sharp]$.
Then $H^{\sharp}_{\gamma,\vo,\V{G}}$
is infinitesimally form bounded on $\core_m$
with respect to $H^\sharp_{\gamma,\vo,\V{0}}$.
More precisely, for all $\delta>0$ and $\ve\in(0,1)$, 
we have, in the sense of quadratic forms on 
$\dom(\DO)\cap\dom(d\Gamma(\vo))$,
\begin{align}\label{iris1}
\pm(H^{\sharp}_{\gamma,\vo,\V{G}}-H^\sharp_{\gamma,\vo,\V{0}})
&\klg\delta\,|\DO|^\ve+\delta\,d\Gamma(\vo)
+C(d,\delta,\ve)\,,
\\\label{iris2}
\pm(H^{\sharp}_{\gamma,\vo,\V{G}}-H^\sharp_{\gamma,\vo,\V{0}})
&\klg\delta\,H^\sharp_{\gamma,\vo,\V{0}}
+C(d,\delta)\,,
\\\label{iris3}
(-\Delta)^{\ve/2}+\delta\,d\Gamma(\vo)
&\klg2\delta\,H^{\sharp}_{\gamma,\vo,\V{G}}+C'(d,\delta,\ve)\,,
\\\label{iris4}
|\DA|^{\ve}&\klg\delta\,H^{\sharp}_{\gamma,\vo,\V{G}}+C''(d,\delta,\ve)\,.
\end{align}
Hence, by the KLMN theorem 
$H^{\sharp}_{\gamma,\vo,\V{G}}$ has a distinguished
self-adjoint extension
-- henceforth again denoted by the same symbol --
such that 
$\dom(H^{\sharp}_{\gamma,\vo,\V{G}})\subset\form(H^{\sharp}_{\gamma,\vo,\V{0}})$.
Furthermore,
$\form(H^{\sharp}_{\gamma,\vo,\V{G}})=\form(H^{\sharp}_{\gamma,\vo,\V{0}})$.
If $\gamma<\gc^\sharp$, then we have
$\form(H^{\sharp}_{\gamma,\vo,\V{G}})=\form(\PF{0,\vo,\V{0}})
=\form(|\DO|)\cap\form(d\Gamma(\vo))$.
In the critical case we have 
$\form(\PF{0,\vo,\V{0}})
\subset\form(H^{\sharp}_{\gc^\sharp,\vo,\V{G}})
\subset\bigcap_{\ve<1}\form(|\DO|^\ve)\cap\form(d\Gamma(\vo))$.
\end{theorem}

\begin{proof}
The form bounds \eqref{iris1}--\eqref{iris3} are consequences of
\eqref{FSSS}, \eqref{max5}, and \eqref{FSW-m}.
\eqref{iris4} follows from \eqref{iris3} and \eqref{tina2} below.

If $\gamma<\gc^\sharp$ is sub-critical, we
have $B^\el_\gamma\grg(1-\gamma/\gcnp)\,|\DO|$ and
$|\DO|-\gamma/|\V{x}|\grg(1-\gamma/\gcPF)\,|\DO|$,
respectively, whence
$(1-\gamma/\gc^\sharp)\,\PF{0,\vo,\V{0}}
\klg H^\sharp_{\gamma,\vo,\V{0}}\klg\PF{0,\vo,\V{0}}$
on $\form(|\DO|)\cap\form(d\Gamma(\vo))$,
where $\PF{0,\vo,\V{0}}=|\DO|+d\Gamma(\vo)$.
In the critical case we only have 
$|\DO|^\ve+d\Gamma(\vo)\klg H^\sharp_{\gc^\sharp,\vo,\V{0}}+C(\ve)$,
for every $\ve\in(0,1)$, as a lower bound.
\end{proof}

\begin{theorem}\label{thm-rb}
For $\sharp\in\{\mathrm{np},\mathrm{PF}\}$ and 
$\gamma\in [0,\gc^\sharp]$,
the following holds true:

\smallskip

\noindent(i)
$H^{\sharp}_{\gamma,\vo,\V{G}}$ and $H^{\sharp}_{\gamma,\vo,\V{0}}$ 
have the same domain and their operator cores coincide.

\smallskip

\noindent(ii)
For all $\delta,\ve>0$ and $\vp\in\dom(H^{\sharp}_{\gamma,\vo,\V{0}})$,
\begin{align}\label{maja1}
\big\|(H^{\sharp}_{\gamma,\vo,\V{G}}-H^{\sharp}_{\gamma,\vo,\V{0}})\,\vp\big\|
&\klg
\delta\,\|\,|\DO|^{\ve}\,\vp\|+\delta\,\|d\Gamma(\vo)\,\vp\|
+C({d},\delta,\ve)\,\|\vp\|\,,
\\\label{maja2}
\big\|(H^{\sharp}_{\gamma,\vo,\V{G}}-H^{\sharp}_{\gamma,\vo,\V{0}})\,\vp\big\|
&\klg \delta\,\|H^{\sharp}_{\gamma,\vo,\V{0}}\,\vp\|
+C({d},\delta)\,\|\vp\|\,.
\end{align}
\noindent(iii)
$\dom(H^{\sharp}_{\gamma,\vo,\V{G}})\subset\dom(d\Gamma(\vo))$ and,
for all $\delta>0$ and $\vp\in\dom(H^{\sharp}_{\gamma,\vo,\V{G}})$,
\begin{align}\label{UsefulBounds1}
\|d\Gamma(\vo)\,\vp\|
&\klg\|H^{\sharp}_{\gamma,\vo,\V{0}}\,\vp\|
\klg (1+\delta)\,\|H^{\sharp}_{\gamma,\vo,\V{G}}\,\vp\|
+C({d},\delta)\,\|\vp\|\,.
\end{align}
\end{theorem}

\begin{proof}
For $\vp\in\cX:=\dom(\DO)\cap\dom(d\Gamma(\vo))$,
the bound \eqref{maja1} follows immediately from
\eqref{max1}--\eqref{max4}, and Lemma~\ref{le-lea}.
We define $T:=H^{\sharp}_{\gamma,\vo,\V{G}}-H^{\sharp}_{\gamma,\vo,\V{0}}$
on the domain $\dom(T):=\cX$.
We also fix some $\ve\in(0,1/2)$ in what follows.
As a symmetric operator $T$ is closable. We deduce that
$\dom(\ol{T})\supset\cY:=\dom(|\DO|^\ve)\cap\dom(d\Gamma(\vo))$
and 
\begin{equation}\label{petra0}
\|\ol{T}\,\vp\|\klg\delta\,\|\,|\DO|^{\ve}\,\vp\|
+\delta\,\|d\Gamma(\vo)\,\vp\|
+C({d},\delta,\ve)\,\|\vp\|\,,
\quad\vp\in\cY.
\end{equation}
As a next step we estimate the RHS of \eqref{petra0} from above
in the case $\sharp=\mathrm{np}$.
For $\ve\in(0,1/2)$ and an appropriate choice of $\ell>0$ in \eqref{FSW-m},
we obtain
\begin{equation}\label{petra1}
\|\,|\DO|^{\ve}\,\vp\|^2\klg\SPn{\vp}{(B^\el_{\gamma}+C(\ve))\,\vp}
\klg\SPn{\vp}{\FNP{\gamma,\vo,\V{0}}\,\vp}+C(\ve)\,\|\vp\|^2,
\end{equation}
for $\vp\in\cZ:=\dom(B^\el_\gamma)\cap\dom(d\Gamma(\vo))\subset\cY$.
In view of $B^\el_{\gamma}\otimes d\Gamma(\vo)\grg0$
and $d\Gamma(\vo)=\SO\,d\Gamma(\vo)\,\SO$
we may further estimate
\begin{equation}\label{petra2}
\|d\Gamma(\vo)\,\vp\|^2
\klg\big\|(B^\el_{\gamma}+d\Gamma(\vo)/2+\SO d\Gamma(\vo)\SO/2)
\,\vp\big\|^2=\|\FNP{\gamma,\vo,\V{0}}\,\vp\|^2,
\end{equation}
for $\vp\in\cZ$.
Combining \eqref{petra0}--\eqref{petra2}
we obtain
\begin{equation}\label{petra3}
\|\ol{T}\,\vp\|\klg\delta\,\|\FNP{\gamma,\vo,\V{0}}\,\vp\|
+C({d},\delta)\,\|\vp\|\,,
\quad\vp\in\cZ\,,
\end{equation}
where $\cZ$ is an operator core for $\FNP{\gamma,\vo,\V{0}}$.
According to the Kato-Rellich theorem
$\wt{H}^{\mathrm{np}}_{\gamma,\vo,\V{G}}:=\FNP{\gamma,\vo,\V{0}}+\ol{T}$ is
self-adjoint on 
$\dom(\FNP{\gamma,\vo,\V{0}})\subset\form(\FNP{\gamma,\vo,\V{0}})$
and the operator cores of $\wt{H}^{\mathrm{np}}_{\gamma,\vo,\V{G}}$
and $\FNP{\gamma,\vo,\V{0}}$ coincide.
Furthermore, $\wt{H}^{\mathrm{np}}_{\gamma,\vo,\V{G}}$
and ${H}^{\mathrm{np}}_{\gamma,\vo,\V{G}}$ coincide
on $\core_m$ and we know from Theorem~\ref{thm-Friedrichs-ext} that
${H}^{\mathrm{np}}_{\gamma,\vo,\V{G}}$
is uniquely determined by the property
$\dom({H}^{\mathrm{np}}_{\gamma,\vo,\V{G}})\subset\form(\FNP{\gamma,\vo,\V{0}})$.
Therefore, 
$\wt{H}^{\mathrm{np}}_{\gamma,\vo,\V{G}}={H}^{\mathrm{np}}_{\gamma,\vo,\V{G}}$
which proves (i)--(iii), for $\sharp=\mathrm{np}$.
In the case $\sharp=\mathrm{PF}$ we use
\eqref{FSSS} instead of \eqref{FSW-m} and put
$\cZ:=\dom(|\DO|-\gamma/|\V{x}|)\cap\dom(d\Gamma(\vo))$.
\end{proof}

\begin{corollary}\label{cor-esa}
(i) 
For $\gamma\in[0,\gcnp]$, the algebraic tensor product
$\dom(B^{\el}_\gamma)\otimes\sC_m$
is an operator core of $\FNP{\gamma,\vo,\V{G}}$.
($\sC_m$ has been defined below \eqref{def-DmCm}.)

\smallskip

\noindent
(ii) For $\gamma\in[0,\gcPF]$, the algebraic tensor product
$\dom(|\DO|-\gamma/|\V{x}|)\otimes\sC_m$
is an operator core of $\PF{\gamma,\vo,\V{G}}$. 

\smallskip

\noindent
(iii) For $\gamma\in[0,1/2)$, 
$\FNP{\gamma,\vo,\V{G}}$ and $\PF{\gamma,\vo,\V{G}}$
are essentially self-adjoint on $\core_m$.
\end{corollary}

\begin{proof}
The domains appearing in (i) (resp. (ii)) are operator cores
of $H^\sharp_{\gamma,\vo,\V{0}}$ and, hence, of
$H^\sharp_{\gamma,\vo,\V{G}}$ by Theorem~\ref{thm-rb}(i).
By Hardy's inequality,
$\|\,|\V{x}|^{-1}\vp\|\klg2\|\,|\DO|\,\vp\|$ and
$\|\SO\,|\V{x}|^{-1}\SO\,\vp\|\klg2\|\,|\DO|\,\vp\|$, whence
$\|(H^{\sharp}_{\gamma,\vo,\V{0}}-\PF{0,\vo,\V{0}})\,\vp\|
\klg2\gamma\,\|\PF{0,\vo,\V{0}}\,\vp\|$,
for all $\vp\in\dom(\DO)\cap\dom(d\Gamma(\vo))$.
Since $\PF{0,\vo,\V{0}}$ is essentially self-adjoint on $\core_m$
the same holds true for $H^{\sharp}_{\gamma,\vo,\V{0}}$
by the Kato-Rellich theorem, provided that $\gamma<1/2$. Hence,
(iii) follows from Theorem~\ref{thm-rb}(i), too.
\end{proof}

\smallskip

\begin{proof}[Proof of Proposition~\ref{prop-stab}]
The semi-boundedness in the case $\gamma\klg\gcnp$
follows from Theorem~\ref{thm-Friedrichs-ext}.
Let $\wt{\gamma}>\gcnp$ and pick some $\gamma\in(\gcnp,\wt{\gamma})$.
Due to
\cite{EPS1996} we find normalized $\psi_n\in\dom(\DO)$, $n\in\NN$,
such that $\SPn{\psi_n}{B^\el_\gamma\,\psi_n}\to-\infty$, as
$n\to\infty$, where $B^\el_\gamma$ now denotes the expression
on the RHS of \eqref{def-BR} with domain $\dom(\DO)$.
Let $\Omega:=(1,0,0,\ldots\:)$ denote the vacuum
vector in $\Fock[\HP_m]$ and set $\Psi_n:=\psi_n\otimes\Omega$,
so that $\|\Psi_n\|=1$ and $d\Gamma(\vo)\,\Psi_n=0$.
On account of \eqref{max5} we obtain
\begin{align*}
\SPb{&\Psi_n}{\FNP{\wt{\gamma},\vo,\V{G}}\,\Psi_n}
\\
&\klg
\SPb{\Psi_n}{(B^\el_{\wt{\gamma}}+d\Gamma(\vo))\,\Psi_n}
+\delta\,\SPb{\Psi_n}{(|\DO|^{1/2}+d\Gamma(\vo))\,\Psi_n}+C(d,\delta)
\\
&\klg
\frac{\wt{\gamma}}{\gamma}\,\SPb{\psi_n}{B^\el_{{\gamma}}\,\psi_n}
-\Big(\frac{\wt{\gamma}}{\gamma}-1-\delta\Big)\,\SPn{\psi_n}{|\DO|\,\psi_n}
+C'(d,\delta)
\,.
\end{align*}
Choosing $\delta:=\wt{\gamma}/\gamma-1>0$ we see that
$\SPn{\Psi_n}{\FNP{\wt{\gamma},\vo,\V{G}}\,\Psi_n}\to-\infty$, $n\to\infty$.
In view of \eqref{def-FNP} this implies that
$\SPn{\PA\,\Psi_n}{\NP{\wt{\gamma},\vo,\V{G}}\,\PA\,\Psi_n}\to-\infty$
or $\SPn{\PAm\,\Psi_n}{\NPneg{\wt{\gamma},\vo,\V{G}}\,\PAm\,\Psi_n}\to-\infty$.
If the latter divergence holds true, then
$\SPn{\PA\,\tau\,\Psi_n}{\NP{\wt{\gamma},\vo,\V{G}}\,\PA\,\tau\,\Psi_n}
\to-\infty$ by \eqref{H+H-} which concludes the proof.
\end{proof}


\section{Convergence of no-pair operators}\label{sec-conv}
 
\noindent
The following localization estimate \cite{MatteStockmeyer2009a}
plays an essential role in the sequel:

\begin{proposition}[{\bf Exponential localization}]\label{prop-exp-loc}
There exists $k\in(0,\infty)$ and, 
for all $\vo$ and $\V{G}$ fulfilling
Hypothesis~\ref{hyp-G} and all $\gamma\in(0,\gcnp)$, 
we find some $C\equiv C(\gamma,d)\in(0,\infty)$
such that the following holds true: 
Let $\lambda<\Th:=\inf\spec[\FNP{0,\vo,\V{G}}]$ and let
$a>0$ satisfy $a\klg k(\gc-\gamma)/(\gc+\gamma)$ 
and $\ve:= 1-\frac{\lambda +C}{\Th+C}\,k\,a^2>0$.
Then $\Ran(\id_\lambda(\FNP{\gamma,\vo,\V{G}}))\subset\dom(e^{a|\V{x}|})$ and
\begin{equation}\label{eq-exp-loc}
\big\|\,e^{a|\V{x}|}\,\id_\lambda(\FNP{\gamma,\vo,\V{G}})\,\big\|\klg\,
(k/\ve^2) (\Th+C)\,e^{k\,a(\Th+C)/\ve}.
\end{equation}
\end{proposition}

\begin{proof}
If $\FNP{\gamma,\vo,\V{G}}$ is replaced by $\NP{\gamma,\vo,\V{G}}$, then
the assertion follows from \cite[Theorem~2.2]{MatteStockmeyer2009a}.
In view of \eqref{def-FNP} and \eqref{H+H-} it is, however, clear
that the same estimate holds also for $\FNP{\gamma,\vo,\V{G}}$.
\end{proof}

\begin{remark}\label{rem-Th}
To apply Proposition~\ref{prop-exp-loc} later on
we note that,
in view of \eqref{iris2},
$\Th=\inf\spec[\FNP{0,\vo,\V{G}}]\klg C'(d)<\infty$,
where $C'(d)$ depends only on $d$.
\end{remark}

\smallskip

\noindent
In the next proposition we 
assume that $\vo$ and $\V{G}$ 
fulfill Hypothesis~\ref{hyp-G} with parameter $d$ and that,
for every $n\in\NN$,
$\vo_n$ and $\V{G}_n$ fulfill Hypothesis~\ref{hyp-G} with 
the same parameter $d$ such that
$$
\forall\,a>0:\quad
\triangle_{n}(a):=
\int\Big(1+\frac{1}{\vo(k)}\Big)\,
\sup_{\V{x}}e^{-2a|\V{x}|}\big|\V{G}_{n,\V{x}}(k)-\V{G}_\V{x}(k)\big|^2\,dk
\xrightarrow{n\to\infty}0\,.
$$
Furthermore, we assume that
$|\vo-\vo_n|\klg \vk_n\,\vo$, for some $\vk_n\grg0$,
$\vk_n\searrow0$.
To simplify the notation 
we put
\begin{align*}
H&:=\FNP{\gamma,\vo,\V{G}}\,,\quad H_n:=\FNP{\gamma,\vo_n,\V{G}_n}\,,
\quad E:=\inf\spec[H]\,,\quad
E_n:=\inf\spec[H_n]\,,
\\
\Th&:=\inf\spec[\FNP{0,\vo,\V{G}}]\,,\quad
\Th_n:=\inf\spec[\FNP{0,\vo_n,\V{G}_n}]\,,
\end{align*}
and, for some $z\in\CC\setminus\RR$,
$$
\cR(z):=(H-z)^{-1},\qquad 
\cR_n(z):=(H_n-z)^{-1}.
$$

\begin{proposition}\label{prop-conv}
For $\gamma\in(0,\gcnp)$ and
under the above assumptions, the following holds:

\smallskip

\noindent(1)
Let $\lambda<\Sigma$ and
$f\in C^\infty_0(\RR)$.
Then 
$$
\lim_{n\to \infty}\big\|\big\{f(H_n)
-f(H)\big\}\,\id_\lambda(H)\big\|
=0\,.
$$
\noindent(2)
Let $\lambda<\Sigma$ and $\mu>\lambda$.
Then we find some $N\in\NN$ such that, for all $n\grg N$,
$$
\dim\Ran\big(\id_\lambda(H)\big)
\klg\dim\Ran\big(\id_\mu(H_n)\big)\,.
$$
\noindent(3) 
$\ol{E}:=\varlimsup E_n\klg E$.

\smallskip

\noindent
If, in addition, there is some $c>0$ such that
$\Th_n-E_n\grg c$, for all $n\in\NN$,
then the following holds true also:

\smallskip

\noindent(4) 
$\ul{E}:=\varliminf E_n\grg E$, thus $\lim\limits_{n\to\infty}E_n=E$.

\smallskip

\noindent(5)
Let 
$\phi_n\in\Ran\big(\id_{E_n+1/n}(H_n)\big)$, $n\in\NN$,
be normalized and let $\phi\in\HR_m$ denote a weak limit
of some subsequence of $\{\phi_n\}$.
If $\phi\not=0$, then $\phi$
is a ground state eigenvector of $H$.
\end{proposition}

\begin{proof}
(1): Let $z\in\CC\setminus\RR$, $\vp,\psi\in\HR$, 
$\vp_{n,z}:=\cR_n(z)\,\vp$,
$\psi_{z}:=\cR(z)\,\psi$, and $\delta S_n:=\SA-S_{\V{A}_n}$,
where $\V{A}\equiv\V{A}[\V{G}]$ and $\V{A}_n\equiv\V{A}[\V{G}_n]$
as defined in \eqref{def-Aphys}.
Theorem~\ref{thm-rb}(i) and the bound 
$(1-\vk_n)\,\vo\klg\vo_n\klg(1+\vk_n)\,\vo$
imply that $H$ and $H_n$
have the same domain and the latter is contained in
$\form(|\DO|)\cap\form(d\Gamma(\vo))$, if $\vk_n<1$.
For large $n$, we thus have
\begin{align}\nonumber
2\,\SPb{\vp}{(\cR_n(z)-\cR(z))&\,\psi}
=
2\,\SPb{\vp_{n,z}}{(|\DA|-|\D{\V{A}_n}|)\,\psi_{z}}
\\
&\nonumber
+\SPb{\vp_{n,z}}{\big(d\Gamma(\vo-\vo_n)+
S_{\V{A}_n}d\Gamma(\vo-\vo_n)\SA\big)\,\psi_{z}}
\\
&\nonumber
+
\SPb{\vp_{n,z}}{
\delta S_n\,(-\gamma/|\V{x}|+d\Gamma(\vo))\,\SA\,\psi_{z}}
\\
&
+
\SPb{\vp_{n,z}}{
S_{\V{A}_n}\,(-\gamma/|\V{x}|+d\Gamma(\vo_n))\,\delta S_n\,
\psi_{z}}\,.\label{eni1}
\end{align}
We fix some $\kappa\in(3/4,1)$ and set $\ve:=1-\kappa\in(0,1/4)$,
$\Theta:=d\Gamma(\vo)+1$, $\Theta_n:=d\Gamma(\vo_n)+1$,
and $\Pi:=\id_\lambda(H)$.
In the sequel we always assume that $\psi=\Pi\,\psi$
and that $n$ is so large that $\vk_n\klg1/2$,
so that $\vo\klg2\vo_n$ and, hence, $\Theta\klg2\Theta_n$ and
$\Theta_n\klg2\Theta$.
On account of Proposition~\ref{prop-exp-loc} we further find
some $a>0$ such that $\|e^{a|\V{x}|/\ve}\,\Pi\|\klg C(d,a,\ve,\lambda)$.
Analogously to \eqref{rb-A} we then have
\begin{equation}\label{eni8}
\big\|\,\valpha\cdot(\V{A}-\V{A}_n)\,\Theta^{-1/2}e^{-a|\V{x}|}\,\big\|
\,\klg\,6^{1/2}\,\triangle_{n}^{1/2}(a)\,,
\end{equation}
and Lemma~\ref{le-faysal} below implies
the following bounds,
\begin{align}\label{eni9}
\big\|\,\cO^\kappa\,\Theta^\mu\,\delta S_n\,\Theta^{\nu}e^{-a|\V{x}|}\,\big\|
+\big\|\,|\V{x}|^{-\kappa}\,\Theta^{\ve}\,
\delta S_n\,\Theta^{-\kappa}e^{-a|\V{x}|}\,\big\|
&\klg C(d,\kappa)\,\triangle_{n}^{1/2}(a)\,,
\end{align}
for $\cO\in\{|\DA|,|\D{\V{A}_n}|\}$ and $\mu,\nu\in[-1,1]$
with $\mu+\nu\klg-1/2$ and $\mu\wedge\nu\klg-1/2$.
Lemma~\ref{le-faysal} further implies
\begin{align}\label{eni99}
\big\|\,|\V{x}|^{-\kappa}\,e^{-a|\V{x}|}\,\delta S_n\,\Theta^{-1}\,\big\|
&\klg C'(d,\kappa)\,\triangle_{n}^{1/2}(a)\,.
\end{align}
Using also 
$[\Pi,\cR(z)]=[\Pi,\SA]=[\SA,\cR(z)]=[S_{\V{A}_n},\cR_n(z)]=0$,
and $\|\SA\|=\|S_{\V{A}_n}\|=1$,
we may estimate the first term on the RHS of \eqref{eni1} as
\begin{align*}
&\big|\SPb{\vp_{n,z}}{(|\DA|-|\D{\V{A}_n}|)\,\psi_{z}}\big|
\\
&\klg\big|\SPb{\valpha\cdot(\V{A}-\V{A}_n)\,\vp_{n,z}}{\SA\,\psi_{z}}\big|
+\big|\SPb{S_{\V{A}_n}\,|\D{\V{A}_n}|^{\ve}\,\vp_{n,z}}{|\D{\V{A}_n}|^{\kappa}
\,\delta S_n\,\psi_{z}}\big|
\\
&\klg6^{1/2}\,\triangle_{n}^{1/2}(a)\,
\|\cR_n(z)\big\|\,\|\vp\|\,
\big\|e^{a|\V{x}|}\,\Theta^{1/2}\,\Pi\big\|\,\|\cR(z)\|\,\|\psi\|
\\
&+
\big\|\,|\D{\V{A}_n}|^{\ve}\,\cR_n(z)\big\|\,\|\vp\|\,
\big\|\,|\D{\V{A}_n}|^{\kappa}\,\delta S_n\,\Theta^{-1/2}\,e^{-a|\V{x}|}\big\|
\,\big\|e^{a|\V{x}|}\,\Theta^{1/2}\,\Pi\big\|\,\|\cR(z)\|\,\|\psi\|.
\end{align*}
In view of $[\cR_n(z),S_{\V{A}_n}]=0$
the second term on the RHS of \eqref{eni1} can be estimated as
\begin{align*}
&\big|\SPb{\vp_{n,z}}{\big(d\Gamma(\vo-\vo_n)+
S_{\V{A}_n}\,d\Gamma(\vo-\vo_n)\,\SA\big)\,\psi_{z}}\big|
\\
&\klg
2\big\|d\Gamma(|\vo-\vo_n|)\,\cR_n(z)\big\|\,\|\vp\|\,
\|\psi_{z}\|
\klg\,
4\,\vkap_n\,\big\|\Theta_n\,\cR_n(z)\big\|\,\|\vp\|\,
\|\cR(z)\|\,\|\psi\|
\,.
\end{align*}
Likewise, we obtain for the third term on the RHS of \eqref{eni1}
\begin{align*}
&\big|\SPb{\vp_{n,z}}{
\delta S_n\,(-\gamma/|\V{x}|+d\Gamma(\vo))\,\SA\,\psi_{z}}\big|
\\
&\klg
2\,\big\|e^{-a|\V{x}|}\,|\V{x}|^{-\kappa}\,\delta S_n\,\Theta^{-1}\big\|\,
\big\|\Theta_n\,\cR_n(z)\big\|\,\|\vp\|\,
\big\|\,|\V{x}|^{-\ve}\,e^{a|\V{x}|}\,\Pi\big\|\,\|\cR(z)\|\,\|\psi\|
\\
&\;\;+
2\,\big\|e^{-a|\V{x}|}\,d\Gamma(\vo)^{1/2}\,\delta S_n\,\Theta^{-1}\big\|\,
\big\|\Theta_n\,\cR_n(z)\big\|\,\|\vp\|\,
\big\|e^{a|\V{x}|}\,\Theta^{1/2}\,\Pi\big\|\,\|\cR(z)\|\,\|\psi\|\,,
\end{align*}
where 
$\|\,|\V{x}|^{-\ve}\,e^{a|\V{x}|}\,\Pi\|^2
\klg C(\ve)\,\|\,|\DO|^{2\ve}\,\Pi\|\,\|e^{a|\V{x}|/\ve}\,\Pi\|$,
since $1/\ve>2$, and the norm of
$e^{-a|\V{x}|}\,d\Gamma(\vo)^{1/2}\delta S_n\,\Theta^{-1}
=\{\Theta^{-1}\delta S_n\,e^{-a|\V{x}|}\,d\Gamma(\vo)^{1/2}\}^*$
is bounded according to \eqref{eni9}.
Finally, we treat the fourth term on the RHS of \eqref{eni1},
\begin{align*}
&\big|\SPb{\vp_{n,z}}{
S_{\V{A}_n}\,(-\gamma/|\V{x}|+d\Gamma(\vo_n))\,\delta S_n\,
\psi_{z}}\big|
\\
&\klg\big\|\,|\V{x}|^{-\ve}\,\Theta^{\ve}\,\cR_n(z)\big\|\,\|\vp\|\,
\big\|\,|\V{x}|^{-\kappa}\,\Theta^{-\ve}\,\delta S_n\,
\Theta^{-\kappa}e^{-a|\V{x}|}\big\|\,
\big\|\,e^{a|\V{x}|}\,\Theta^{\kappa}\,\Pi\big\|\,\|\cR(z)\|\,\|\psi\|
\\
&\;\;+
\big\|\,\Theta_n\,\cR_n(z)\big\|\,\|\vp\|\,
\big\|\delta S_n\,\Theta^{-1/2}e^{-a|\V{x}|}\,\big\|\,
\big\|e^{a|\V{x}|}\,\Theta^{1/2}\,\Pi\big\|\,\|\cR(z)\|\,\|\psi\|\,,
\end{align*}
where $\|\,|\V{x}|^{-\ve}\,\Theta^{\ve}\,\cR_n(z)\|^2\klg
C(\ve)\,\|\,|\DO|^{2\ve}\,\cR_n(z)\|\,\|\Theta_n\,\cR_n(z)\|$
since $2\ve<1/2$.
On account of 
\eqref{iris3}, \eqref{iris4}, \eqref{UsefulBounds1}, and $2\ve<1/2$,
\begin{equation*}
\sup_{n\in\NN}\big\|\,\cO\,\cR_n(z)\,\big\|
\klg\frac{C(d,\ve)}{1\wedge|\Im z|}\,,\quad\textrm{for}\;
\cO\in\big\{|\D{\V{A}_n}|^\ve,\,|\DO|^{2\ve},\,\Theta_n\big\}\,,
\end{equation*}
where $a\wedge b:=\min\{a,b\}$.
By virtue of \eqref{UsefulBounds1} and \eqref{eq-exp-loc}
we further have
$$
\|e^{a|\V{x}|}\,\Theta^{1/2}\,\Pi\|\klg
\|e^{a|\V{x}|}\,\Theta^{\kappa}\,\Pi\|\klg
\|e^{a|\V{x}|/\ve}\,\Pi\|^\ve\,\|\Theta\,\Pi\|^\kappa
\klg C(d,\kappa,\lambda)\,,
$$ 
and \eqref{iris3} implies $\|\,|\DO|^{2\ve}\,\Pi\|\klg C'(d,\kappa,\lambda)$.
Combining all the above estimates 
with \eqref{eni9} and \eqref{eni99} we arrive at
\begin{equation}\label{eni13}
\big\|(\cR_n(z)-\cR(z))\,\Pi\big\|=\!\!\!\sup_{\|\vp\|=\|\psi\|=1}
\!\big|\SPb{\vp}{(\cR_n(z)-\cR(z))\,\Pi\,\psi}\big|
\klg\frac{b(n)}{1\wedge|\Im z|^2},
\end{equation}
where $b(n)=\bigO(\vk_n+\triangle^{1/2}_n(a))\to0$, $n\to\infty$.
Now, Part~(1) follows from \eqref{eni13} and the Helffer-Sj\"ostrand formula
(see, e.g., \cite{DimassiSjoestrand}),
\begin{equation}\label{HelfferSjoestrand}
f(T)=\frac{1}{2\pi i}\int_\CC (T-z)^{-1}\, 
\partial_{\ol{z}}\tilde{f}(z)\,dz\wedge d\ol{z}\,,
\end{equation}
valid for every self-adjoint operator $T$ on some Hilbert space.
Here $\tilde{f}\in C_0^\infty(\CC)$ is a compactly supported, 
almost analytic extension of $f$ such that 
$\tilde{f}\!\!\upharpoonright_{\RR}=f$ and 
$$
|\partial_{\ol{z}}\tilde{f}(z)|=\bigO\big(|\Im z|^N\big)\,,
\quad\textrm{for every}\;N\in\NN\,.
$$
\noindent(2):
It suffices to show that
$$
\big\|\big\{\id_\mu(H_n)
-\id_\lambda(H)\big\}\,
\id_\lambda(H)\big\|<1\,,
$$
for all sufficiently large $n$; see, e.g., 
\cite[Lemma~6.8]{DimassiSjoestrand}. 
To this end
we choose $f\in C_0^\infty(\RR,[0,1])$ such that
$f\equiv1$ on $[e_0,\lambda]$, where
$e_0:=\min\{E,\inf_nE_n\}>-\infty$ (by \eqref{iris2}).
Supposing further that $f$ is decreasing on
$[\lambda,\infty)$ with $f(\mu)=1/2$ we may ensure that
$|f-\id_{(-\infty,\mu]}|\klg1/2$ on 
$\bigcup_{n\in\NN}\spec[H_n]$.
Then 
$\id_\lambda(H)
=f(H)\,\id_\lambda(H)$,
whence, by Part~(1),
$$
\big\|\big\{\id_\mu(H_n)
-\id_\lambda(H)\big\}\,
\id_\lambda(H)\big\|\klg\frac{1}{2}+
\big\|\big\{f(H_n)-f(H)\big\}\,\id_\lambda(H)\big\|
\xrightarrow{n\to\infty}\frac{1}{2}\,.
$$
\noindent(3):
Assume we had $E<\ol{E}$.
Then we find $\ve>0$ and integers $n_1<n_2<\ldots\:$
such that $E+\ve<E_{n_\ell}$,
for all $\ell\in\NN$.
Applying (2) with $\lambda:=E+\ve/2$
and $\mu:=E+\ve$ 
we obtain the following contradiction,
for all sufficiently large $\ell$,
$$
0<\dim\Ran\big(\id_\lambda(H)\big)
\klg\dim\Ran\big(\id_\mu(H_{n_\ell})\big)=0\,.
$$
(4)\&(5): We set $\Pi_n:=\id_{E_n+1/n}(H_n)$.
Thanks to the additional assumption $\Th_n-E_n\grg c>0$ 
and Remark~\ref{rem-Th}
we may apply Proposition~\ref{prop-exp-loc} to find
some $n$-independent constants $a,C\in(0,\infty)$
such that
\begin{equation}\label{eni12}
\forall\,n\in\NN\,,\;n>1/c\::\quad
\big\|\,e^{a|\V{x}|}\,\Pi_n\,\big\|\,\klg\,C\,.
\end{equation}
Let $z\in\CC\setminus\RR$.
We observe that in the proof of Part~(1)
we may interchange the roles of $H$ and $H_n$
and the new bound \eqref{eni12} permits to get
the following analogue of \eqref{eni13},
\begin{equation}\label{eni14}
\big\|\,(\cR(z)-\cR_n(z))\,\Pi_n\,\big\|\klg\frac{b'(n)}{1\wedge|\Im z|^2}\,,
\end{equation}
where $0<b'(n)\to0$.
For every $n\in\NN$, we pick some normalized
$\phi_n\in\Ran(\Pi_n)$.
By the spectral calculus
$
(\cR_n(z)-(E_n-z)^{-1})\,\phi_n\to0
$ strongly, as $n\to\infty$.
Furthermore, we find integers $n_1<n_2<\ldots\:$ such that
$E_{n_\ell}\to \ul{E}$, as $\ell\to\infty$, and such that
$\phi:=\underset{\ell\to\infty}{\textrm{w-lim}}\,\phi_{n_\ell}$
exists. By virtue of \eqref{eni14} we first infer that
$$
(\cR(z)-\cR_{n_\ell}(z))\,\phi_{n_\ell}
+\Big(\cR_{n_\ell}(z)-\frac{1}{E_{n_\ell}\!-z}\Big)\,\phi_{n_\ell}
+\Big(\frac{1}{E_{n_\ell}\!-z}-\frac{1}{\ul{E}-z}\Big)\,\phi_{n_\ell}
\xrightarrow{\ell\to\infty}0\,,
$$
strongly.
As the expression on the left equals 
$(\cR(z)-(\ul{E}-z)^{-1})\,\phi_{n_\ell}$
we deduce that $\ul{E}\in\spec[H]$, thus $E\klg\ul{E}$, thus
$E=\ul{E}$ by (3).
Moreover, we obtain 
$$
0=\underset{\ell\to\infty}{\textrm{w-lim}}
\Big(\cR(z)-\frac{1}{E-z}\Big)\,\phi_{n_\ell}=
\Big(\cR(z)-\frac{1}{E-z}\Big)\,\phi\,.
$$
Therefore, $\phi\in\dom(H)$ and
$H\,\phi=E\,\phi$.
\end{proof}


\section{Existence of binding}
\label{sec-binding}

\noindent
In the whole Section~\ref{sec-binding} we assume that
$\vo$ and $\V{G}$ fulfill Hypothesis~\ref{hyp-G}
and that $\V{G}_\V{x}$ can be written as
$\V{G}_\V{x}=e^{-i\vmu\cdot\V{x}}\,\V{g}$
almost everywhere on $\cA_m\times\ZZ_2$, where
$\vmu,\V{g}:\cA_m\times\ZZ_2\to\RR^3$ are measurable
and $|\vmu|\klg\vo$ almost everywhere.

\subsection{Fiber decomposition}
\label{ssec-fiber}

\noindent
In order to prove the binding condition 
we replace $\FNP{\gamma,\vo,\V{G}}$ by some 
suitable unitarily equivalent operator.
This is the reason why we restrict our attention
to coupling functions of the form 
$\V{G}_\V{x}=e^{-i\vmu\cdot\V{x}}\,\V{g}$.
Let us denote the components of $\vmu$
as $\mu^{(j)}$, $j=1,2,3$, and define
$$
\pf:=d\Gamma(\vmu):=
\big(d\Gamma(\mu^{(1)}),d\Gamma(\mu^{(2)}),d\Gamma(\mu^{(3)})\big)\,.
$$
Then a conjugation of the Dirac operator with
the unitary operator $e^{i\pf\cdot\V{x}}$
-- which is simply a multiplication
with the phase $e^{i(\vmu(k_1)+\dots+\vmu(k_n))\cdot\V{x}}$
in each Fock space sector $\Fock^{(n)}[\HP_m]$ -- 
yields
$$
e^{i\pf\cdot\V{x}}\,\DA\,e^{-i\pf\cdot\V{x}}
=\valpha\cdot(-i\nabla_{\V{x}}-\pf+\V{A}(\V{0}))+\beta\,.
$$
A further conjugation with the Fourier transform, 
$\fourier:L^2(\RR^3_\V{x})\to L^2(\RR^3_{\vxi})$,
with respect to the variable 
$\V{x}$ 
turns the transformed Dirac operator into
\begin{equation}\label{tf-Dirac}
\fourier\,e^{i\pf\cdot\V{x}}\,\DA\,
e^{-i\pf\cdot\V{x}}\,\fourier^{*}
=\int_{\RR^3}^\oplus \wh{D}(\vxi)\,d^3\vxi\,,
\end{equation}
where, as usual, $\fourier\equiv\fourier\otimes\id$.
Here the operators
$$
\FD(\vxi):=\valpha\cdot(\vxi-\pf+\V{A}(\V{0}))+\beta\,,
\qquad \vxi\in\RR^3,
$$
act in $\CC^4\otimes\Fock[\HP_m]$. They
are fiber Hamiltonians of
the transformed Dirac operator in \eqref{tf-Dirac}
with respect to the isomorphism
\begin{equation}
  \label{eq:2} 
\HR_m\cong
\int_{\RR^3}^\oplus\CC^4\otimes\Fock[\HP_m]\,d^3\vxi\,.
\end{equation}
For every $\vxi\in\RR^3$, we introduce
\begin{equation}\label{def-projP}
\wh{S}(\vxi):=\FD(\vxi)\,|\FD(\vxi)|^{-1}.
\end{equation}
Corresponding to \eqref{eq:2} we 
then have the following direct integral representation
of the no-pair operator without exterior potential,
\begin{equation}
  \label{eq:3} 
\fourier\,e^{i\pf\cdot\V{x}}\,\FNP{0,\vo,\V{G}}\,
e^{-i\pf\cdot\V{x}}\,\fourier^{*}=
\int_{\RR^3}^\oplus\wh{H}(\vxi)
\,d^3\vxi \,,
\end{equation}
where
$$
\wh{H}(\vxi):=|\FD(\vxi)|+(1/2)\,d\Gamma(\vo)+(1/2)\,
\wh{S}(\vxi)\,d\Gamma(\vo)\,\wh{S}(\vxi)\,.
$$

\subsection{Proof of the binding condition}\label{ssec-binding-proof}

\noindent
In view of \eqref{def-ETh}, \eqref{def-FNP}, and \eqref{H+H-}
we have
$E_\gamma(\vo,\V{G})=\spec[\FNP{\gamma,\vo,\V{G}}]$,
for $\gamma\in(0,\gcnp]$, and
$\Th(\vo,\V{G})=\inf\spec[\FNP{0,\vo,\V{G}}]$.
Let $\rho\in(0,1]$ be fixed in what follows.
From the general theory of direct integrals of self-adjoint
operators and \eqref{eq:3} it follows that there exist
$\ps\in\RR^3$ and $\vp_\star\in\dom(\wh{H}(\vxi_\star))$,
$\|\vp_\star\|=1$, such that
\begin{equation}\label{}
\SPn{\vp_\star}{\wh{H}(\ps)\,\vp_\star}<\Th(\vo,\V{G})+\rho\,.
\end{equation}
We abbreviate
$$
\FD_\star(\vxi):=\FD(\vxi+\ps)\,,\quad
\wh{S}_\star(\vxi):=\wh{S}(\vxi+\ps)\,,\quad
\wh{H}_\star(\vxi):=\wh{H}(\vxi+\ps)\,,
$$
and introduce the unitary transformation
$$
U:=e^{i(\pf-\ps)\cdot\V{x}}\,.
$$
Then $\FNP{\gamma,\vo,\V{G}}$ can be written as
$$
\FNP{\gamma,\vo,\V{G}}
=U^*\,\fourier^*\int_{\RR^3}^\oplus
\wh{H}_\star(\vxi)\,d^3\vxi\,\fourier\,U-(\gamma/2)\,|\V{x}|^{-1}-
(\gamma/2)\,\SA\,|\V{x}|^{-1}\SA\,.
$$

\begin{proof}[Proof of Theorem~\ref{thm-binding-NP}]
Let $\rho\in(0,1]$, $\ps$, and $\vp_\star$ be as above.
We shall employ the following
bound proven in \cite{KMS2009a}:
For all {\em real-valued} $\vp_1\in H^{1/2}(\RR^3)$, $\|\vp_1\|=1$,
\begin{align}\nonumber
\SPB{\wh{\vp}_1\otimes\vp_\star&}{
\int_{\RR^3}^\oplus|\FD_\star(\vxi)|\,d^3\vxi\,\wh{\vp}_1\otimes\vp_\star}
\\
&\klg\label{retno1}
\SPb{\vp_1}{\big(\sqrt{1-\Delta}-1\big)\,\vp_1}
+
\SPb{\vp_\star}{|\FD(\ps)|\,\vp_\star}\,.
\end{align}
Moreover, we estimate trivially
\begin{equation}\label{retno2}
(\gamma/2)\SPb{U^*\,\vp_1\otimes\vp_\star}{\SA\,|\V{x}|^{-1}\,\SA\,
U^*\,\vp_1\otimes\vp_\star}\grg0\,.
\end{equation}
In Lemma~\ref{le-retno} 
we show that, for every {\em real-valued} 
$\vp_1\in C_0^\infty(\RR^3)$, $\|\vp_1\|=1$,
\begin{align}\nonumber
\Big|\SPB{\wh{\vp}_1\otimes\vp_\star}{\int_{\RR^3}^\oplus
\big(\wh{S}_\star(\vxi)\,d\Gamma(\vo)\,\wh{S}_\star(\vxi)
-\wh{S}(\ps)\,&d\Gamma(\vo)\,\wh{S}(\ps)\big)\,d^3\vxi\,
\wh{\vp}_1\otimes\vp_\star}\Big|
\\\label{retno3}
&\klg C(d)\,\|\nabla\vp_1\|_{L^2}^2.
\end{align}
Combining \eqref{retno1}--\eqref{retno3} we obtain
\begin{align}\nonumber
\Th(\vo,\V{G})&+\rho-E_\gamma(\vo,\V{G})
\\\nonumber
&\grg\SPn{\vp_\star}{\wh{H}(\ps)\,\vp_\star}-
\SPb{U^*\vp_1\otimes\vp_\star}{\FNP{\gamma,\vo,\V{G}}\,U^*\vp_1\otimes\vp_\star}
\\\nonumber
&\grg
-\SPb{\vp_1}{\big(\sqrt{1-\Delta}-1-(\gamma/2)\,|\V{x}|^{-1}\big)\,\vp_1}
-C(d)\,\|\nabla\vp_1\|^2/2
\\\label{retno4}
&\grg(\gamma/2)\,\SPn{\vp_1}{|\V{x}|^{-1}\vp_1}
-(1+C(d))\,\|\nabla\vp_1\|^2/2\,.
\end{align}
In the last step we used $\sqrt{1+t}-1\klg t/2$, $t\grg0$.
We pick some
$\theta\in C_0^\infty(\RR,[0,1])$ with
$\theta\equiv1$ on $[-1,1]$, $\theta\equiv0$
on $\RR\setminus(-2,2)$ and set
$$
\vp_1(\V{x}):=\frac{1}{\sZ^{1/2}\,R^{3/2}}\,\theta(|\V{x}|/R)\,,
\;\;\,\V{x}\in\RR^3,\quad
\sZ:=\int_{\RR^3}\theta^2(|\V{x}|)\,d^3\V{x}\,,
$$
for some $R\grg1$.
Then it is straightforward to see that the first, positive term
in the last line of \eqref{retno4} behaves like $R^{-1}$
whereas the second term is some $\bigO(R^{-2})$.
Hence, choosing $R$ sufficiently large, depending only on $\gamma$
and $d$, we obtain 
$\Th(\vo,\V{G})+\rho-E_\gamma(\vo,\V{G})\grg c(\gamma,d)>0$, where
$\rho\in(0,1]$ is arbitrarily small. 
\end{proof}

\smallskip

\noindent
It remains to prove the bound \eqref{retno3}.
For this we need, however, a few preparations.
In what follows we set
$$
\wh{R}_{\vxi}(iy):=(\FD(\vxi+\ps)-iy)^{-1},\quad\vxi\in\RR^3,\quad
\wh{R}(iy):=\wh{R}_{\V{0}}(iy)\,,\quad
y\in\RR\,,
$$
so that, analogously to \eqref{sgn},
\begin{equation}\label{sgn-xi}
\wh{S}_\star(\vxi)\,\psi=\lim_{\tau\to\infty}\int_{-\tau}^\tau
\wh{R}_{\vxi}(iy)\,\psi\,\frac{dy}{\pi}\,,\quad
\psi\in\CC^4\otimes\Fock\,.
\end{equation}

\begin{lemma}\label{le-zita} 
There is some $K\equiv K(d)\in(0,\infty)$,
and, for all $\vxi\in\RR^3$, $\nu\in(0,1]$, and $y\in\RR$, 
we can construct $\wh{\Upsilon}_{\vxi}(y)\in\LO(\CC^4\otimes\Fock)$, 
$\|\wh{\Upsilon}_{\vxi}(y)\|\klg2$,
such that
\begin{equation}\label{zita11}
\wh{R}_{\vxi}(iy)\,(d\Gamma(\vo)+K)^{-1/2}=
(d\Gamma(\vo)+K)^{-1/2}\,\wh{R}_{\vxi}(iy)\,\wh{\Upsilon}_{\vxi}(y)\,.
\end{equation}
\end{lemma}

\begin{proof}
We set $\Theta:=d\Gamma(\vo)+K$.
Due to \cite[Lemma~3.1]{MatteStockmeyer2009a} we know that
$[\valpha\cdot\V{A}(\V{0}),\Theta^{-1/2}]\,\Theta^{1/2}$
extends to a bounded operator on $\CC^4\otimes\Fock$,
henceforth denoted by $Z$, and
$\|Z\|\klg C(d)/K^{1/2}$.
We choose $K$ so large that $\|Z\|\klg1/2$.
Then we readily infer 
(compare \cite[Corollary~3.1]{MatteStockmeyer2009a})
that
$\Theta^{-1/2}\,\wh{R}_{\vxi}(iy)=\wh{R}_{\vxi}(iy)\,\Theta^{-1/2}
(\id-Z^*\,\wh{R}_{\vxi}(iy))$.
Since $\|\wh{R}_{\vxi}(iy)\|\klg1$, the assertion
follows with 
$\wh{\Upsilon}_{\vxi}(y):=(\id-Z^*\,\wh{R}_{\vxi}(iy))^{-1}$.
\end{proof}

\begin{lemma}\label{le-retno}
The bound \eqref{retno3} holds true. 
\end{lemma}

\begin{proof}
We put 
$\delta\wh{S}:=\int^\oplus_{\RR^3}(\wh{S}_\star(\vxi)-\wh{S}(\ps))\,d^3\vxi$. 
Then
the LHS of \eqref{retno3} equals $|2\Re I_1+I_2^2|$ with
$$
I_1:=\SPb{\delta\wh{S}\,\wh{\vp}_1\otimes\vp_\star}{
d\Gamma(\vo)\,\wh{S}(\ps)\,\wh{\vp}_1\otimes\vp_\star}\,,\;\;
I_2:=\big\|d\Gamma(\vo)^{1/2}\,
\delta\wh{S}\,\wh{\vp}_1\otimes\vp_\star\big\|\,.
$$
Notice that the operator $\wh{S}(\ps)$ acts only on 
$\CC^4\otimes\Fock=\Fock^4$.
By virtue of \eqref{sgn-xi} and a two-fold application of the
second resolvent identity we thus obtain
\begin{align*}
&I_1=-\int\limits_\RR\!\int\limits_{\RR^3}|\wh{\vp}_1(\vxi)|^2\,\vxi\cdot
\SPb{\wh{R}(iy)\,\valpha\,\wh{R}(iy)\,\vp_\star}{d\Gamma(\vo)\,
\wh{S}(\ps)\,\vp_\star}_{\Fock^4}\,
{d^3\vxi}\,\frac{dy}{\pi}
\\
&+
\int\limits_\RR\!\int\limits_{\RR^3}|\wh{\vp}_1(\vxi)|^2
\SPb{\wh{R}(iy)\valpha\cdot\vxi\wh{R}_{\vxi}(iy)
\valpha\cdot\vxi\wh{R}(iy)\vp_\star}{d\Gamma(\vo)
\wh{S}(\ps)\vp_\star}_{\Fock^4}
{d^3\vxi}\frac{dy}{\pi}.
\end{align*} 
Since $\vp_1$ is real-valued its Fourier transform
satisfies $|\wh{\vp}_1(\vxi)|=|\wh{\vp}_1(-\vxi)|$.
Substituting $\vxi\to -\vxi$ we thus observe that
the integral in the first line of the above formula for $I_1$
is equal to zero.
A straightforward application of Lemma~\ref{le-zita}
to the integral in the second line using 
$\|\wh{R}_{\vxi}(iy)\|\klg(1+y^2)^{-1/2}$
then yields
$$
|I_1|\klg\int_\RR\frac{(8/\pi)\,dy}{(1+y^2)^{3/2}}\,
\big\|(d\Gamma(\vo)+K)^{1/2}\vp_\star\big\|\,
\big\|d\Gamma(\vo)^{1/2}\,\wh{S}(\ps)\,\vp_\star\big\|
\int_{\RR^3}|\vxi\,\wh{\vp}_1|^2.
$$
Denoting
the set of all normalized 
$\wh{\psi}\in\form(\id\otimes d\Gamma(\vo))$ by ${\sf S}$, we further have
\begin{align*}
|I_2|&\klg\sup_{\wh{\psi}\in{\sf S}}
\big|\SPb{d\Gamma(\vo)^{1/2}\,\wh{\psi}}{
\delta\wh{S}\,\wh{\vp}_1\otimes\vp_\star}\big|
\\
&\klg\sup_{\wh{\psi}\in{\sf S}}
\int_\RR\big|\SPb{d\Gamma(\vo)^{1/2}\,\wh{\psi}}{
\wh{R}_{\vxi}(iy)\,(\id\otimes\valpha)\cdot(\vxi\,
\wh{\vp}_1)\otimes(\wh{R}(iy)\,\vp_\star)\,}\big|\,\frac{dy}{\pi}\,.
\end{align*}
Applying Lemma~\ref{le-zita} once more we deduce that
\begin{align*}
|I_2|&\klg\frac{4}{\pi}\int_\RR\frac{dy}{1+y^2}\,
\big\|(d\Gamma(\vo)+K)^{1/2}\,\vp_\star\big\|\,\|\vxi\,\wh{\vp}_1\|_{L^2}\,.
\end{align*}
Next, we observe that 
$\SPn{\vp_\star}{d\Gamma(\vo)\,\vp_\star}/2
+\SPn{\wh{S}(\ps)\,\vp_\star}{d\Gamma(\vo)\,\wh{S}(\ps)\,\vp_\star}/2
\klg\Th(\vo,\V{G})-1+\rho$, 
where $\rho\klg1$.
By Remark~\ref{rem-Th} we have a finite
upper bound on $\Th(\vo,\V{G})$ depending only on $d$.
Since the LHS of \eqref{retno3} is $\klg2|I_1|+|I_2|^2$ this
concludes the proof.
\end{proof}


\section{Existence of ground states for massive photons}
\label{sec-NPm}

\noindent
In this section we prove that the no-pair operator
defined by means of the physical choices $\vo=\omega$ and
$\V{G}=\V{G}^{e,\UV}$ given in Example~\ref{ex-Gphys}
has ground state eigenvectors, provided that the photons
are given a mass. The photon mass, $m>0$, is introduced
as follows:

\subsection{Introduction of a photon mass}

\noindent
As the underlying Hilbert space we choose $\HR_m$.
We let $\omega_m$ and $\V{G}^{e,\UV}_{m,\V{x}}$ denote
the restrictions of $\omega$ and $\V{G}_\V{x}^{e,\UV}$
to $\cA_m\times\ZZ_2$ with $\cA_m=\{|\V{k}|\grg m\}$,
respectively, and set
\begin{equation}\label{def-FNPo}
\FNPmo{\gamma}:=\FNP{\gamma,\omega_m,\V{G}^{e,\UV}_m}\,,
\qquad \gamma\in[0,\gcnp]\,,\;m>0\,.
\end{equation}
In order to show that
$\FNPmo{\gamma}$ has ground state eigenvectors
we compare $\FNPmo{\gamma}$ with a modified version of it
where all Fock space operators are discretized.
This strategy is also used in \cite{BFS1998b,BFS1999}.
We point out, however, that the
proof of Theorem~\ref{prop-gs-NPm} below contains
a new idea which allows to deal with arbitrarily large values
of $e$ and $\UV$.

\subsection{Discretization of the photon momenta}
\label{ssec-NPe}

\noindent
Let $m>0$ be fixed and let $\ve>0$.
We decompose $\cA_m$ as
$$
\cA_m=\!\!\bigcup_{\vnu\in(\ve\ZZ)^3}\!\!Q^\ve_m(\vnu)\,,
\quad Q^{\ve}_m(\vnu):=\big(\vnu+[-\ve/2\,,\,\ve/2)^3\big)
\cap\cA_m\,,\;\;
\vnu\in(\ve\ZZ)^3.
$$
For every $\V{k}\in\cA_m$, we find a unique
vector, ${\vmu}^\ve(\V{k})\in(\ve\ZZ)^3$, such that
$\V{k}\in Q^\ve_m({\vmu}^\ve(\V{k}))$,
and we put 
\begin{equation}\label{mulgrew0}
\omega_m^\ve(k):=|\vmu^\ve(\V{k})|\,,\quad k=(\V{k},\lambda)\in
\cA_m\times\ZZ_2\,,
\qquad\He:=d\Gamma(\omega_m^\ve)\,,
\end{equation}
so that
\begin{equation}\label{mulgrew1}
|\omega_m-\omega_m^\ve|\klg\sqrt{3}\,\ve/2\klg(\sqrt{3}\,\ve/2m)\,\omega_m\,.
\end{equation}
We further define an $\ve$-average of 
$f\in\HP_m$ by
\begin{equation}\label{def-Peps}
P^\ve_m\,f:=
\sum_{{\vnu\in(\ve\ZZ)^3:\atop Q^\ve_m(\vnu)\not=\varnothing}}
\SPn{\chi_{Q^\ve_m(\vnu)}}{f}\,\chi_{Q^\ve_m(\vnu)}\,,
\end{equation}
where $\chi_{Q^\ve_m(\vnu)}$ denotes the normalized characteristic
function of the set $Q^\ve_m(\vnu)$, so that
$P^\ve_m$ is an orthogonal projection in $\HP_m$.
Finally, we set
\begin{align}
\V{G}^{e,\UV,\ve}_{m,\V{x}}&:=e^{-i\vmu^\ve\cdot\V{x}}
P^\ve_m[\V{G}_{m,\V{0}}^{e,\UV}],\qquad 
\FNPme{\gamma}:=\FNP{\gamma,\omega_m^\ve,\V{G}^{e,\UV,\ve}_{m}}\,.
\end{align}
It is an easy and well-known exercise to verify that
\begin{equation}\label{mulgrew3}
\int\Big(1+\frac{1}{\omega_m(k)}\Big)\,\sup_{\V{x}}
e^{-a|\V{x}|}\big|\V{G}^{e,\UV,\ve}_{m,\V{x}}(k)-\V{G}^{e,\UV}_{m,\V{x}}(k)\big|^2
dk\klg c_{a,m}(\ve)\,,
\end{equation}
where $c_{a,m}(\ve)\to0$, $\ve\searrow0$, for all fixed
$a,m>0$. Notice that some $\V{x}$-dependent weights are required
in the above estimate since we use the bound
$$
|e^{-i\V{k}\cdot\V{x}}-e^{-i\vmu^\ve(\V{k})\cdot\V{x}}|
\klg|\V{k}-\vmu^\ve(\V{k})|\,|\V{x}|\klg
\sqrt{3}\,\ve\,|\V{x}|/2\,.
$$

\subsection{Discrete and fluctuating subspaces}

\noindent
In the proof of the main result of this section,
Theorem~\ref{prop-gs-NPm}, we employ a certain
tensor product representation of $\HR_m$ we
shall explain first.

We introduce the subspaces of discrete and fluctuating
photon states,
$$
\HP_m^d:=P^\ve_m\HP_m\,,\qquad \HP_m^f:=(\id-P_m^\ve)\,\HP_m\,,
$$
where $P^\ve_m$ is defined in \eqref{def-Peps}.
Corresponding to the orthogonal decomposition
$\HP_m=\HP_m^d\oplus\HP_m^f$ there is an isomorphism
of Fock spaces, 
$\Fock[\HP_m]=\Fock[\HP_m^d]\otimes\Fock[\HP_m^f]$. 
If $\{g_i\}$ and $\{h_j\}$ denote orthonormal bases
of $\HP_m^d$ and $\HP_m^f$, respectively,
then this isomorphism maps
$\ad(g_{i_1})\dots\ad(g_{i_r})\,\ad(h_{j_1})\dots
\ad(h_{j_s})\Omega$
to
$
\ad(g_{i_1})\dots\ad(g_{i_r})\,\Omega_d\otimes
\ad(h_{j_1})\dots
\ad(h_{j_s})\,\Omega_f
$,
where $\Omega_\ell$ is the vacuum vector in
$\Fock[\HP_m^\ell]$, for $\ell=d,f$.
Let $\V{A}_m^\ve:=\V{A}[\V{G}_{m}^{e,\UV,\ve}]$
be defined by the formula \eqref{def-Aphys}
and let $\id_f$ denote the identity in $\Fock[\HP_m^f]$.
Since $\V{G}_{m,\V{x}}^{e,\UV,\ve}\in\HP_m^d$,
for every $\V{x}$, it follows that
$$
\V{A}_m^\ve=\V{A}_m^{\ve,d}\otimes\id_f\,,\quad
\V{A}_m^{\ve,d}:=
\int_{\RR^3}^\oplus
\id_{\CC^4}\otimes\big(\underbrace{\ad(\V{G}_{m,\V{x}}^{e,\UV,\ve})+
a(\V{G}_{m,\V{x}}^{e,\UV,\ve})}_{\textrm{acting in}\:\Fock[\HP_m^d]}\big)\,d^3\V{x}\,,
$$
corresponding to the isomorphism
\begin{equation}\label{mia0}
\HR_m=\big(L^2(\RR^3,\CC^4)\otimes\Fock[\HP_m^d]\big)
\otimes\Fock[\HP_m^f]\,.
\end{equation}
We infer that the Dirac
operator and all functions of it can be written as
$$
\D{\V{A}_m^\ve}=\D{\V{A}_m^{\ve,d}}\otimes\id_f\,,\quad
|\D{\V{A}_m^\ve}|=|\D{\V{A}_m^{\ve,d}}|\otimes\id_f\,,\quad
S_{\V{A}_m^\ve}=S_{\V{A}_m^{\ve,d}}\otimes\id_f\,.
$$
Since $\omega_m^\ve$ commutes with $P_m^\ve$ 
and $(S_{\V{A}_m^{\ve,d}})^2=\id$ we conclude that
\begin{align*}
\FNPme{\gamma}&=\FNPmed{\gamma}\otimes\id_f+\id\otimes\Hef\,,
\\
\FNPmed{\gamma}&:=|\D{\V{A}_m^{\ve,d}}|
+\frac{1}{2}\,\Big(-\frac{\gamma}{|\V{x}|}+\Hed\Big)
+\frac{1}{2}\,S_{\V{A}_m^{\ve,d}}\Big(-\frac{\gamma}{|\V{x}|}
+\Hed\Big)\,S_{\V{A}_m^{\ve,d}}\,,
\end{align*}
where $\Hed:=d\Gamma(\omega_m^\ve\,P_m^\ve)$
and $\Hef:=d\Gamma(\omega_m^\ve\,(\id-P_m^\ve))$.
Tensor-multiplying minimizing sequences for
$\FNPmed{\gamma}$ with $\Omega_f$ we finally verify that
\begin{equation}\label{mia1}
\inf\spec[\FNPme{\gamma}]=\inf\spec[\FNPmed{\gamma}]\,.
\end{equation}

\subsection{Existence of ground states with photon mass}

\begin{theorem}\label{prop-gs-NPm}
Let $e\in\RR$, $\UV>0$, $\gamma\in(0,\gcnp)$, and $m>0$. Then
the spectral projection $\id_{E_m+m/4}(\FNPmo{\gamma})$
has a finite rank.
\end{theorem}

\begin{proof}
We pick some null sequence $\ve_n\searrow0$ and apply
Proposition~\ref{prop-conv} with
$\vo:=\omega_m$, $\V{G}_\V{x}:=\V{G}_{m,\V{x}}^{e,\UV}$
and $\vo_n:=\omega_m^{\ve_n}$,
$\V{G}_{n,\V{x}}:=\V{G}_{m,\V{x}}^{e,\UV,\ve_n}$, that is,
$$
H:=\FNPmo{\gamma}\,,\quad 
H_n:=H_{\gamma,m}^{\ve_n}\,,\quad
E:=\inf\spec[H]\,,\quad E_n:=\inf\spec[H_n]\,.
$$
On account of \eqref{mulgrew0}, \eqref{mulgrew1}, and \eqref{mulgrew3}
the assumptions of Proposition~\ref{prop-conv} are satisfied,
for every fixed $m>0$.
By Theorem~\ref{thm-binding-NP} we have a uniform, strictly
positive lower bound on the binding energy of $H_n$, $n\in\NN$,
so that $E_n\to E$ by Proposition~\ref{prop-conv}(4). 
By virtue of Proposition~\ref{prop-conv}(2) 
we further know that, for all sufficiently large $n$
with $E+3m/8\klg E_n+m/2$,
$$
\dim\Ran\big(\id_{E+m/4}(H)\big)
\klg
\dim\Ran(\Pi_n)\,,
\quad\Pi_n:=\id_{E_n+m/2}(H_n)\,.
$$
It remains to show that $\Pi_n$ is a finite rank projection,
for all sufficiently large $n$. To this end
we employ the isomorphism \eqref{mia0} explained
in the previous subsection.
We denote the projection in $\Fock[\HP_m^f]$ onto the 
vacuum sector by $P_{\Omega_f}$, write 
$P^\bot_{\Omega_f}:=\id_f-P_{\Omega_f}$, and
set $H_n^d:=H_{\gamma,m}^{\ve_n,d}$.
In view of $\Hef\,P_{\Omega_f}=0$ we obtain
\begin{align}\nonumber
-\frac{m}{2}\,\Pi_n&\grg
\Pi_n\,(H_n-E_n-m)\,\Pi_n
\\\nonumber
&=
\Pi_n\,(\id\otimes P_{\Omega_f})\,(H_n-E_n-m)(\id\otimes P_{\Omega_f})\,\Pi_n
\\\nonumber
&\quad
+\Pi_n\,\big\{(H_n^d-E_n)\otimes P^\bot_{\Omega_f}
+\id\otimes(\Hef-m)\,P^\bot_{\Omega_f}\,\big\}\,\Pi_n\,.
\end{align}
The operator in the last line is non-negative since
$H_n^d-E_n\grg0$ by \eqref{mia1}, 
and $\Hef\,P^\bot_{\Omega_f}\grg m\,P^\bot_{\Omega_f}$.
In order to bound the term in the second line
from below we use \eqref{iris3}, that is, 
$H_n\grg(-\Delta)^s+(1/2)\,H_{\mathrm{f},m}^{\ve_n}-C$,
for some $s\in(0,1/2)$ and $C\equiv C(e,\UV,s)\in(0,\infty)$.
Furthermore, we observe that 
$T:=\Pi_n\,(|\V{x}|^2\otimes P_{\Omega_f})\,\Pi_n$ is bounded uniformly
in $n\in\NN$ by Proposition~\ref{prop-exp-loc},
Remark~\ref{rem-Th}, and our uniform lower
bound on the binding energy of $H_n$.
Choosing $\delta>0$ so small that $\delta\,T\klg(m/4)\,\Pi_n$
we arrive at
\begin{align}\label{valerie1}
-\frac{m}{4}\,\Pi_n&\grg
\Pi_n\,\Big\{(-\Delta)^s+\delta\,|\V{x}|^2+
\frac{1}{2}\,H_{\mathrm{f},m}^{\ve_n,d}
-E_n-m-C\Big\}
\otimes
P_{\Omega_f}\,\Pi_n\,.
\end{align}
Now, both $(-\Delta)^s+\delta\,|\V{x}|^2$ and $\Hed$ have purely discrete
spectrum as operators on the electron and photon Hilbert
spaces and $P_{\Omega^f}$ has rank one. 
(Recall that $\Hed=d\Gamma(\omega_m^\ve\,P_m^\ve)$ 
and $\omega_m^\ve\,P_m^\ve$, {\em as an operator in $\HP_m^d$},
has purely discrete spectrum and is strictly positive.)
Let $X\grg0$ denote the negative part of the operator $\{\cdots\}$
in \eqref{valerie1}. Then $X\otimes\,P_{\Omega_f}$
has a finite rank and 
$\Pi_n\klg(4/m)\,\Pi_n\,(X\otimes\,P_{\Omega_f})\,\Pi_n$.
Therefore, $\Pi_n$ is a finite rank projection,
if $n$ is sufficiently large. 
\end{proof}


\section{Existence of ground states}
\label{sec-ex}

\noindent
We extend $\FNPmo{\gamma}$ (defined in \eqref{def-FNPo})
to an operator acting in $\HR=\HR_0$
by setting
\begin{align}\label{def-Gm}
\V{G}_{m,\V{x}}^{e,\UV}(k)&:=\id_{\cA_m}(\V{k})\,\V{G}^{e,\UV}_{\V{x}}(k)\,,
\quad\textrm{almost every}\;\,k=(\V{k},\lambda)\in\RR^3\times\ZZ_2\,,
\\\label{def-Hm}
\FNPm{\gamma}&:=\FNP{\gamma,\omega,\V{G}_m^{e,\UV}}\,.
\end{align}
(We are abusing the notation slightly since the symbol
$\V{G}_{m,\V{x}}^{e,\UV}$ used to denote the restriction
of $\V{G}^{e,\UV}_{\V{x}}$ to $\cA_m\times\ZZ_2$.
From now on it denotes a function on $\RR^3\times\ZZ_2$.)
The splitting $\HP_0=\HP_m\oplus\HP_m^\bot$ gives
rise to an isomorphism
$$
\HR=\HR_m\otimes\Fock[\HP_m^\bot]\,,
$$
and with respect to this isomorphism we have
$$
\FNPm{\gamma}=\FNPmo{\gamma}\otimes\id
+\id\otimes 
d\Gamma(\omega\!\!\upharpoonright_{(\RR^3\setminus\cA_m)\times\ZZ_2})\,.
$$
By Theorem~\ref{prop-gs-NPm} $\FNPmo{\gamma}$ has a 
normalized ground state
eigenvector, $\phi_m^0$, and we readily infer that
$\phi_m:=\phi_m^0\otimes\Omega_\bot$ is a normalized ground state
eigenvector of $\FNPm{\gamma}$, where $\Omega_\bot$ denotes the
vacuum vector in $\Fock[\HP_m^\bot]$.
In what follows we represent $\phi_m$ as
\begin{equation}\label{phimn}
\phi_m=(\phi_m^{(n)})_{n=0}^\infty
\in\bigoplus_{n=0}^\infty L^2(\RR^3\times\ZZ_4)\otimes\Fock^{(n)}[\HP_0]\,.
\end{equation}
The aim of this section is to show that each
sequence $\{\phi_{m_j}\}$, $m_j\searrow0$,
contains a strongly convergent subsequence.
By virtue of Proposition~\ref{prop-conv} the limit
of such a subsequence then turns out to be a ground state
eigenvector of 
$$
H_{\gamma}:=H_{\gamma,0}:=\FNP{\gamma,\omega,\V{G}^{e,\UV}}\,.
$$
As in \cite{KMS2009a}
we shall prove this compactness property by a suitably adapted version
of an argument from \cite{GLL2001}. For this purpose we need
the two infra-red bounds stated in the following
proposition. Their proofs are deferred to Section~\ref{sec-IR-bounds}.
We recall the notation
$$
(a(k)\,\psi)^{(n)}(k_1,\dots,k_n)\,=\,
(n+1)^{1/2}\,\psi^{(n+1)}(k,k_1,\dots,k_n)\,,\quad n\in\NN_0\,,
$$
almost everywhere, for $\psi=(\psi^{(n)})_{n=0}^\infty\in\Fock[\HP_0]$,
and $a(k)\,\Omega=0$.

\begin{proposition}[{\bf Infra-red bounds}]\label{prop-IR}
Let $e\in\RR$, $\UV>0$, and $\gamma\in(0,\gcnp)$.
Then there is some $C\in(0,\infty)$,
such that, for all $m\in[0,\UV)$ and every normalized
ground state eigenvector, $\phi_m$, of $\FNPm{\gamma}$,
we have the {\em soft photon bound},
\begin{equation}\label{eq-spb}
\big\|\,a(k)\,\phi_m\,\big\|^2\,\klg\,\id_{\{m\klg|\V{k}|\klg\UV\}}
\:\frac{C}{|\V{k}|}\,,
\end{equation}
for almost every $k=(\V{k},\lambda)\in\RR^3\times\ZZ_2$, and
the {\em photon derivative bound},
\begin{equation}\label{pdb-kp}
\big\|\,a(\V{k},\lambda)\,\phi_m-a(\V{p},\lambda)\,\phi_m\,\big\|
\,\klg\,
C\,|\V{k}-\V{p}|\,
\Big(\frac{1}{|\V{k}|^{1/2}|\V{k}_\bot|}+
\frac{1}{|\V{p}|^{1/2}|\V{p}_\bot|}\Big)\,,
\end{equation}
for almost every $\V{k},\V{p}\in\RR^3$
with $m<|\V{k}|<\UV$, $m<|\V{p}|<\UV$, and $\lambda\in\ZZ_2$.
(Here we use the notation introduced in \eqref{def-kbot}.)
In particular, 
\begin{equation}\label{eq-spbN}
\sup_{m\in(0,\UV)}\sum_{n=1}^\infty n\,\|\phi_m^{(n)}\|^2<\infty\,.
\end{equation}
\end{proposition}

\smallskip

\noindent
The proof of \eqref{pdb-kp} is actually the only place in the whole article
where the special choice of the polarization vectors \eqref{pol-vec}
is used explicitly.

\begin{remark}\label{rem-IR-gc}
Assume that $\phi_{\gamma,m}$ is some normalized ground state
eigenvector of $H_{\gamma,m}$, for all $m\in[0,\UV)$
and every $\gamma\in(0,\gcnp)$.
If we find {\em $\gamma$-independent}
$a,C'\in(0,\infty)$ such that
\begin{equation}\label{exp-gc}
\forall\,m\in[0,\UV)\,,\:\gamma\in(0,\gcnp)\::\quad
\big\|e^{a|\V{x}|}\phi_{\gamma,m}\big\|\klg C'\,,
\end{equation}
then the constant $C$ appearing in the statement of Proposition~\ref{prop-IR}
can be chosen independently of $\gamma\in(0,\gcnp)$, too. 
This remark shall be important in order to prove the existence
of ground states at critical coupling ($\gamma=\gcnp$)
in a forthcoming note by two of the present authors.
(Due to lack of space \eqref{exp-gc} cannot be derived in the present article.)
A brief explanation of this remark is given at the end of 
Subsection~\ref{ssec-IR-proof}.
\end{remark}

\smallskip

\begin{proof}[Proof of Theorem~\ref{thm-ex-NP}]
Let $\{m_j\}$, $m_j\searrow0$, be some null sequence
and let $\phi_{m_j}$ denote some normalized ground state
eigenvector of $\FNPmj{\gamma}$, whose existence
is guaranteed by the remarks at the beginning of this section.
Passing to some subsequence if necessary, we may assume that
$\{\phi_{m_j}\}$ converges weakly to some $\phi\in\HR$.
It suffices to show that $\phi\not=0$.

In fact, if we set $\vo_j:=\vo:=\omega$, $\V{G}:=\V{G}^{e,\UV}$,
and $\V{G}_j:=\V{G}^{e,\UV}_{m_j}$, for $j\in\NN$,
then the assumptions of Proposition~\ref{prop-conv} are
obviously fulfilled. Since Theorem~\ref{thm-binding-NP}
provides a uniform, strictly positive lower bound
on the binding energy of $\FNPmj{\gamma}$, $j\in\NN$, 
Parts~(4) and~(5) of that proposition
are available with
$$
H:=H_{\gamma}\,,\quad H_j:=\FNPmj{\gamma}\,,\quad
E:=\inf\spec[H]\,,\quad E_j:=\inf\spec[H_j]\,.
$$
In particular, $\phi\in\dom(H)$ and 
$H\,\phi=E\,\phi$.
On account of \eqref{def-FNP} and \eqref{H+H-}
this proves Theorem~\ref{thm-ex-NP}, if $\phi\not=0$.

In what follows
we only sketch how to prove that
$\{\phi_{m_j}\}$ converges actually strongly to $\phi$
along some subsequence, so that $\|\phi\|=1$. 
For this proof is
almost literally the same as the one of \cite[Theorem~2.2]{KMS2009a},
which in turn
is based on the same ideas as
the corresponding proof in \cite{GLL2001}. Only different compact imbedding
theorems have to be employed since we have weaker bounds
on (fractional) derivatives of $\phi_m$ with respect to the
electron coordinates than in the non-relativistic case; 
see \eqref{lb-MS-phim} below.

\eqref{eq-spbN} shows that the largest portion of 
$\phi_m$ belongs to Fock space sectors with low particle numbers
so that the norm of 
$(0,\ldots,0,\phi_m^{(n_0)},\phi_m^{(n_0+1)},\ldots\;)$ is small,
for large $n_0\in\NN$, {\em uniformly} in small $m>0$.
Moreover, \eqref{eq-spb} shows that the functions $\phi_m^{(n)}$
-- which are symmetric in the photon variables --
are localized with respect to the photon momenta, again
uniformly in small $m>0$.
The photon derivative bound provides uniform bounds on
the weak first order derivatives w.r.t. the
photon momenta in $L^p$, for {\em every} $p<2$; see, e.g.,
\cite[\textsection4.8]{Nikolskii}. 
Similar information is available also with respect
to the electron coordinates:
Since Theorem~\ref{thm-binding-NP} gives a lower bound
on the binding energy of $\FNPm{\gamma}$, uniformly in $m>0$,
it is clear from Proposition~\ref{prop-exp-loc} 
and Remark~\ref{rem-Th} 
that there exist $C,a>0$ such that
\begin{equation}\label{exp-loc-phim}
\big\|\,e^{a|\V{x}|}\,\phi_m\,\big\|\klg C,\qquad m>0\,.
\end{equation}
This gives uniform localization in $\V{x}$. 
Uniform $L^2$-bounds on fractional derivatives with
respect to $\V{x}$
of order $s<1/2$ follow from \eqref{iris3}
which implies
\begin{align}
\SPn{\phi_{m_j}^{(n)}}{(-\Delta)^{s}\,\phi_{m_j}^{(n)}}
\klg\,\label{lb-MS-phim} 
\SPn{\phi_{m_j}}{H_j\,\phi_{m_j}}+c=
E_j+c \klg E+c'\,,
\end{align}
where the constants $c,c'\in(0,\infty)$ do not depend on $j,n\in\NN$.
Here we used that $E_j\to E$ due to Proposition~\ref{prop-conv}(4).
Our aim is to exploit all this information
to single out a strongly convergent subsequence
from $\{\phi_{m_j}\}$
by applying a suitable compact imbedding theorem.
Notice that we are dealing with (fractional) derivatives
of different orders in different $L^p$-spaces
which are, moreover, defined by different means
(via Fourier transformation or as weak derivatives).
The classical anisotropic function spaces
$H^{(r_1,\ldots,r_d)}_{q_1,\ldots,q_d}(\RR^{d})$ introduced by
Nikol{$'$}ski{\u\i} turn out to be convenient in this situation.
They are defined as follows:

For $r_1,\ldots,r_d\in[0,1]$ and $q_1,\ldots,q_d\grg1$, the space
$H^{(r_1,\ldots,r_d)}_{q_1,\ldots,q_d}(\RR^d)$ is equal to the intersection
$\bigcap_{i=1}^d H^{r_i}_{q_ix_i}(\RR^d)$.
For $r_i\in[0,1)$, a measurable function
$f:\RR^d\to\CC$ belongs to the class $H^{r_i}_{q_ix_i}(\RR^d)$,
if $f\in L^{q_i}(\RR^d)$ and there is some $M\in(0,\infty)$
such that
\begin{equation}\label{Nik1}
\|f(\cdot+h\,{\sf e}_i)-f\|_{L^{q_i}(\RR^d)}
\,\klg\,M\,|h|^{r_i}\,,\qquad h\in\RR\,,
\end{equation}
where ${\sf e}_i$ is the $i$-th canonical unit vector in $\RR^d$.
If $r_i=1$ then \eqref{Nik1} is replaced by
\begin{equation}\label{Nik2}
\|f(\cdot+h\,{\sf e}_i)-2f+f(\cdot-h\,{\sf e}_i)\|_{L^{q_i}(\RR^d)}\,
\klg\,M\,|h|\,,\qquad h\in\RR\,.
\end{equation}
$H^{(r_1,\ldots,r_d)}_{q_1,\ldots,q_d}(\RR^d)$ is a Banach space
with norm
$
\|f\|^{(r_1,\ldots,r_d)}_{q_1,\ldots,q_d}:=
\max_{1\klg i\klg d}\|f\|_{L^{q_i}(\RR^d)}+
\max_{1\klg i\klg d} M_i
$ ,
where $M_i$ is the infimum of all constants $M>0$
satisfying \eqref{Nik1} or \eqref{Nik2}, respectively.

For $n\in\NN$ and some fixed
$\ul{\theta}
=(\vs,\lambda_1,\dots,\lambda_n)\in\{1,2,3,4\}\times\ZZ_2^n$,
we now abbreviate 
$$
\phi_{m,\ul{\theta}}^{(n)}(\V{x},\V{k}_1,\dots,\V{k}_n)
\,:=\,\phi_m^{(n)}(\V{x},\vs,\V{k}_1,\lambda_1,\ldots,\V{k}_n,\lambda_n)\,,
$$
and similarly for the weak limit $\phi$.
For every $\delta,R>0$, we further set
$$
Q_{n,\delta}^R\,:=\,
\big\{\,(\V{x},\V{k}_1,\dots,\V{k}_n)\::\;
|\V{x}|< R-\delta\,,
\,\delta<|\V{k}_j|<\Lambda-\delta\,,\;j=1,\ldots,n\,\big\}\,.
$$
For some small $\delta>0$, we
pick some cut-off function $\chi\in C_0^\infty(\RR^{3(n+1)},[0,1])$ such that
$\chi\equiv1$ on $Q_{n,2\delta}^R$ and $\supp(\chi)\subset Q_{n,\delta}^R$
and define $\psi^{(n)}_{m,\ul{\theta}}:=\chi\,\phi^{(n)}_{m,\ul{\theta}}$.
Employing the ideas sketched in the first paragraphs of this proof
we can now argue exactly as in \cite[Proof of Theorem~2.2]{KMS2009a} 
to conclude  
that $\{\psi_{m_j,\ul{\theta}}^{(n)}\}_{j\in\NN}$
is bounded in the Nikol{$'$}ski{\u\i} space
$H^{(s,s,s,1,\dots,1)}_{2,2,2,p,\dots,p}(\RR^{3(n+1)})$, for every
$p\in[1,2)$.
We may thus apply Nikol{$'$}ski{\u\i}'s compactness theorem
\cite[Theorem~3.2]{Nikolskii1958} which implies 
that $\{\psi_{m_j,\ul{\theta}}^{(n)}\}_{j\in\NN}$
contains a subsequence which is strongly convergent
in $L^2(Q_{n,2\delta}^R)$, provided that $1-3n\,(p^{-1}-2^{-1})>0$.
For fixed $n_0\in\NN$, we may choose $p<2$ large enough such that
the latter condition is fulfilled, for
all $n=1,\ldots,n_0$. 
By finitely many repeated selections of subsequences 
we may hence assume without loss of generality that
$\{\phi_{m_j,\ul{\theta}}^{(n)}\}_{j\in\NN}$ converges strongly
in $L^2(Q_{n,2\delta}^R)$ to $\phi^{(n)}_{\ul{\theta}}$, for $n=0,\ldots,n_0$
and every choice of $\ul{\theta}$.
Since $\delta>0$ can be chosen arbitrary small
and $R>0$ and $n_0\in\NN$ arbitrary large, 
and since $\{\phi_{m_j}\}$ is localized w.r.t.
$\V{x}$ and $n$,
we can further argue that
$\{\phi_{m_j}\}$ contains a strongly convergent subsequence.
\end{proof}


\section{Infra-red bounds}
\label{sec-IR-bounds}

\noindent
In this section we prove the soft photon and photon 
derivative bounds which served as two of the main ingredients
for the compactness argument presented in
Section~\ref{sec-ex}.
In non-relativistic QED a
soft photon bound without infra-red regularizations
has been derived  first in \cite{BFS1999}.
The observation 
which made it possible to get rid of the
mild infra-red regularizations employed earlier
in \cite{BFS1998b} is that, after a suitable unitary
gauge transformation, namely the Pauli-Fierz transformation
explained in Subsection~\ref{ssec-gauge},
the quantized vector
potential attains a better infra-red behavior.
Working in the new gauge one can thus avoid the
infra-red divergent integrals that appeared in the
original gauge. For this reason the gauge invariance
of the no-pair operator becomes absolutely crucial 
to derive the results of the present section.
Photon derivative bounds have been introduced in
\cite{GLL2001} where also an alternative strategy
to prove the infra-red bounds has been proposed. 
As in our earlier companion paper \cite{KMS2009a}, where
we proved both infra-red
bounds for the semi-relativistic Pauli-Fierz operator,
our proofs rest on a 
suitable representation formula for $a(k)\,\phi_m$.

In the whole section we always assume that 
$e\in\RR$, $\UV>0$, $\gamma\in(0,\gcnp)$, $m\grg 0$, and that
$\phi_m$ is a ground state eigenvector of $\FNPm{\gamma}$,
where $\FNPm{\gamma}$ is defined in \eqref{def-Hm}.
Notice that we include the case $m=0$.
We set $\V{A}_m:=\V{A}[\V{G}^{e,\UV}_m]$; compare
\eqref{def-Aphys} and \eqref{def-Gm}.

We add one remark we shall use repeatedly later on:
Since the unitary operator $\SAm$ commutes with
$\FNPm{\gamma}$ we know that $\SAm\,\phi_m$ is 
a ground state eigenvector of $\FNPm{\gamma}$, too.
In view of Proposition~\ref{prop-exp-loc}
we thus find some $a\in(0,1/2]$ and some
$F\in C^\infty(\RR^3_\V{x},[0,\infty))$
with $F(\V{x})=a|\V{x}|$, for large $|\V{x}|$,
and $|\nabla F|\klg a$ on $\RR^3$, such that,
uniformly in $m\grg0$,
\begin{equation}\label{exp-loc-phim2}
\|e^{2F}\phi_m\|\klg C(d,a,\gamma)\,,\quad
\|e^{2F}\,\SAm\,\phi_m\|\klg C(d,a,\gamma)\,.
\end{equation} 
We shall keep the parameter $a$ and the weight function $F$ fixed
in the whole section and use them without further explanations. 
Moreover, we put
\begin{equation}\label{def-RAm}
R_{\V{A}_m}^{\pm F}(iy):=\big(\DAm\pm i\valpha\cdot\nabla F-iy\big)^{-1},
\qquad y\in\RR\,,
\end{equation}
which is the continuous extension of $e^{\pm F}\RAm{iy}\,e^{\mp F}$
and satisfies
\begin{equation}\label{bd-RAF}
\|R_{\V{A}_m}^{\pm F}(iy)\|\klg C\,(1+y^2)^{-1/2},\qquad y\in\RR\,;
\end{equation}
see \eqref{BoundRF} and \eqref{BoundRF2}.

\subsection{Pauli-Fierz transformation}\label{ssec-gauge}

\noindent
The unitary Pauli-Fierz transformation, $U$, is given as
\begin{equation*}
U:=\int_{\RR^3}^\oplus \id_{\CC^4}\otimes 
e^{i\V{x}\cdot\V{A}_m(\V{0})}\,d^3\V{x}\,.
\end{equation*}
For all $\V{x}\in\RR^3$ and almost every 
$k=(\V{k},\lambda)\in\RR^3\times\ZZ_2$, we set
\begin{equation}\label{def-wtG}
\wt{\V{G}}_{m,\V{x}}^{e,\UV}(k):=
(e^{-i\V{k}\cdot\V{x}}-1)\,\V{G}_{m,\V{0}}^{e,\UV}(k)\,,
\end{equation}
where $\V{G}_{m,\V{0}}^{e,\UV}(k)=\id_{\{|\V{k}|\grg m\}}\,\V{G}^{e,\UV}_\V{0}(k)$
and $\V{G}^{e,\UV}_\V{0}(k)$ is given by \eqref{def-Gphys}.
Then the
gauge transformed vector potential is 
\begin{equation*}
\wt{\V{A}}_m:=\V{A}_m-\id\otimes\V{A}_m(\V{0})
=\int_{\RR^3}^\oplus
\id_{\CC^4}\otimes\big(
\ad(\wt{\V{G}}_{m,\V{x}}^{e,\UV})+a(\wt{\V{G}}_{m,\V{x}}^{e,\UV})\big)\,d^3\V{x}\,.
\end{equation*}
In fact, using $[U,\valpha\cdot\V{A}_m]=0$ we deduce that
$U \DAm U^*=\DAmt$, whence
\begin{equation*}
U\,\RAm{iy}\,U^*=R_{\wt{\V{A}}_m}(iy)\,,
\qquad U\,\SAm\,U^*=\SAmt\,,\qquad
U\,|\DAm|\,U^*=|\DAmt|\,.
\end{equation*}
It is favorable to work in the new gauge since
$\wt{\V{G}}_m^{e,\UV}$ has a less singular infra-red behavior
than $\V{G}_m^{e,\UV}$. In fact, we have the elementary bound
\begin{equation}\label{laura1}
|\wt{\V{G}}_{m,\V{x}}^{e,\UV}(k)|\klg\id_{\{|\V{k}|\grg m\}}\,
\min\big\{2,\,|\V{k}|\,|\V{x}|\big\}\,|\V{G}_\V{0}^{e,\UV}(k)|\,.
\end{equation}
In order to gain an extra power of $|\V{k}|$ from
the previous estimate we have to control the
multiplication operator $|\V{x}|$ in \eqref{laura1}.
In our estimates below
this is possible thanks to the spatial localization
of $\phi_m$. We put
\begin{equation}
\wt{H}_{\gamma,m}:= U\,\FNPm{\gamma}\,U^*,\quad 
\Hft:=U\, \Hf \,U^*,\quad \wt{\phi}_m:=U\,\phi_m\,.
\end{equation}
On $U\core_0$ we have
\begin{equation}\label{dora1}
\wt{H}_{\gamma,m}=\SAmt\,\DAmt
+\frac{1}{2}\Big(\Hft-\frac{\gamma}{|\V{x}|}\Big)
+\frac{1}{2}\,\SAmt\,\Big(\Hft-\frac{\gamma}{|\V{x}|}\Big)\,\SAmt\,.
\end{equation}
We recall that, for $f\in\HP_0$,
\begin{align}\label{dora2}
a(f)\,U&=U\,\big(a(f)+i\SPn{f}{\V{G}_{m,\V{0}}^{e,\UV}}\cdot\V{x}\,\big)\,,
\\\label{dora2b}
U^*\ad(f)&=\big(\ad(f)-i\SPn{\V{G}_{m,\V{0}}^{e,\UV}}{f}\cdot\V{x}\,\big)\,U^*.
\end{align}


\subsection{Remarks on the commutator $\boldsymbol{[a^\sharp(f),\SAmt]}$}

\noindent
Let $a^\sharp$ be $a$ or $a^\dagger$ and 
$f\in\HP_0$ with $\omega^{-1/2}\,f\in\HP_0$.
The coupling function $\wt{\V{G}}^{e,\UV}_m$ satisfies
Hypothesis~\ref{hyp-G} with $\vo=\omega$. Thus, we again know 
from \cite[Lemma~3.3]{MatteStockmeyer2009a} that
$\SAmt$ maps $\dom(\Hf^{\nu})$ into itself, for $\nu\in[0,1]$.
In particular, 
$[a^\sharp(f)\,,\,\SAmt]$ is well-defined a priori on $\dom(\Hf^{1/2})$.
Using \eqref{CCR}, \eqref{sgn}, and \eqref{bd-RAF} 
it is straightforward to show that
it has an extension to an element of $\LO(\HR)$ given as
\begin{align}\label{dora3}
\ol{[a^\sharp(f)\,,\,\SAmt]}
=\pm U\!\int_\RR\RAm{iy}\,\valpha\cdot\SPn{f}{\wt{\V{G}}_{m,\V{x}}^{e,\UV}}^\sharp
\,\RAm{iy}\,\frac{dy}{\pi}\,U^*.
\end{align}
If $\sharp=\dagger$, then we choose the $+$-sign 
and the superscript $\sharp$ at the scalar product
in \eqref{dora3} denotes complex conjugation. Otherwise 
we choose $-$ and $\sharp$ has to be ignored.
By the spectral calculus, \eqref{Com0a}\&\eqref{Com1b} of the appendix,
and Lemma~\ref{le-tina2} we further know that, for all $y\in\RR$,
$\kappa\in[0,1)$, and $\sigma\in\{1/2,1\}$,
\begin{align}\label{ulf2}
\|\,|\DAm|^\kappa\,\RAm{iy}\|
&\klg C(\kappa)\,(1+y^2)^{-(1-\kappa)/2},
\\\label{ulf3}
\big\|\,|\V{x}|^{-\kappa}R_{\V{A}_m}^{\pm F}(iy)\,(\Hf+1)^{-1/2}\big\|
&\klg C(d,\kappa)\,(1+y^2)^{-(1-\kappa)/2},
\\\label{ulf4}
\|(\Hf+1)^{\sigma}\,R_{\V{A}_m}^{\pm F}(iy)\,(\Hf+1)^{-\sigma}\|
&\klg C(d)\,(1+y^2)^{-1/2}.
\end{align}
From these bounds we readily infer that the operator in \eqref{dora3}
maps $\HR$ into $\dom(|\DAmt|^\kappa)$ and $\dom(\Hft^{\sigma})$
into $\dom(|\V{x}|^{-\kappa})\cap\dom(\Hft^{\sigma})$ 
and that, uniformly in $m\grg0$,
\begin{align}\label{ulf5}
\big\|\,|\DAmt|^\kappa\,\ol{[a^\sharp(f)\,,\,\SAmt]}\,\big\|
&\klg C(\kappa)\,,
\\\label{ulf5b}
\big\|\,|\V{x}|^{-\kappa}\,\ol{[a^\sharp(f)\,,\,\SAmt]}\,(\Hft+1)^{-1/2}\big\|
&\klg C(d,\kappa)\,,
\\
\big\|(\Hft+1)^{\sigma}\,\ol{[a^\sharp(f)\,,\,\SAmt]}\,(\Hft+1)^{-\sigma}\big\|
&\klg C(d)\,,\quad\sigma\in\{1/2,1\}\,.\label{ulf5c}
\end{align}


\subsection{ A formula for $\boldsymbol{a(k)\,{\phi}_m}$}

\noindent
Our aim in the following is to derive the formula
$a(k)\,\phi_m=\Phi(k)$, for almost every $k$,
where $\Phi(k)$ is defined in \eqref{def-Phi(k)} below.
The infra-red bounds can then be easily read off from
this representation.
From now on we drop the reference to $e$, $\UV$, $\gamma$, and $m$ 
in the notation.

We fix some
$p=(\V{p},\mu)\in \RR^3\times\ZZ_2$ with $\V{p}\not=0$
and set $\omega_\V{p}=|\V{p}|$. Moreover, we pick
$\eta'\in U\,\core_0$ and 
$f\in C_0^\infty((\RR^3\setminus\{0\})\times \ZZ_2)$.
Then the eigenvalue equation 
$\wt{H}\,\wt{\phi}=E\,\wt{\phi}$ implies
\begin{align}\nonumber
\SPb{&(\wt{H}-E+\omega_\V{p})\,\eta'}{a(f)\,\wt{\phi}}
\\\nonumber
&=\SPn{\wt{H}\,\eta'}{a(f)\wt{\phi}}-\SPn{\ad(f)\,\eta'}{\wt{H}\,\wt{\phi}}
+\SPn{\eta'}{a(\omega_\V{p}\,f)\,\wt{\phi}}
\\\label{def-uj}
&=u_{1}(\eta')/2+u_{2}(\eta')/2+u_3(\eta')+u_4(\eta')/2\,,
\end{align}
where the functionals $u_j$ contain contributions from various terms
in \eqref{dora1}. They are defined in the course of the following discussion.
Using \eqref{dora2}\&\eqref{dora2b}, $[\Hf,\V{x}]=0$,
and $[\Hf,a(f)]=-a(\omega\,f)$,
we observe that
\begin{align}\nonumber
u_{1}(\eta')&:=
\SPb{U\,\Hf\,U^*\eta'}{a(f)\,U\,{\phi}}
-\SPb{U^*\ad(f)\,\eta'}{\Hf\,{\phi}}
+\SPn{\eta'}{a(\omega_\V{p}\,f)\,\wt{\phi}}
\\\label{u1}
&=-\SPb{U^*\eta'}{a\big((\omega-\omega_\V{p})\,f\big)\,{\phi}}
+i\omega_\V{p}\,\SPb{U^*\eta'}{\SPn{f}{\V{G}_{\V{0}}}\cdot\V{x}\,\phi}\,.
\end{align}
Furthermore, 
\begin{align*}
u_{2}(\eta')&:=\SPb{\SAt\,\Hft\,\SAt\,\eta'}{a(f)\,\wt{\phi}}
-\SPb{\ad(f)\,\eta'}{\SAt\,\Hft\,\SAt\,\wt{\phi}}
+\SPn{\eta'}{a(\omega_\V{p}\,f)\,\wt{\phi}}
\\
&=\SPb{\Hft\,\SAt\,\eta'}{a(f)\,\SAt\,\wt{\phi}}
-\SPb{\ad(f)\,\SAt\,\eta'}{\Hft\,\SAt\,\wt{\phi}}
+\SPn{\eta'}{a(\omega_\V{p}\,f)\,\wt{\phi}}
\\
&\quad
+\SPb{\Hft\,\SAt\,\eta'}{[\SAt\,,\,a(f)]\,\wt{\phi}}
-\SPb{[\SAt\,,\,\ad(f)]\,\eta'}{\Hft\,\SAt\,\wt{\phi}}
\,.
\end{align*}
Using a computation analogous to \eqref{u1} and
writing
$$
-\SA\,a(\omega\,f)\,\SA
=-a(\omega\,f)+\SA\,[\SA,a(\omega\,f)]\,,
$$
we arrive at
\begin{align}\nonumber
u_{2}(\eta')&=u_1(\eta')
+\SPb{\SA\,U^*\eta'}{[\SA,a(\omega\,f)]\,\phi}
\\
&\quad\label{u2}
+\SPb{\Hft^{1/2}\,\SAt\,\eta'}{\Hft^{1/2}\ol{[\SAt\,,\,a(f)]}\,\wt{\phi}}
-\SPb{\ol{[\SAt\,,\,\ad(f)]}\,\eta'}{\Hft\,\SAt\,\wt{\phi}}\,.
\end{align}
To treat the remaining terms in \eqref{def-uj} 
we pick some $\kappa\in(1/2,1)$ and set
$\nu:=1-\kappa$.
Since 
$\DAt=|\DAt|^\kappa\,\SAt\,|\DAt|^\nu$ we obtain
\begin{align}\label{u3}
u_3(\eta')&:=\SPb{[\ad(f)\,,\,\SAt\,\DAt]\,\eta'}{\wt{\phi}}
\\\nonumber
&=\SPb{|\DAt|^\nu\,\eta'}{\SAt\,|\DAt|^\kappa\,\ol{[\SAt\,,\,a(f)]}\,\wt{\phi}}
-\SPb{\eta'}{\valpha\cdot\SPn{f}{\wt{\V{G}}_\V{x}}\,\SAt\,\wt{\phi}}\,.
\end{align}
Finally, we have
\begin{align}
u_4(\eta')&:=\nonumber
-\gamma\,
\SPb{|\V{x}|^{-\nu}\,\SAt\,\eta'}{|\V{x}|^{-\kappa}\,
\ol{[\SAt\,,\,a(f)]}\,\wt{\phi}}
\\\label{u4}
&\quad-\gamma\,\SPb{|\V{x}|^{-\kappa}\,\ol{[\ad(f)\,,\,\SAt]}\,\eta'}{
|\V{x}|^{-\nu}\,\SAt\,\wt{\phi}}\,.
\end{align}
We briefly explain why $\eta'\in U\core_0$ can be replaced
by any element of $\form(\wt{H})$ in the above formulas:
On the one hand this is due to \eqref{ulf5}--\eqref{ulf5c}
and the following consequences of
\eqref{iris3}, \eqref{iris4}, and \eqref{UsefulBounds1}
(here we use $\nu<1/2$), 
\begin{align}\label{ulf6}
\big\|\,|\DA|^{\nu}\,(H+C(d,\nu))^{-1/2}\big\|&\klg1\,,
\quad
\big\|\,|\V{x}|^{-\nu}\,(H+C(d,\nu))^{-1/2}\big\|\klg1\,,
\\
\big\|\Hf^{\sigma}\SA\,(H+C(d))^{-\sigma}\big\|&\klg1\,,
\quad\sigma\in\{1/2,1\}\,.\label{ulf7}
\end{align}
Indeed,
using \eqref{ulf5}--\eqref{ulf5c} and \eqref{ulf6}\&\eqref{ulf7}
we conclude by inspection that $u_{1},\ldots,u_4$ 
extend to continuous linear functionals 
on $\form(\wt{H})$
(equipped with the form norm). 
On the other hand we show in Appendix~\ref{app-form} that
$a(f)\,\wt{\phi}\in\form(\wt{H})$.
Since $U\,\core_0$ is a form core for $\wt{H}$
this implies that the equality
\begin{equation}\label{ulf1}
\SPb{(\wt{H}-E+\omega_\V{p})\,\eta'}{a(f)\,\wt{\phi}}
=\ol{u}_{1}(\eta')/2+\ol{u}_{2}(\eta')/2+\ol{u}_3(\eta')+\ol{u}_4(\eta')/2
\end{equation}
holds, for all $\eta'\in\form(\wt{H})$.
In particular,
we may choose $\eta':=U\,{\cR}_\V{p}\,\eta$, 
for every $\eta \in \HR$, with
$$
{\cR}_\V{p}:=({H}-E+\omega_\V{p})^{-1}.
$$ 
In the next step we substitute a family, 
$\{f_{p,\epsilon}\}_{\epsilon>0}$,
of approximate delta functions for $f$
and pass to the limit $\epsilon\searrow0$.
So let $h\in C^\infty_0(\RR^3,[0,\infty))$ 
with $\supp(h)\subset \{|\V{k}|<1\}$ 
and $\int_{\RR^3}h(\V{k})\,d^3\V{k}=1$ and set 
$h_\epsilon:=\epsilon^{-3}h(\cdot/\epsilon)$.
Then we choose $f:=f_{p,\epsilon}$, 
where $f_{p,\epsilon}(k):= h_\epsilon(\V{k}-\V{p})\,\delta_{\mu,\lambda}$,
for $k=(\V{k},\lambda)\in \RR^3\times\ZZ_2$ and $\epsilon>0$.
Multiplying both sides of \eqref{ulf1}, where now
$\eta'=U{\cR}_\V{p}\,\eta$,
with some $g\in C^\infty_0((\RR^3\setminus \{0\})\times \ZZ_2,\CC)$ and
integrating with respect to $p=(\V{p},\mu)$, we obtain
\begin{equation}\label{Neu-2}
\int g(p)\,\SPb{U\,\eta}{a(f_{p,\epsilon})\,\wt{\phi}}\,dp
=\sum_{i=0}^8C_i(\epsilon)\,.
\end{equation}
Here $C_{0}(\epsilon),\ldots,C_4(\epsilon)$ contain all contributions
from $u_1$ and $u_2$, $C_5(\epsilon)$ and $C_6(\epsilon)$ 
contain those of $u_3$, and $C_7(\epsilon)$ and $C_8(\epsilon)$ 
account for $u_4$.
As $\epsilon\searrow0$, the LHS of \eqref{Neu-2} tends to
\begin{equation}\label{ulf99} 
\SPn{U\,\eta}{a(\ol{g})\,\wt{\phi}}=\SPn{\eta}{a(\ol{g})\,\phi}
-i\SPb{\eta}{
\SPn{\ol{g}}{\V{G}_{\V{0}}\cdot\V{x}}\,\phi}\,,
\end{equation}
because of $h_\epsilon\ast g\to g$ in $L^2$, 
Fubini's theorem, and \eqref{dora2}. 
The terms contained both in $u_1$ and $u_2$ give rise to
(compare \eqref{u1} and \eqref{u2})
\begin{align*}
C_{0}(\epsilon)&:=\int g(p)\,\SPb{{\cR}_\V{p}\,\eta}
{a\big((\omega_{\V{p}}-\omega)\,f_{p,\epsilon}\big)\,{\phi}}\,dp\,,
\\
C_{1}(\epsilon)&:=i\int g(p)\,\omega_{\V{p}}\,
\SPn{f_{p,\epsilon}}{\V{G}_\V{0}}\cdot\SPn{\cR_\V{p}\,\eta}{\V{x}\,\phi}\,dp\,.
\end{align*}
The remaining terms in $u_2$ are accounted for by
(compare \eqref{u2})
\begin{align*}
C_{2}(\epsilon)&:=
\frac{1}{2}\int g(p)\,\SPb{{\cR}_\V{p}\,\SA\,\eta}{
\ol{[\SA\,,\,a(\omega\,f_{p,\epsilon})]}\,{\phi}}\,dp\,,
\\
C_{3}(\epsilon)&:=
\frac{1}{2}\int g(p)\,\SPb{U\,\Hf\,
{\cR_\V{p}}\,\SA\,\eta}{
\ol{[\SAt\,,\,a(f_{p,\epsilon})]}\,\wt{\phi}}\,dp\,,
\\
C_{4}(\epsilon)&:=
\frac{1}{2}\int g(p)\,\SPb{\Hft\,\ol{[\ad(f_{p,\epsilon})\,,\,\SAt]}\,
U\,{\cR_\V{p}}\,\eta}{\SAt\,\wt{\phi}}\,dp\,.
\end{align*}
Likewise, we have
(compare \eqref{u3})
\begin{align*}
C_5(\epsilon)&:= 
\int g(p)\,\SPb{U\,S_{{\V{A}}}\,|D_{{\V{A}}}|^{\nu}\,
{\cR}_\V{p}\,\eta}
{|D_{\wt{\V{A}}}|^{\kappa}\,\ol{[S_{\wt{\V{A}}}\,,\,a(f_{p,\epsilon})]}
\,\wt{\phi}}\,dp\,,
\\
C_{6}(\epsilon)&:=
-\int g(p)\,\SPb{{\cR}_\V{p}\,\eta}{
{\valpha}\cdot\SPn{f_{p,\epsilon}}{\wt{\V{G}}_{\V{x}}}\,\SA\,{\phi}}\,dp\,,
\end{align*}
and (see \eqref{u4})
\begin{align*}
C_{7}(\epsilon)&:=-\frac{\gamma}{2}
\int g(p)\,\SPb{U\,|\V{x}|^{-\nu}\,
{\cR_\V{p}}\,\SA\,\eta}{|\V{x}|^{-\kappa}\,
\ol{[\SAt\,,\,a(f_{p,\epsilon})]}\,\wt{\phi}}\,dp\,,
\\
C_{8}(\epsilon)&:=-\frac{\gamma}{2}
\int g(p)\,\SPb{|\V{x}|^{-\kappa}\,\ol{[\ad(f_{p,\epsilon})\,,\,\SAt]}\,U\,
{\cR_\V{p}}\,\eta}{|\V{x}|^{-\nu}\,\SAt\,\wt{\phi}}\,dp\,.
\end{align*}
To discuss the RHS of \eqref{Neu-2} we start
with $C_0(\epsilon)$, which converges to zero.
(Almost the same term is treated in \cite{KMS2009a}. We repeat
its discussion for the sake of completeness.)
In fact, 
\begin{align}\nonumber
C_0(\epsilon)&\klg \epsilon\int |g|(p)\,\|{\cR}_\V{p}\|^2
\,\|\eta\|^2\,dp 
+\frac{\wt{C}_0(\epsilon)}{4\epsilon}\,,
\\\nonumber
\wt{C}_0(\epsilon)&:=
\int |g|(p)\,
\big\|a\big((\omega_\V{p}-\omega)\,f_{p,\epsilon}\big)\,{\phi}\big\|^2\,dp\,.
\end{align}
Since $|\omega_\V{p}-\omega|\klg\epsilon$ on $\supp(f_{p,\epsilon})$,
we further have
\begin{align}\nonumber
\wt{C}_0(\epsilon)&=
\int |g|(p)\,\Big\|\int(\omega_\V{p}-\omega_\V{k})\,f_{p,\epsilon}(k)\,
a(k)\,{\phi}\,dk\,\Big\|^2dp
\\\nonumber
&\klg
\epsilon^2\int |g|(p)\Big\{\int\frac{f_{p,\epsilon}(k')}{\omega(k')}\,dk'\Big\}
\int f_{p,\epsilon}(k)\,\omega(k)\,\|a(k)\,{\phi}\|^2\,dk\,dp\,.
\end{align}
Here the integral in the curly brackets $\{\cdots\}$ is bounded
by some $K\in(0,\infty)$
uniformly in $p$ as long as $\epsilon\klg\dist(0,\supp(g))/2$,
whence
$$
\wt{C}_0(\epsilon)\,\klg\,\epsilon^2\,K\int (|g|*h_\epsilon)(k)
\,\omega(k)\,\|a(k)\,{\phi}\|^2\,dk\,.
$$
Since $|g|*h_\epsilon\to|g|$ in $L^\infty$ and 
${\phi}\in\dom(\Hf^{1/2})$
we conclude that
$C_0(\epsilon)\to0$.

Next, we claim that,
by means of Fubini's theorem,
all the remaining expressions 
can be written in the form
\begin{align*}
C_{i}(\epsilon)=\sum_{\lambda\in\ZZ_2}
\int_{\RR^3}\int_{\RR^3} g(\V{p},\lambda)h_\epsilon(\V{k}-\V{p})
\V{G}_{\V{0}}(\V{k},\lambda)\cdot\V{s}_{i}(\V{k},\V{p})\,d^3\V{k}d^3\V{p},
\;i=1,\ldots,8,
\end{align*}
where the vectors $\V{s}_{i}$ are continuous on 
$(\RR^3\setminus\{0\})\times\RR^3$, so that
\begin{align}\label{ulf100}
\lim_{\epsilon\searrow0}C_{i}(\epsilon)=
\int g(k)\,
\V{G}_{\V{0}}(k)\,\V{s}_{i}(\V{k},\V{k})\,dk\,,
\quad i=1,\ldots,8\,.
\end{align}
In fact, 
using the representation \eqref{dora3}
it can be easily read off from the definitions
of $C_i(\epsilon)$ that
\begin{align*}
\V{s}_{1}(\V{k},\V{p})
&:=i|\V{p}|\,\SPn{\cR_\V{p}\,\eta}{\V{x}\,\phi}\,,
\\
\V{s}_{2}(\V{k},\V{p})
&:={|\V{k}|}\int_\RR\SPb{\cR_\V{p}\,\SA\,\eta}{
\V{T}_2(y,\V{k})\,e^F\phi}\,\frac{dy}{2\pi}\,,
\\
\V{s}_{3}(\V{k},\V{p})
&:=\int_\RR\SPb{\Hf\,\cR_\V{p}\,\SA\,\eta}{
\V{T}_3(y,\V{k})\,e^F\phi}\,\frac{dy}{2\pi}\,,
\\
\V{s}_{4}(\V{k},\V{p})
&:=\int_\RR\SPb{\V{T}_{4}(y,k)\,\Theta\,{\cR_\V{p}}\,\eta}{
e^F\SA\,{\phi}}\,\frac{dy}{2\pi}\,,
\end{align*}
where (with the notation \eqref{def-RAm} and $\Theta:=\Hf+1$)
\begin{align*}
\V{T}_2(y,\V{k})&:=
\RA{iy}\,\valpha\,e^{-i\V{k}\cdot\V{x}}\,e^{-F}\RAF{iy}\,,
\\
\V{T}_3(y,\V{k})&:=
\RA{iy}\,\valpha\,(e^{-i\V{k}\cdot\V{x}}-1)\,e^{-F}\RAF{iy}\,,
\\
\V{T}_{4}(y,k)
&:=\{\Hf\,R_{\V{A}}^{-F}(iy)\,\Theta^{-1}\}
\valpha\,(e^{i\V{k}\cdot\V{x}}-1)\,e^{-F}
\{\Theta\,R_{\V{A}}(iy)\,\Theta^{-1}\}\,.
\end{align*}
Moreover,
\begin{align*}
\V{s}_{5}(\V{k},\V{p})
&:=\int_\RR\SPb{|\DA|^\nu\,\cR_\V{p}\,\SA\,\eta}{
\V{T}_5(y,\V{k})\,e^F\phi}\,\frac{dy}{\pi}\,,
\\
\V{s}_{6}(\V{k},\V{p})
&:=-\SPb{\cR_\V{p}\,\eta}{
\valpha\,(e^{-i\V{k}\cdot\V{x}}-1)\,e^{-F}\,(e^F\,\SA\,\phi)}\,,
\\
\V{s}_{7}(\V{k},\V{p})
&:=-\frac{\gamma}{2}\int_\RR
\SPb{|\V{x}|^{-\nu}\,\cR_\V{p}\,\SA\,\eta}{
\V{T}_{7}(y,k)\,\Theta^{1/2}e^F\phi}\,\frac{dy}{\pi}\,.
\\
\V{s}_{8}(\V{k},\V{p})
&:=-\frac{\gamma}{2}\int_\RR\SPb{\V{T}_{8}(y,k)\,
\Theta^{1/2}\,{\cR_\V{p}}\,\eta}
{|\V{x}|^{-\nu}e^F\SA\,{\phi}}\,\frac{dy}{\pi}\,,
\end{align*}
where
\begin{align*}
\V{T}_{5}(y,k)
&:=\{|\DA|^\kappa\,\RA{iy}\}\,
\valpha\,(e^{-i\V{k}\cdot\V{x}}-1)\,e^{-F}\RAF{iy}\,,
\\
\V{T}_{7}(y,k)
&:=
\{|\V{x}|^{-\kappa}\RA{iy}\,\Theta^{-\frac{1}{2}}\}\,
\valpha\,(e^{-i\V{k}\cdot\V{x}}-1)\,e^{-F}
\{\Theta^{\frac{1}{2}}\RAF{iy}\,\Theta^{-\frac{1}{2}}\}\,,
\\
\V{T}_{8}(y,k)
&:=\{|\V{x}|^{-\kappa}R_{\V{A}}^{-F}(iy)\,\Theta^{-\frac{1}{2}}\}
\valpha\,(e^{i\V{k}\cdot\V{x}}-1)\,e^{-F}
\{\Theta^{\frac{1}{2}}R_{\V{A}}(iy)\,\Theta^{-\frac{1}{2}}\}\,.
\end{align*}
On account of \eqref{ulf2}--\eqref{ulf4}
all operators in curly brackets
$\{\cdots\}$ appearing in the definitions
of $\V{T}_{4}$, $\V{T}_{5}$, $\V{T}_{7}$, and $\V{T}_{8}$
are bounded and the integrals over $y$ in the definitions
of $\V{s}_i(\V{k},\V{p})$ converge absolutely.
In virtue of \eqref{iris3}, \eqref{iris4}, and \eqref{UsefulBounds1}
we further have (since $\nu<1/2$)
\begin{equation}\label{ulf101}
\big\||\DA|^\nu\cR_\V{k}\big\|\klg\frac{C'(d,\nu)}{1\wedge|\V{k}|},
\;
\big\||\V{x}|^{-\nu}\cR_\V{k}\big\|\klg 
\frac{C'(d,\nu)}{1\wedge|\V{k}|},
\;
\big\|\Hf\cR_\V{k}\big\|\klg\frac{C(d)}{1\wedge|\V{k}|}.
\end{equation}
Moreover, as a trivial consequence of \eqref{exp-loc-phim2},
\eqref{ulf6}, and $[\SA,H]=0$ we have
\begin{equation}\label{exp-loc-phim-V}
\big\|\,|\V{x}|^{-\nu} e^F\phi\big\|\klg C(d,a,\nu,\gamma)\,,
\quad 
\big\|\,|\V{x}|^{-\nu} e^F\,\SA\,\phi\big\|\klg C(d,a,\nu,\gamma)\,.
\end{equation}
Thanks to \eqref{UsefulBounds1} and \eqref{exp-loc-phim2}
we finally know that
\begin{equation}\label{exp-loc-phim-Theta}
\big\|\Theta^{1/2}e^F\phi\big\|^2\klg\|\Theta\,\phi\|\,\|e^{2F}\phi\|
\klg C(d,\gamma)\,\big(\|H\,\phi\|+1\big)=
C(d,\gamma)\,(|E|+1)\,.
\end{equation}
Recall that \eqref{ulf99} is the limit of the LHS of \eqref{Neu-2} and
$C_0(\epsilon)\to0$. Taking this, \eqref{ulf100}, and the preceding
remarks into account we see that the limit of
\eqref{Neu-2} can be written as
\begin{align}\label{ulf99b}
\int g(k)\, \SPn{\eta}{a(k)\,\phi}\,dk
&=\int g(k)\,\SPn{\eta}{\Phi(k)}\,dk\,,
\end{align}
for some function 
$(\RR^3\setminus\{0\})\times\ZZ_2\ni k\mapsto\Phi(k)\in\HR$.
Indeed, writing
$$
T_i(y,k):=\V{G}_{\V{0}}(k)\cdot\V{T}_i(y,\V{k})\,,
\quad i=2,3,4,5,7,8\,,
$$
and using 
$\V{G}_{\V{0}}(k)\cdot\valpha\,(e^{-i\V{k}\cdot\V{x}}-1)
=\valpha\cdot\wt{\V{G}}_\V{x}(k)$ 
we find 
\begin{align}
\Phi(k)&:=\nonumber
i\big(1+|\V{k}|\,\cR_\V{k}\big)\,\V{G}_{\V{0}}(k)\cdot\V{x}\,\phi-
\cR_\V{k}\,\valpha\cdot\wt{\V{G}}_\V{x}(k)\,e^{-F}\,(e^F\SA\,\phi)
\\\nonumber
&\quad+
\SA\,|\V{k}|\,\cR_\V{k}\,
\int_\RR{T}_2(y,k)\,\frac{dy}{2\pi}\,e^F\phi
+
\SA\,\{\Hf\,\cR_\V{k}\}^*\,
\int_\RR{T}_3(y,k)\,\frac{dy}{2\pi}\,e^F\phi
\\\nonumber
&\quad+
\{\Theta\,\cR_\V{k}\}^*
\int_\RR{T}_4^*(y,k)\,\frac{dy}{2\pi}\,e^F\SA\,\phi
\\\nonumber
&\quad+\SA\,\{|\DA|^\nu\,\cR_\V{k}\}^*\,
\int_\RR{T}_5(y,k)\,\frac{dy}{\pi}\,e^F\phi
\\\nonumber
&\quad
-\frac{\gamma}{2}\,\SA\,\{|\V{x}|^{-\nu}\,\cR_\V{k}\}^*
\int_\RR{T}_{7}(y,k)\,\frac{dy}{\pi}\,\Theta^{1/2}e^F\phi
\\
&\quad
-\frac{\gamma}{2}\,\{\Theta^{1/2}\,\cR_\V{k}\}^*
\int_\RR{T}_{8}^*(y,k)\,\frac{dy}{\pi}
\,\big\{|\V{x}|^{-\nu}e^F\SA\,{\phi}\big\}
\,.\label{def-Phi(k)}
\end{align}
As $g\in C_0^\infty((\RR^3\setminus\{0\})\times\ZZ_2,\CC)$ 
is arbitrary in \eqref{ulf99b}
we have
$\SPn{\eta}{a(k)\,\phi}=\SPb{\eta}{\Phi(k)}$,
for all $k$ 
outside some set of measure zero, $N_\eta$, which depends on $\eta$. 
Choosing $\eta$ from a countable dense subset $\sX\subset\sH$
we conclude that
$a(k)\,\phi=\Phi(k)$, for all $k\notin N$, where
$N:=\bigcup_{\eta\in\sX}N_\eta$ has zero measure.
Thus, we have proved the following lemma:

\begin{lemma}\label{le-a(k)phi}
Let $e\in\RR$, $\UV>0$, $m\grg0$, $\gamma\in(0,\gcnp)$, and 
$\phi_m\in\dom(H_{\gamma,m})$ with $H_{\gamma,m}\,\phi_m=E_{\gamma,m}\,\phi_m$.
Then $a(k)\,\phi_m=\Phi(k)$, for almost every $k\in \RR^3\times\ZZ_2$, 
where $\Phi(k)$ is defined in \eqref{def-Phi(k)}.
\end{lemma}


\subsection{ Derivation of the infra-red bounds}\label{ssec-IR-proof}

\noindent
In the following proof we again use the notation of the 
previous subsection. We drop the reference to $e$, $\UV$,
and $\gamma$ in the notation, but re-introduce a subscript $m$
when it becomes important.

\smallskip

\begin{proof}[Proof of Proposition~\ref{prop-IR}]
The soft photon bound \eqref{eq-spb} follows 
by combining
Lemma~\ref{le-a(k)phi} with the bounds 
\eqref{ulf101}--\eqref{exp-loc-phim-Theta}, 
$|\V{G}_{m,\V{0}}(k)|\klg C\,|\V{k}|^{-\frac{1}{2}}\id_{\{m\klg|\V{k}|\klg\UV\}}$,
and 
$|\wt{\V{G}}_{m,\V{x}}(k)|\,e^{-F(\V{x})}
\klg C'\,|\V{k}|^{1/2}\id_{\{m\klg|\V{k}|\klg\UV\}}$, as well as 
$$
\Big\|\,|\V{k}|\int{T}_2(y,k)\,\frac{dy}{\pi}\Big\|\klg |\V{k}|
\,|\V{G}_{m,\V{0}}(k)|\,\int_\RR\frac{C\,dy}{1+y^2}\,
\klg C'\,|\V{k}|^{1/2}\id_{\{m\klg|\V{k}|\klg\UV\}}\,,
$$
and, for $i=3,4,5,7,8$,
\begin{align*}
\Big\|\int{T}_i(y,k)\,\frac{dy}{\pi}\Big\|&\klg 
\sup_{\V{x}}\big\{|\wt{\V{G}}_{m,\V{x}}(k)|\,e^{-F(\V{x})}\big\}
\int_\RR\frac{C(d)\,dy}{(1+y^2)^{1-\kappa/2}}
\\
&\klg C'(d)\,(|\V{k}|\wedge 1)^{1/2}\,\id_{\{m\klg|\V{k}|\klg\UV\}}\,.
\end{align*}
Recall that $T_i(y,k)$, $i=3,4,5,7,8$, is defined by
the same formula as $\V{T}_i(y,\V{k})$ except that
$\valpha\,(e^{-i\V{k}\cdot\V{x}}-1)$ is replaced by
$\valpha\cdot\wt{\V{G}}_{m,\V{x}}(k)$.

In order to prove the photon derivative bound
\eqref{pdb-kp} we again use the representation
$a(\V{k},\lambda)\,\phi_m-a(\V{p},\lambda)\,\phi_m
=\Phi(\V{k},\lambda)-\Phi(\V{p},\lambda)$. 
Moreover, we apply the following bounds: 
First, by the resolvent identity,
$$
\big\|\cO\,\cR_\V{k}-\cO\,\cR_\V{p}\big\|
\klg \frac{C\,|\V{k}-\V{p}|}{(1\wedge|\V{k}|)(1\wedge|\V{p}|)},\quad
\cO\in\big\{\id,\,\Hf,\,\Theta^{1/2},\,|\DAm|^\nu,\,|\V{x}|^{-\nu}\big\}\,.
$$
Second, we have, for $m<|\V{k}|<\UV$ and $m<|\V{p}|<\UV$,
\begin{align*}
\Big\|\int\big(|\V{k}|{T}_2(y,\lambda,\V{k})
-|\V{p}|{T}_2(y,\lambda,\V{p})\big)
\frac{dy}{\pi}\Big\|\klg 
{\triangle}'(\V{k},\V{p})
\int_\RR\frac{C\,dy}{1+y^2}=C'{\triangle}'(\V{k},\V{p}),
\end{align*}
where
$$
{\triangle}'(\V{k},\V{p}):=\max_{\lambda=0,1}\sup_{\V{x}}
\big|\,|\V{k}|\,{\V{G}}_{m,\V{x}}(\lambda,\V{k})
-|\V{p}|\,{\V{G}}_{m,\V{x}}(\lambda,\V{p})\big|\,e^{-F(\V{x})}.
$$
In view of \eqref{ulf2}--\eqref{ulf4}
we further have, for $i=3,4,5,7,8$, and $m<|\V{k}|<\UV$ and $m<|\V{p}|<\UV$,
$$
\Big\|\int\!\!\big({T}_i(y,\lambda,\V{k})-{T}_i(y,\lambda,\V{p})\big)
\frac{dy}{\pi}\Big\|\klg 
{\triangle}''(\V{k},\V{p})\!\int\!\frac{C(d)\,dy}{(1+y^2)^{1-\frac{\kappa}{2}}}
=C'(d){\triangle}''(\V{k},\V{p}),
$$
where, again for $m<|\V{k}|,|\V{p}|<\UV$, 
\begin{align*}
\triangle''(\V{k},\V{p})&:=\max_{\lambda=0,1}\sup_{\V{x}}
|\wt{\V{G}}_{m,\V{x}}(\lambda,\V{k})
-\wt{\V{G}}_{m,\V{x}}(\lambda,\V{p})|\,e^{-F(\V{x})}.
\end{align*}
To obtain \eqref{pdb-kp}
it now suffices to recall the following bound from \cite{GLL2001}
(see also \cite[Appendix~A]{KMS2009a}):
For $m<|\V{k}|,|\V{p}|<\UV$,
\begin{align}\nonumber
\frac{1}{|\V{k}|}\,&\big\{
\big|\,|\V{k}|\,\V{G}_{m,\V{0}}(\lambda,\V{k})
-|\V{p}|\,\V{G}_{m,\V{0}}(\lambda,\V{p})\big|
+{\triangle}'(\V{k},\V{p})+{\triangle}''(\V{k},\V{p})\big\}
\\
&\quad\nonumber
+\frac{|\V{k}-\V{p}|}{|\V{k}|\,|\V{p}|}\,
\big\{\big|\,|\V{p}|\,\V{G}_{m,\V{0}}(\lambda,\V{p})\big|
+|\wt{\V{G}}_{m,\V{x}}(\lambda,\V{p})|\big\}
\\
&\klg\,\nonumber
C\,|\V{k}-\V{p}|\,
\Big(\frac{1}{|\V{k}|^{1/2}|\V{k}_\bot|}\,+\,
\frac{1}{|\V{p}|^{1/2}|\V{p}_\bot|}\Big)\,.
\end{align}
Here the special form
of the polarization vectors \eqref{pol-vec} is exploited.
Notice also that some weight like $e^{-F(\V{x})}$ is required
in ${\triangle}'(\V{k},\V{p})$ and ${\triangle}''(\V{k},\V{p})$
since the RHS of 
$|e^{-i\V{k}\cdot\V{x}}-e^{-i\V{p}\cdot\V{x}}|\klg|\V{k}-\V{p}|\,|\V{x}|$
in unbounded w.r.t. $\V{x}$.
\end{proof}

\smallskip

\noindent
Finally, an inspection of the above proof and the preceding
subsection readily shows that the assertion of
Remark~\ref{rem-IR-gc} holds true.
In fact, all $\gamma$-dependent contributions to the
constant on the RHS of the infra-red bounds
stem from \eqref{exp-loc-phim-V} and
\eqref{exp-loc-phim-Theta}. 
If, however, the uniform bound \eqref{exp-gc} is valid, then
the RHS of \eqref{exp-loc-phim-V} and
\eqref{exp-loc-phim-Theta} can be replaced 
by $\gamma$-independent constants.
Notice also that $E\equiv E_{\gamma,m}$ in \eqref{exp-loc-phim-Theta} satisfies
$E_{\gcnp,m}\klg E\klg\Th(\omega,\V{G}_m^{e,\UV})\klg C(e,\UV)$.


\appendix

\section{Estimates on functions of the Dirac operator}
\label{app-conv}

\noindent
In this appendix we derive some technical estimates 
we have repeatedly referred to in the main text.
To this end we always assume that $\vo$ and $\V{G}$
fulfill Hypothesis~\ref{hyp-G} with constant $d$
and that $\wt{\V{G}}$ is another coupling function
such that $\vo$ and $\wt{\V{G}}$
fulfill Hypothesis~\ref{hyp-G} with the same constant $d$.
Furthermore, we introduce the parameter
\begin{equation*}
\triangle(a):=\int\Big(1+\frac{1}{\vo(k)}\Big)
\,\sup_{\V{x}\in\RR^3}e^{-2a|\V{x}|}
\big|\V{G}_\V{x}(k)-\wt{\V{G}}_\V{x}(k)\big|^2\,dk\,,
\quad a\grg0\,.
\end{equation*}
We define $\V{A}$ as usual by \eqref{def-Aphys}
and $\wt{\V{A}}$ by \eqref{def-Aphys} with $\wt{\V{G}}$ 
instead of $\V{G}$.

First, we collect some necessary prerequisites:
We recall that, for $y\in \RR$, $a\in[0, 1/2]$, and 
$F\in C^\infty(\RR_{\V{x}}^3,\RR)$ with fixed sign 
and satisfying $|\nabla F|\klg a$, we have 
$iy\in\vr(D_{\V{A}}+i\valpha\cdot\nabla F)$,
\begin{align}\label{BoundRF}
R^F_{\V{A}}(iy)&:= e^{F}\,R_{\V{A}}(iy)\,e^{-F}
=(\DA+i\valpha\cdot\nabla F-iy)^{-1}\quad \textrm{on}\;\dom(e^{-F})\,,
\\
\|R^F_{\V{A}}(iy)\|&\klg C\,(1+y^2)^{-1/2}.\label{BoundRF2}
\end{align}
The bound \eqref{BoundRF2} is essentially well-known. 
For instance, its
proof given in \cite{MatteStockmeyer2008b} for classical vector
potentials works for quantized ones as well. 
Next, we set
\begin{equation}\label{def-HT}
\HT:=d\Gamma(\vo)+K\,,\qquad
Z_{\nu,\delta}:=
\HT^{\delta}\,[\HT^{-\nu},\valpha\cdot\V{A}]\,\HT^{\nu-\delta},
\end{equation}
and recall from \cite[Lemma~3.1]{MatteStockmeyer2009a}
and \cite[Lemma~3.2]{Matte2009} that 
\begin{equation}\label{lisa1}
\|Z_{\nu,\delta}\|\klg C(d)/K^{1/2},
\quad K\grg1\,,\;\nu,\delta\in[-1,1]\,.
\end{equation}
Thus, if $K$ is chosen sufficiently large, depending only on $d$, then
the following Neumann series converges,
for all $y\in \RR$, $\nu\in[-1,1]$, and $a$, $F$ as above,
\begin{align}\label{Com0a}
\Upsilon_{\nu}^F(y):=\sum_{\ell=0}^\infty\{-Z_{\nu,0}^*\,\RAF{iy}\}^\ell,
\qquad\textrm{and, say,}\qquad
\|\Upsilon_{\nu}^F(y)\|\klg2\,.
\end{align}
It is straightforward (compare \cite[Corollary~3.1]{MatteStockmeyer2009a}
and \cite[Lemma~3.3]{Matte2009} for negative $\nu$)
to verify that
\begin{align}
\label{Com1b}
R^F_{\V{A}}(iy)\,\HT^{-\nu}&=\HT^{-\nu}\,R^F_{\V{A}}(iy)\,\Upsilon_{\nu}^F(y)\,.
\end{align}
In particular,
$R^F_{\V{A}}(iy)$ maps 
$\dom(d\Gamma(\vo)^{\nu})$ into itself.
If $\V{A}$ is replaced by $\wt{\V{A}}$ in \eqref{def-HT}--\eqref{Com1b},
then we denote the corresponding operator as
$\wt{\Upsilon}_{\nu}^F(y)$.

\begin{lemma}\label{le-faysal}
Let $\mu,\nu\in[-1,1]$ with $\mu\wedge\nu\klg-1/2$ and $\mu+\nu\klg-1/2$, 
and $\kappa\in[0,1)$.
Assume that $a\in[0, 1/2]$ and 
$F\in C^\infty(\RR_{\V{x}}^3,[0,\infty))$
satisfies $|\nabla F|\klg a$ and $F(\V{x})\grg a|\V{x}|$,
for all $\V{x}\in\RR^3$.
Then 
\begin{align}\label{faysal-app1}
\big\|\,|\DA|^\kappa\,\HT^\mu\,
(\SA-\SAt)\,\HT^{\nu}\,e^{-F}\,\big\|&\klg C(d,\kappa)\,\triangle^{1/2}(a)\,,
\\\label{faysal-app2}
\big\|\,|\DA|^\kappa\,\HT^\mu\,e^{-F}\,
(\SA-\SAt)\,\HT^{\nu}\,\big\|&\klg C(d,\kappa)\,\triangle^{1/2}(a)\,,
\\\label{faysal-app3}
\big\|\,|\V{x}|^{-\kappa}e^{-F}\,
(\SA-\SAt)\,\HT^{-1}\big\|&\klg C(d,\kappa)\,\triangle^{1/2}(a)\,,
\\\label{faysal-app4}
\big\|\,|\V{x}|^{-\kappa}\HT^{\kappa-1}\,
(\SA-\SAt)\,\HT^{-\kappa}e^{-F}\,\big\|&\klg C(d,\kappa)\,\triangle^{1/2}(a)\,,
\;\;\kappa\in(1/2,1)\,.
\end{align}
\end{lemma}

\begin{proof}
It is easy to verify the following resolvent formula,
\begin{equation}\label{faysal-app5}
\HT^\mu\,(\RA{iy}-\RAt{iy})\,\HT^{\nu}\,e^{-F}
=\RA{iy}\,{\Upsilon}_{\mu}^0(y)\,
T_{\mu,\nu}^F\,R_{\wt{\V{A}}}^F(iy)\,\wt{\Upsilon}_{-\nu}^F(y)\,,
\end{equation}
where $T_{\mu,\nu}^F\in\LO(\HR_m)$ is the closure of 
${\HT^\mu\,\valpha\cdot(\wt{\V{A}}-\V{A})\,e^{-F}\,\HT^\nu}$
and satisfies $\|T_{\mu,\nu}^F\|\klg C(d)\,\triangle^{1/2}(a)$
by a standard estimate and an interpolation argument
or \eqref{lisa1}.
Likewise, we have
\begin{equation}\label{faysal-app6}
e^{-F}\HT^\mu\,(\RA{iy}-\RAt{iy})\,\HT^{\nu}
=R_{\V{A}}^{-F}(iy)\,{\Upsilon}_{\mu}^{-F}(y)\,
T_{\mu,\nu}^F\,\RAt{iy}\,\wt{\Upsilon}_{-\nu}^0(y)\,.
\end{equation}
Applying \eqref{sgn} and \eqref{faysal-app5} we obtain,
for all $\vp\in\dom(\HT^\mu\,|\DA|^\kappa)$
and $\psi\in\HR_m$,
\begin{align*}
\big|\SPb{&\HT^{\mu}\,|\DA|^\kappa\,\vp}{(\SA-\SAt)\,\HT^{\nu}e^{-F}\,\psi}\big|
\\
&\klg
\int_\RR
\big|\SPb{\vp}{|\DA|^\kappa\,\RA{iy}\,{\Upsilon}_{\mu}^0(y)\,
T_{\mu,\nu}^F\,R_{\wt{\V{A}}}^F(iy)\,\wt{\Upsilon}_{-\nu}^F(y)
\,\psi}\big|\,\frac{dy}{\pi}
\\
&\klg
C'(d)\,\triangle^{1/2}(a)\int_\RR\frac{dy}{(1+y^2)^{1-\kappa/2}}\cdot\|\vp\|
\,\|\psi\|\,.
\end{align*}
Here we also used \eqref{BoundRF2},
$\|\,|\DA|^\kappa\,\RA{iy}\|\klg C(\kappa)\,(1+y^2)^{(\kappa-1)/2}$,
and \eqref{Com0a}.
This shows that $\SA-\SAt$ maps $\Ran(e^{-F}\otimes\HT^\nu)$ into
$\dom(\{\HT^\mu\,|\DA|^\kappa\}^*)$.
Since $|\DA|^\kappa$ is continuously invertible
we have $\{\HT^\mu\,|\DA|^\kappa\}^*=|\DA|^\kappa\,\HT^\mu$ 
and we obtain \eqref{faysal-app1}.
\eqref{faysal-app2} is proved in the same way using
\eqref{faysal-app6} and the bound
$\|\,|\DA|^\kappa\,R_{\V{A}}^{-F}(iy)\|\klg C'(\kappa)\,(1+y^2)^{(\kappa-1)/2}$,
which is an easy consequence of \eqref{BoundRF2} and the second resolvent
identity. 
In order to derive \eqref{faysal-app3} we employ the
identity
$$
e^{-F}\,(\RA{iy}-\RAt{iy})\,\HT^{-1}=
R_\V{A}^{-F}(iy)\,\HT^{-1/2}\,T_{1/2,-1}^F\,\RAt{iy}\,\wt{\Upsilon}^0_1(y)\,,
$$
which together with the
generalized Hardy inequality and \eqref{tina3} below
yields
\begin{align*}
&\big|\SPb{|\V{x}|^{-\kappa}\,\vp}{e^{-F}\,(\SA-\SAt)\,\HT^{-1}\psi}\big|
\\
&\klg C(\kappa)
\int_\RR\|\vp\|\,\big\|\,|\DO|^\kappa\,R_\V{A}^{-F}(iy)\,\HT^{-1/2}\,\big\|
\,\|T_{1/2,-1}^F\|\,\|\RAt{iy}\|\,\|\wt{\Upsilon}^0_1(y)\|\,\|\psi\|
\,\frac{dy}{\pi}
\\
&\klg
C(d,\kappa)\,\triangle^{1/2}(a)\int_\RR\frac{dy}{(1+y^2)^{1-\kappa/2}}\cdot\|\vp\|
\,\|\psi\|\,,
\end{align*}
for all $\vp\in\dom(|\V{x}|^{-\kappa})$ and $\psi\in\HR_m$.
In a similar fashion we prove 
\eqref{faysal-app4} by means of \eqref{tina3} and the identity
\begin{align*}
\HT^{\kappa-1}\,(\RA{iy}-\RAt{iy})\,\HT^{-\kappa}e^{-F}
&=\HT^{\kappa-1}\RA{iy}\,\HT^{-\delta}\, T^F_{\delta,-\kappa}
\,R_{\wt{\V{A}}}^{F}(iy)\,
\wt{\Upsilon}^F_\kappa(y)\,,
\end{align*}
where $\delta:=\kappa-1/2$, so that, with $\nu:=1-\kappa=1/2-\delta$,
\begin{align}\nonumber
\HT^{\kappa-1}&\RA{iy}\,\HT^{-\delta}
=\RA{iy}\,\HT^{-1/2}+\big[\HT^{-\nu},\RA{iy}\big]\,\HT^{-\delta}
\\\nonumber
&=\RA{iy}\,\HT^{-1/2}+
\RA{iy}\,\big[\DA,\HT^{-\nu}\big]\,\RA{iy}\,\HT^{-\delta}
\\\nonumber
&=\RA{iy}\,\HT^{-1/2}
\big\{\id+\HT^{1/2}\big[\valpha\cdot\V{A},\HT^{-\nu}\big]
\,\HT^{-\delta}\,\big(\HT^\delta\RA{iy}\HT^{-\delta}\big)\big\}
\\\label{lisel}
&=\RA{iy}\,\HT^{-1/2}\big\{\id-
\ol{Z}_{\nu,1/2}\,\RA{iy}\,\Upsilon^0_\delta(y)\big\}\,.
\end{align}
Here the operator $X:=\{\cdots\}$ is bounded due to \eqref{lisa1}
and \eqref{Com0a}.
Therefore, the computation \eqref{lisel}, which is justified
at least on the domain $\core_m$, shows that
$B:=\HT^{\kappa-1}\RA{iy}\,\HT^{-\delta}$ maps $\HR_m$ into the domain
of $|\DO|^\kappa$ and 
$\|\,|\DO|^\kappa B\|
=\|\,|\DO|^\kappa\RA{iy}\,\HT^{-1/2}{X}\|\klg
C(d,\kappa)\,(1+y^2)^{(\kappa-1)/2}$, by \eqref{tina3}.
\end{proof}

\begin{lemma}\label{le-tina2}
Let $a\in[0, 1/2]$ and 
$F\in C^\infty(\RR_{\V{x}}^3,\RR)$ have a fixed sign 
with $|\nabla F|\klg a$. Then 
$\RAF{iy}$ maps $\dom(d\Gamma(\vo)^{1/2})$ into
$\dom(|\DO|^\kappa)$, $\kappa\in[0,1]$, and
\begin{equation}\label{tina3}
\big\|\,|\DO|^\kappa\,\RAF{iy}\,(d\Gamma(\vo)+1)^{-1/2}\,\big\|
\klg C(d,\kappa)\,(1+y^2)^{(\kappa-1)/2},
\quad y\in\RR\,.
\end{equation}
\end{lemma}

\begin{proof}
Using
$\|\,|\DO|^\kappa\,\RO{iy}\|\klg f_\kappa(y):=C(\kappa)\,(1+y^2)^{(\kappa-1)/2}$,
$\|\valpha\cdot\nabla F\|\klg a\klg1/2$,
and the notation \eqref{def-HT},
we obtain, for all $\vp\in\core_m$,
\begin{align}\nonumber
\big\|\,|\DO|^\kappa\,\vp\big\|&\klg f_\kappa(y)\big\|(\DO+iy)\,\vp\big\|
\\\label{tina88}
&\klg
f_\kappa(y)\,\big(\|(\DA+i\valpha\cdot\nabla F+iy)\,\vp\|
+C(d)\,\|\HT^{1/2}\,\vp\|\big)\,.
\end{align}
Now, let $\eta\in\HR_m$ and put 
$\psi:=\RAF{iy}\,\Upsilon_{1/2}^F(y)\,\eta$,
so that $\psi\in\dom(\DA)$. Since $\DA$ is essentially self-adjoint
on $\core_m$ we find $\psi_n\in\core_m$, $n\in\NN$, such that
$\psi_n\to\psi$ in the graph norm of $\DA$.
By \eqref{lisa1} $[\DA,\HT^{-1/2}]$, defined on $\core_m$,
extends to a bounded operator on $\HR_m$.
We deduce that 
$\vp_n:=\HT^{-1/2}\psi_n\to\HT^{-1/2}\psi=\RAF{iy}\,\HT^{-1/2}\eta$
in the graph norm of $\DA$, too.
Thus, we may plug $\vp_n\in\core_m$ into \eqref{tina88},
pass to the limit $n\to\infty$, and apply \eqref{Com0a}
to arrive at 
$\|\,|\DO|^\kappa\,\RAF{iy}\,\HT^{-1/2}\eta\|
\klg C'(d)\,f_\kappa(y)\,\|\eta\|$.
\end{proof}

\begin{lemma}\label{le-tina1}
Let $\kappa\in(0,1)$ and $\delta>0$. Then 
\begin{align}\nonumber
|\DA|^{2\kappa}&\klg2^{\kappa}\,|\DO|^{2\kappa}+C(d,\kappa)
\,(d\Gamma(\vo)+1)^{\kappa}
\\
&\klg2^{\kappa}\,|\DO|^{2\kappa}+\delta\,d\Gamma(\vo)
+C(d,\kappa,\delta)\,.\label{tina2}
\end{align}
\end{lemma}

\begin{proof}
We put $\Theta:=d\Gamma(\vo)+1$.
For every $\vp\in\core_m$, we have
$\|\,|\DA|\,\vp\|=\|(\DO+\valpha\cdot\V{A})\,\vp\|
\klg\|\,|\DO|\,\vp\|+C(d)\,\|\Theta^{1/2}\,\vp\|$,
thus $|\DA|^2\klg2|\DO|^2+2C(d)\,\Theta$ on 
$\dom(\DO)\cap\dom(d\Gamma(\vo)^{1/2})$.
For $\kappa\in(0,1)$, the map $t\mapsto t^{\kappa}$
is operator monotone.
Together with $(a+b)^{\kappa}\klg a^{\kappa}+b^{\kappa}$, $a,b\grg0$, 
this implies \eqref{tina2}.
\end{proof}


\section{Some properties of ground state eigenvectors}
\label{app-form}

\noindent
In this appendix we always assume that
$e\in\RR$, $\UV>0$, $m\grg0$, 
$\gamma\in(0,\gcnp)$, and that $\phi_m$ is
a ground state eigenvector of the operator $\FNPm{\gamma}$
defined in \eqref{def-Hm}, so that
$\FNPm{\gamma}\,\phi_m=E_{\gamma,m}\,\phi_m$.
Our aim is to show that,
for every  $f\in\HP_0$ with $\omega^{-1/2}f\in\HP_0$,
the vector $a(f)\,\wt{\phi}_m=a(f)\,U\,\phi_m$ belongs to the
form domain of the unitarily transformed 
operator $\wt{H}_{\gamma,m}=U\,\FNPm{\gamma}\,U^*$ defined in
Subsection~\ref{ssec-gauge}. This result has been
used in order to derive the infra-red bounds.

\begin{lemma}\label{prop-a(f)00}
$\Hf^{1/2}\phi_m\in\form(\FNPm{\gamma})$.
\end{lemma}

\begin{proof}
We set
$\HT:=\Hf+K$ and
$$
f_\ve(t)\,:=(t+K)/(1+\ve\,t+\ve\,K)\,,\quad
t\grg0\,,\qquad F_\ve\,:=\,f_\ve^{1/2}(\Hf)\,,
$$
for all $\ve>0$ and some $K\grg1$.
Moreover, we put
$Y:=\FNPm{\gamma}-E_{\gamma,m}+1$.
In \cite[Proof of Theorem~6.1]{Matte2009} one of the present
authors proved that, for every sufficiently large value of $K$
and for all $\vp_1,\vp_2\in\core_0$ and $\ve>0$,
\begin{align}
\big|\SPn{Y\,\vp_1&}{F_\ve\,\vp_2}-
\SPn{F_\ve\,\vp_1}{Y\,\vp_2}\big|
\klg C\,\big(\SPn{\vp_1}{Y\,\vp_1}+\SPn{\vp_2}{Y\,\vp_2}\big)
\,.\label{ex-np-bd-(c)}
\end{align}
Moreover, it follows from \cite[Proof of Theorem~6.1]{Matte2009}
that $F_\ve$ maps the form domain of $\FNPm{\gamma}$
continuously (with respect to the form norm)
into itself. In particular, the inequality
\eqref{ex-np-bd-(c)} extends to all $\vp_1,\vp_2\in\form(\FNPm{\gamma})$.
Using \eqref{ex-np-bd-(c)} with
$\vp_1=(2C)^{-1/2}\,F_{\ve}\,\phi_m\in\form(\FNPm{\gamma})$ and
$\vp_2=(2C)^{1/2}\,\phi_m$ we then obtain
\begin{align}\nonumber
\SPb{F_{\ve}&\phi_m}{YF_{\ve}\phi_m}
\\
&\klg\nonumber
\big|\SPb{F_{\ve}^{2}\phi_m}{Y\phi_m}\big|
+\frac{1}{2}\SPb{F_{\ve}\phi_m}{YF_{\ve}\phi_m}
+2C^2\SPb{\phi_m}{Y\phi_m},
\end{align}
for all $\ve>0$.
Since $Y\phi_m=\phi_m$ and
$\|F_{\ve}\,\phi_m\|\nearrow\|\HT^{1/2}\,\phi_m\|$,
as $\ve\searrow0$,  
because of $\phi_m\in\dom(\Hf^{1/2})$,
we obtain, for $\ve>0$,
\begin{align*}
\SPb{F_{\ve}\,\phi_m}{Y\,F_{\ve}\,\phi_m}
&\klg2\,\|\HT^{1/2}\,\phi_m\|^2+4C^2\,\|\phi_m\|^2=:\,B\,.
\end{align*}
In particular, the densely defined functional
$u(\eta):=\SPn{\HT^{1/2}\,\phi_m}{Y^{1/2}\,\eta}$,
$\eta\in\form(\FNPm{\gamma})$, is bounded,
$
|u(\eta)|=\lim_{\ve\searrow0}
|\SPn{Y^{1/2}\,F_\ve\,\phi_m}{\eta}|
\klg B^{1/2}\,\|\eta\|
$,
whence
$\HT^{1/2}\,\phi_m\in\dom(Y^{1/2})=\form(\FNPm{\gamma})$
since $Y^{1/2}$ is self-adjoint.
\end{proof}

\begin{lemma}\label{prop-a(f)0}
$a(f)\,\phi_m\in\form(\FNPm{\gamma})$,
for all $f\in\HP_0$ with 
$\omega^{-1/2}f\in\HP_0$.
\end{lemma}

\begin{proof}
By Theorem~\ref{thm-Friedrichs-ext}
$\form(\FNPm{\gamma})
=\dom(|\DO|^{1/2})\cap\dom(\Hf^{1/2})$.
It well-known and not difficult to show that
$a(f)$ maps $\dom(\Hf)$ into $\dom(\Hf^{1/2})$.
Thus, $a(f)\,\phi_m\in\dom(\Hf^{1/2})$ 
as we know from Theorem~\ref{thm-rb}(iii)
that $\phi_m\in\dom(\Hf)$.
Moreover, by Lemma~\ref{prop-a(f)00},
$\Hf^{1/2}\phi_m\in\form(\FNPm{\gamma})\subset\dom(|\DO|^{1/2})$.
By a simple standard estimate we can check directly that
$a(f)\,\psi\in\dom(|\DO|^{1/2})$, for every $\psi\in\dom(\Hf^{1/2})$
with
$\Hf^{1/2}\psi\in\dom(|\DO|^{1/2})$.
In particular, $a(f)\,\phi_m\in\dom(|\DO|^{1/2})$.
\end{proof}

\begin{lemma}\label{prop-a(f)}
$a(f)\,\wt{\phi}_m\in\form(\wt{H}_{\gamma,m})$,
for all $f\in\HP_0$ with 
$\omega^{-1/2}\,f\in\HP_0$.
\end{lemma}

\begin{proof}
On account of \eqref{dora2}
and $U\,a(f)\,\phi_m\in\form(\wt{H}_{\gamma,m})$
(by Lemma~\ref{prop-a(f)0}) it remains to show that 
$U\,x_j\,\phi_m\in\form(\wt{H}_{\gamma,m})$, or equivalently,
$x_j\,\phi_m\in\form(\FNPm{\gamma})=\dom(|\DO|^{1/2})\cap\dom(\Hf^{1/2})$, 
for every
component $x_j$, $j\in\{1,2,3\}$, of $\V{x}$.
To this end we recall the following bounds from
\cite[Lemma~5.4]{MatteStockmeyer2009a}: 
Let $a\in[0,1/2]$ and let
$\wt{F}\in C^\infty(\RR^3_\V{x},[0,\infty))\cap L^\infty$ satisfy
$|\nabla \wt{F}|\klg a$. 
Put $\HT:=\Hf+K$, for some sufficiently large $K\grg1$,
depending on $a$ and $d$, and let
$\cO$ be $\DAm$, $|\V{x}|^{-1}$, $\HT$, or $\valpha\cdot\nabla \wt{F}$.
Then we have, for all $\vp\in\core_0$ and $\ve>0$,
\begin{align}\nonumber
\big|\SPb{\vp}{e^{\wt{F}}\,\PApmm\,e^{-\wt{F}}&\,\cO\,
e^{\wt{F}}\,\PApmm\,e^{-\wt{F}}\,\vp}-
\SPb{\vp}{\PApmm\,\cO\,\PApmm\,\vp}\big|
\\
&\klg\label{herbert0}
\ve\,\SPb{\vp}{\PApmm\,|\cO|\,\PApmm\,\vp}+C\,a\,(1+\ve^{-1})\,
\SPn{\vp}{\HT\,\vp}\,.
\end{align}
In view of \eqref{def-FNP}
and the sub-criticality of $\gamma\in(0,\gcnp)$
we have the following straightforward
consequence of \eqref{herbert0},
\begin{align}\label{herbert1}
\big|\SPb{\vp}{[e^{\wt{F}}\,,\,\FNPm{\gamma}]\,e^{-\wt{F}}\,\vp}\big|
\,&\klg\,
a\,C(\gamma)\,\SPn{\vp}{\FNPm{\gamma}\,\vp}+
C(d,\gamma,a)\,\|\vp\|^2,
\end{align}
for all $\vp\in\core_0$.
Moreover, as in \cite[Lemma~5.5]{MatteStockmeyer2009a}
we can show that $e^{\wt{F}}$ maps the form domain
of $\FNPm{\gamma}$ continuously into itself.
Since $\core_0$ is a form core for $\FNPm{\gamma}$ we conclude
that $\vp$ can be replaced by $e^{\wt{F}}\,\phi_m$ in \eqref{herbert1}.
Using $\FNPm{\gamma}\,\phi_m=E_{\gamma,m}\,\phi_m$
we readily infer from \eqref{herbert1} that,
for sufficiently small $a>0$,
\begin{align}\label{herbert2}
\big\|\,Y^{1/2}\,e^{\wt{F}}\,\phi_m\,\big\|^2\,\klg\,
C\,\|e^{\wt{F}}\,\phi_m\|^2,
\end{align}
where $Y:=\FNPm{\gamma}-E_{\gamma,m}+1$.
Next, we assume that $F\in C^\infty(\RR^3_\V{x},[0,\infty))$
equals $a|\V{x}|$, for large $|\V{x}|$, and pick a suitable
monotonically increasing sequence, $\{F_n\}_{n\in\NN}$,
of bounded, smooth functions $F_n$ such that $|\nabla F_n|\klg a$.
Since $\phi_m\in\dom(e^F)$, for sufficiently small
$a>0$, it makes sense to introduce
the densely defined functional $u:\dom(Y^{1/2})\to\CC$,
$$
u(\eta)\,:=\,\SPn{e^F\,\phi_m}{Y^{1/2}\,\eta}\,,
\qquad \eta\in\dom(Y^{1/2})\,.
$$
By virtue of \eqref{herbert2} and 
$e^{F_n}:\form(\FNPm{\gamma})\to\form(\FNPm{\gamma})$
we conclude that $u$ is bounded. In fact,
$$
|u(\eta)|=\lim_{n\to\infty}|\SPn{Y^{1/2}\,e^{F_n}\phi_m}{\eta}|
\klg C^{1/2}
\lim_{n\to\infty}\|e^{F_n}\phi_m\|\,\|\eta\|=
C^{1/2}\|e^F\phi_m\|\,\|\eta\|\,,
$$
for all $\eta\in\form(\FNPm{\gamma})$.
As $Y^{1/2}$ is self-adjoint this implies
$e^F\,\phi_m\in\form(\FNPm{\gamma})\subset\dom(|\DO|^{1/2})\cap\dom(\Hf^{1/2})$.
Since multiplication with $x_j\,e^{-F}$ leaves $\dom(|\DO|^{1/2})$
invariant we arrive at
$x_j\,\phi_m\in\dom(|\DO|^{1/2})\cap\dom(\Hf^{1/2})$.
\end{proof}


\bigskip

\noindent
{\bf Acknowledgements.}
This work has been partially supported by the DFG (SFB/TR12).
O.M. and E.S. thank the Institute for Mathematical Sciences
and Center for Quantum Technologies of the National University
of Singapore, where parts of this work have been prepared, 
for their kind hospitality.


\def\cprime{$'$}

\end{document}